\pdfoutput=1 

\documentclass[acmsmall,screen,dvipsnames,x11names,svgnames]{acmart} 
\renewcommand\footnotetextcopyrightpermission[1]{}

\usepackage{etoolbox}
\newtoggle{arxiv}
\toggletrue{arxiv}

\usepackage{booktabs}   
\usepackage{subcaption} 

\usepackage{ifthen}
\usepackage{tensor}
\usepackage{adjustbox}
\usepackage{nicefrac}

\usepackage{listings}

\usepackage{color}
\usepackage{xcolor}

\usepackage{cancel}
\usepackage{csquotes}

\usepackage{xfrac}
\usepackage{xspace}

\usepackage{mathtools}
\usepackage{stmaryrd}

\usepackage{proof}
\usepackage{relsize}

\usepackage[nameinlink,noabbrev]{cleveref}

\usepackage[textwidth=1.25cm,disable]{todonotes}
\usepackage{thm-restate}

\usepackage{tikz}
\usetikzlibrary{automata}
\usetikzlibrary{positioning}

\usepackage{wrapfig}

\DeclareMathAlphabet{\mathpzc}{OT1}{pzc}{m}{it}
\DeclarePairedDelimiter{\ceil}{\lceil}{\rceil}
\DeclarePairedDelimiter{\floor}{\lfloor}{\rfloor}

\newcommand{\genericcomment}[3]{\todo[color=#1,size=\tiny,fancyline,author=#2]{#3}\xspace}

\newcommand{\kb}[1]{\todo[color=cyan!30,size=\tiny,fancyline,author=Kevin]{#1}\xspace}
\newcommand{\kbin}[1]{\todo[inline,color=cyan!30,size=\footnotesize,author=Kevin]{#1}\xspace}

\newcommand{\tw}[1]{\genericcomment{violet!35}{Tobi}{#1}\xspace}

\newcommand{\tb}[1]{\todo[color=orange!60,size=\tiny,fancyline,author=Tom]{#1}\xspace}

\newcommand{\jp}[1]{\genericcomment{green!30}{JP}{#1}\xspace}

\newcommand{\gray}[1]{\textcolor{gray}{#1}}
\newcommand{\black}[1]{\textcolor{black}{#1}}
\newcommand{\blue}[1]{\textcolor{blue}{#1}}
\newcommand{\red}[1]{\textcolor{red}{#1}}

\newcommand{\sfsymbol}[1]{\textsf{\upshape {#1}}}
\newcommand{\Vars}{\ensuremath{\mathsf{Vars}}\xspace}   
\newcommand{\Vals}{\ensuremath{\mathsf{Vals}}\xspace}    
\newcommand{\States}{\mathsf{States}}
\newcommand{\pgcl}{\textnormal{\sfsymbol{pGCL}}\xspace}   

\newcommand{\Nats}{\ensuremath{\mathbb{N}}\xspace}

\newcommand{\Rats}{\ensuremath{\mathbb{Q}}\xspace}
\newcommand{\Bools}{\ensuremath{\mathbb{B}}\xspace}
\newcommand{\true}{\ensuremath{\mathsf{true}}}
\newcommand{\false}{\ensuremath{\mathsf{false}}}

\newcommand{\PosReals}{\mathbb{R}_{\geq 0}}

\newcommand{\PosRealsInf}{\mathbb{R}_{\geq 0}^\infty}
\newcommand{\NonnegRealsInf}{\PosRealsInf}
\newcommand{\PosRats}{\mathbb{Q}_{\geq 0}}

\newcommand{\pstate}{\ensuremath{\sigma}}
\newcommand{\cc}{C}
\newcommand{\ccbody}{\cc_{\textnormal{body}}}
\newcommand{\guard}{\varphi}
\newcommand{\guardb}{\psi}
\newcommand{\aexpr}{E}
\newcommand{\pexpr}{p}
\newcommand{\pnum}{q}
\newcommand{\val}{v}
\newcommand{\refines}{\multimap}
\newcommand{\entails}{\models}

\newcommand{\progclass}{\Pi}

\newcommand{\SKIP}{\mathtt{skip}}

\newcommand{\AssignSymbol}{\mathrel{\textnormal{\texttt{:=}}}}
\newcommand{\ASSIGN}[2]{\ensuremath{#1 \AssignSymbol #2}}
\newcommand{\UNIFASSIGN}[2]{\ensuremath{\ASSIGN{#1}{\textnormal{\texttt{unif}} \left(#2\right)}}}

\newcommand{\COMPSEMIC}{\fatsemi}
\newcommand{\COMPOSE}[2]{\ensuremath{{#1}{\,\COMPSEMIC}~ {#2}}}
\newcommand{\PCHOICE}[3]{\ensuremath{\{\, {#1} \,\}\mathrel{[{#2}]}\{\, {#3} \,\}}}

\newcommand{\GCARROW}{\shortrightarrow}
\newcommand{\IFSYMBOL}{\ensuremath{\textnormal{\texttt{if}}}}
\newcommand{\ENDIFSYMBOL}{}
\newcommand{\GC}[4]{\ensuremath{\IFSYMBOL\,{#1}\GCARROW\{{#2}\}~{}\square{}~{#3}\GCARROW\{{#4}\}\ENDIFSYMBOL}} 
\newcommand{\GCFST}[2]{\ensuremath{\IFSYMBOL\,{#1}\GCARROW\{{#2}\}}}
\newcommand{\GCFSTOPEN}[1]{\ensuremath{\IFSYMBOL\,{#1}\GCARROW\{ }}
\newcommand{\GCFSTCLOSE}{\ensuremath{\}}}
\newcommand{\GCSEC}[2]{\ensuremath{{~{}\square{}~}\hspace{0.015cm}{#1}\GCARROW\{{#2}\}\,\ENDIFSYMBOL}}
\newcommand{\GCSECOPEN}[1]{\ensuremath{{~{}\square{}~}\hspace{0.015cm}{#1}\GCARROW\{ }}
\newcommand{\GCSECCLOSE}{\ensuremath{\} \ENDIFSYMBOL}}

\newcommand{\ELSESYMBOL}{\ensuremath{\textnormal{\texttt{else}}}}

\newcommand{\ITE}[3]{\ensuremath{\IFSYMBOL\,\left(\, {#1} \,\right)\,\left\{\, {#2} \,\right\}\,\ELSESYMBOL\,\left\{\, {#3} \,\right\}}}

\newcommand{\WHILESYMBOL}{\ensuremath{\textnormal{\texttt{while}}}}
\newcommand{\ENDWHILESYMBOL}{}

\newcommand{\WHILE}[1]{\WHILESYMBOL\,{#1}\GCARROW\{ }
\newcommand{\WHILECLOSE}{\}}

\newcommand{\WHILEDO}[2]{\WHILE{#1} {#2} \WHILECLOSE\ENDWHILESYMBOL}

\newcommand{\WHILEDOINV}[3]{\WHILEDOINVOPEN{#1} {#2} \WHILEDOINVCLOSE{#3}}
\newcommand{\WHILEDOINVOPEN}[1]{\WHILESYMBOL\,{#1}\GCARROW\{}
\newcommand{\WHILEDOINVCLOSE}[1]{\}\,\langle#1\rangle\,\ENDWHILESYMBOL}

\newcommand{\ExecSymbol}{\ensuremath{\rightarrow}}
\newcommand{\Term}{\ensuremath{\Downarrow}}
\newcommand{\OpStates}{\ensuremath{\mathsf{Conf}}}
\newcommand{\Exec}[6]{\ensuremath{#1,#2 \,\xrightarrow{#3, #4}\, #5, #6 }}
\newcommand{\ExecAbbr}[4]{\ensuremath{#1 \,\xrightarrow{#2, #3}\, #4 }}
\newcommand{\actlab}{\ell}
\newcommand{\labtau}{\tau}
\newcommand{\labalpha}{\alpha}
\newcommand{\labbeta}{\beta}

\newcommand{\MDP}{\mathcal{M}}
\newcommand{\opMDP}[1]{\mathcal{M} \left( #1\right)}
\newcommand{\MS}{S}
\newcommand{\MT}{T} 
\newcommand{\MA}{A}
\newcommand{\MP}{\textsf{P}}
\newcommand{\MI}{\MS_{\textnormal{\text{init}}}}
\renewcommand{\MR}{\textnormal{\textsf{rew}}}
\newcommand{\Dist}[1]{Dist(#1)}
\newcommand{\supp}[1]{supp(#1)} 
\newcommand{\pathstotarget}[1]{\mathlarger{\mathlarger{\lozenge}}#1}
\newcommand{\probofpath}[2]{\textnormal{\textsf{Prob}}_{#1}\left( #2 \right)}
\newcommand{\exprew}[3]{\textnormal{\textsf{ExpRew}}_{#1}\llbracket #2 \rrbracket\left(#3\right)}
\newcommand{\exprewmc}[2]{\textnormal{\textsf{ExpRew}}_{}\llbracket #1 \rrbracket\left(#2\right)}
\newcommand{\minexprew}[2]{\textnormal{\textsf{MinExpRew}}_{}\llbracket #1 \rrbracket\left(#2\right)}
\newcommand{\maxexprew}[2]{\textnormal{\textsf{MaxExpRew}}_{}\llbracket #1 \rrbracket\left(#2\right)}
\newcommand{\thresh}{\rho}

\newcommand{\ms}{s}
\newcommand{\mi}{\ms_{\textnormal{\text{init}}}}
\newcommand{\ma}{a}
\newcommand{\msched}{\mathfrak{S}}
\newcommand{\MSCHED}{\mathit{Strats}}

\newcommand{\mc}{t}

\newcommand{\EXPTRANSSYMB}[1]{\mathsf{#1}}
\newcommand{\WPSYMBOL}{\EXPTRANSSYMB{wp}}

\newcommand{\DWPSYMBOL}{\EXPTRANSSYMB{dwp}}

\newcommand{\DWPRESYMBOL}{\EXPTRANSSYMB{dwp}^*}
\newcommand{\AWPSYMBOL}{\EXPTRANSSYMB{awp}}

\newcommand{\AWPRESYMBOL}{\EXPTRANSSYMB{awp}^*}
\newcommand{\VCSYMBOL}{\EXPTRANSSYMB{vc}} 
\newcommand{\VCSYMBOLPARA}[2]{\VCSYMBOL_{#1}^{#2}}
\newcommand{\SOMEWPSYMBOL}{\mathcal{T}}
\newcommand{\SOMEWPONESYMBOL}{\mathcal{T}_1}

\newcommand{\SOMEWPRESYMBOL}{\SOMEWPSYMBOL^{*}}
\newcommand{\SOMEWPREONESYMBOL}{\SOMEWPONESYMBOL^{*}}

\newcommand{\exptransT}[2]{{#1}\llbracket{#2}\rrbracket} 
\newcommand{\exptrans}[3]{\exptransT{#1}{#2}({#3})}

\renewcommand{\wp}[2]{\exptrans{\WPSYMBOL}{#1}{#2}}

\newcommand{\dwp}[2]{\exptrans{\DWPSYMBOL}{#1}{#2}}

\newcommand{\awp}[2]{\exptrans{\AWPSYMBOL}{#1}{#2}}

\newcommand{\dwpre}[2]{\exptrans{\DWPRESYMBOL}{#1}{#2}}

\newcommand{\awpre}[2]{\exptrans{\AWPRESYMBOL}{#1}{#2}}

\newcommand{\somewp}[2]{\exptrans{\SOMEWPSYMBOL}{#1}{#2}}

\newcommand{\somewpre}[2]{\exptrans{\SOMEWPRESYMBOL}{#1}{#2}}

\newcommand{\charfunn}[5]{\tensor*[^{#1}_{\ifthenelse{\equal{#2}{}\and\equal{#3}{}}{}{\langle #2, #3 \rangle}}]{\Phi}{_{{#4}}^{{#5}}}}

\newcommand{\vccond}{\mathfrak{C}}
\newcommand{\vc}[2]{\exptrans{\VCSYMBOLPARA{\vccond}{\SOMEWPSYMBOL}}{#1}{#2}}
\newcommand{\vcdwp}[2]{\exptrans{\VCSYMBOLPARA{\vccond}{\DWPSYMBOL}}{#1}{#2}}
\newcommand{\vcawp}[2]{\exptrans{\VCSYMBOLPARA{\vccond}{\AWPSYMBOL}}{#1}{#2}}

\newcommand{\superinvcond}{\textnormal{\textsf{SUPERINV}}}
\newcommand{\subinvcond}{\textnormal{\textsf{SUBINV}}}

\newcommand{\dastsubinvcond}{\textnormal{dAST}\subinvcond}

\newcommand{\dpastsubinvcond}{\textnormal{dPAST}\subinvcond}

\newcommand{\vcsuperinvdwp}[2]{\exptrans{\VCSYMBOLPARA{\superinvcond}{\DWPSYMBOL}}{#1}{#2}}

\newcommand{\inv}{I}

\newcommand{\lfp}{\textsf{lfp}\,}

\newcommand{\expsubs}[3]{{#1}[{#2} {{}\mapsto{}} {#3}]} 
\newcommand{\statesubst}[2]{[ #1 \mapsto #2 ]}

\newcommand{\expecs}{\mathbb{E}} 
\newcommand{\predicates}{\mathbb{P}} 
\newcommand{\expclass}{\Xi}

\newcommand{\eleq}{\sqsubseteq} 
\newcommand{\egeq}{\sqsupseteq} 
\newcommand{\eeleq}{~{}\sqsubseteq{}~} 
\newcommand{\eegeq}{~{}\sqsupseteq{}~} 
\newcommand{\someexprel}{\sim}
\newcommand{\somenumrel}{\sim}

\newcommand{\emin}{\sqcap} 
\newcommand{\emax}{\sqcup} 

\newcommand{\gimpsymbol}{\mathrel{\to}}

\newcommand{\gimp}[2]{{#1} \gimpsymbol {#2}} 

\newcommand{\iverson}[1]{\ensuremath{\left[ #1 \right]}}

\newcommand{\limpl}{\mathrel{\Rightarrow}}

\newcommand{\comprel}{\operatorname{\bowtie}}
\newcommand{\TRANSSYMBOL}{\EXPTRANSSYMB{trans}}
\newcommand{\transT}[2]{\EXPTRANSSYMB{trans}_{#1}^{#2}}
\newcommand{\dtransT}[1]{\EXPTRANSSYMB{d}\EXPTRANSSYMB{trans}^{}}
\newcommand{\atransT}[1]{\EXPTRANSSYMB{a}\EXPTRANSSYMB{trans}^{}}

\newcommand{\trans}[4]{\transT{#1}{#2}(#3,#4)}
\newcommand{\dtrans}[3]{\dtransT{#1}(#2,#3)}
\newcommand{\atrans}[3]{\atransT{#1}(#2,#3)}

\newcommand{\predeleq}{\preceq}
\newcommand{\predegeq}{\succeq}

\newcommand{\qiff}{\quad\textnormal{iff}\quad}

\newcommand{\qand}{\quad\textnormal{and}\quad}
\newcommand{\qqand}{\qquad\textnormal{and}\qquad}

\newcommand{\qqimplies}{\qquad\textnormal{implies}\qquad}

\newcommand{\eeq}{~{}={}~}

\newcommand{\ddefeq}{~{}\coloneqq{}~}

\newcommand{\qeq}{\quad{}={}\quad}

\newcommand{\lleq}{~{}\leq{}~}

\newcommand{\morespace}[1]{~{}#1{}~}

\newcommand{\llor}{\morespace{\lor}}

\newcommand{\annocolor}[1]{\textcolor{DodgerBlue3}{#1}}
\newcommand{\alertannocolor}[1]{\textcolor{orange}{#1}}
\newcommand{\annotate}[1]{{\annocolor{\!\!{\fatslash}\!\!{\fatslash}~~\vphantom{G'} {#1}}}}

\newcommand{\starannotate}[1]{{\annocolor{{\talloblong}\!{\talloblong}\:\vphantom{G'} {#1}}}}
\newcommand{\phiannotate}[1]{{\annocolor{\!\!\hspace{-.15ex}{}^{\annocolor{\Phi}}\!\!\!{\fatslash}\!\!{\fatslash}~~\vphantom{G'} {#1}}}}

\newcommand{\succeqannotate}[1]{\annocolor{\!\!\hspace{-.25ex}{}^{\annocolor{{\egeq}}}{\!\!\!{\fatslash}\!\!{\fatslash}~~\vphantom{G'} {#1}}}}
\newcommand{\eqannotate}[1]{\annocolor{\!\!\hspace{-.1ex}{}^{\annocolor{{=}}}{\!\!\!{\fatslash}\!\!{\fatslash}~~\vphantom{G'} {#1}}}}
\newcommand{\wpannotate}[1]{{\annocolor{\!\!\hspace{-0.75ex}{}^{\annocolor{\text{\tiny $\WPSYMBOL$}}}\!\!\!{\fatslash}\!\!{\fatslash}~~\vphantom{G'} {#1}}}}

\newcommand{\dwpannotate}[1]{{\annocolor{\!\!\hspace{-0.75ex}{}^{\annocolor{\text{\tiny $\DWPSYMBOL$}}}\!\!\!{\fatslash}\!\!{\fatslash}~~\vphantom{G'} {#1}}}}

\newcommand{\dwpreannotate}[1]{{\annocolor{\!\!\hspace{-2.5ex}{}^{\annocolor{\text{\tiny $\DWPRESYMBOL$}}}\!\!\!{\fatslash}\!\!{\fatslash}~~\vphantom{G'} {#1}}}}

\newcommand{\dwpalertannotate}[1]{{\alertannocolor{\!\!\hspace{-0.75ex}{}^{\alertannocolor{\text{\tiny $\DWPSYMBOL$}}}\!\!\!{\fatslash}\!\!{\fatslash}~~\vphantom{G'} {#1}}}}

\newcommand{\dwprealertannotate}[1]{{\alertannocolor{\!\!\hspace{-0.75ex}{}^{\alertannocolor{\text{\tiny $\DWPRESYMBOL$}}}\!\!\!{\fatslash}\!\!{\fatslash}~~\vphantom{G'} {#1}}}}

\newcommand{\awpannotate}[1]{{\annocolor{\!\!\hspace{-0.75ex}{}^{\annocolor{\text{\tiny $\AWPSYMBOL$}}}\!\!\!{\fatslash}\!\!{\fatslash}~~\vphantom{G'} {#1}}}}

\newcommand{\singlelineannotatespace}{\hspace{0.25cm}}
\newcommand{\singlelinepreannotatespace}{\hspace{0.55cm}}

\newcommand{\IChain}{\ensuremath{\inv_\mathit{Gamb}}}
\newcommand{\CChain}{\ensuremath{\cc_\mathit{Gamb}}}
\newcommand{\IGame}{\ensuremath{\inv_\mathit{Nim}}}
\newcommand{\CNim}{\ensuremath{C_{\mathit{Nim}}}}
\newcommand{\CGame}{\CNim}
\newcommand{\CGeo}{\ensuremath{\cc_{\mathit{Gamb2}}}}
\newcommand{\IGeo}{\ensuremath{\inv_{\mathit{Gamb2}}}}
\newcommand{\clever}{\ensuremath{switch}}
\newcommand{\contestantCurtain}{\ensuremath{cc}}
\newcommand{\priceCurtain}{\ensuremath{pc}}
\newcommand{\alternativeCurtain}{\ensuremath{hc}}
\newcommand{\otherCurtain}[2]{otherCurtain(#1,#2)}
\newcommand{\randomOtherCurtain}[1]{rndOtherCurtain(#1)}
\newcommand{\varTurn}{turn}
\newcommand{\POneTurn}{\ensuremath{turn = 1}}
\newcommand{\PTwoTurn}{\ensuremath{turn = 2}}

\newcommand{\congmod}[1]{\equiv_{#1}} 
\newcommand{\ncongmod}[1]{\nequiv_{#1}} 

\AtEndPreamble{%
    \theoremstyle{acmdefinition}
    \newtheorem{remark}[theorem]{Remark}
    
}
\newcommand{\defqed}{\hfill$\triangle$} 

\setcopyright{rightsretained}
\acmDOI{10.1145/3632935}
\acmYear{2024}
\copyrightyear{2024}
\acmSubmissionID{popl24main-p766-p}
\acmJournal{PACMPL}
\acmVolume{8}
\acmNumber{POPL}
\acmArticle{93}
\acmMonth{1}
\received{2023-07-11}
\received[accepted]{2023-11-07}

\begin{document}

\iftoggle{arxiv}{
\title[Programmatic Strategy Synthesis]{Programmatic Strategy Synthesis \\ \normalsize Resolving Nondeterminism in Probabilistic Programs}
}{
\title{Programmatic Strategy Synthesis: Resolving Nondeterminism in Probabilistic Programs}
}

\titlenote{Batz and Katoen are supported by the ERC AdG 787914 FRAPPANT. Winkler is supported by the DFG RTG 2236 UnRAVeL.}

\author{Kevin Batz}
\orcid{0000-0001-8705-2564}
\affiliation{%
    \institution{RWTH Aachen University}
    \city{}
    \country{Germany}
}
\email{kevin.batz@cs.rwth-aachen.de}

\author{Tom Jannik Biskup}
\orcid{0009-0002-0169-641X}
\affiliation{%
    \institution{RWTH Aachen University}
    \city{}
    \country{Germany}
}
\email{tom.biskup@rwth-aachen.de}

\author{Joost-Pieter Katoen}
\orcid{0000-0002-6143-1926}
\affiliation{%
    \institution{RWTH Aachen University}
    \city{}
    \country{Germany}
}
\email{katoen@cs.rwth-aachen.de}

\author{Tobias Winkler}
\orcid{0000-0003-1084-6408}
\affiliation{%
    \institution{RWTH Aachen University}
    \city{}
    \country{Germany}
}
\email{tobias.winkler@cs.rwth-aachen.de}

\begin{abstract}
    We consider imperative programs that involve \emph{both} randomization \emph{and} pure nondeterminism.
The central question is how to find a strategy resolving the pure nondeterminism such that the so-obtained \emph{determinized} program satisfies a given quantitative specification, i.e., bounds on expected outcomes such as the expected final value of a program variable or the probability to terminate in a given set of states.
We show how \emph{memoryless and deterministic (MD)} strategies can be obtained in a semi-automatic fashion using deductive verification techniques. 
For loop-free programs, the MD strategies resulting from our weakest precondition-style framework are correct by construction.
This extends to loopy programs, provided the loops are equipped with suitable loop invariants --- just like in program verification.
We show how our technique relates to the well-studied problem of obtaining strategies in countably infinite Markov decision processes with reachability-reward objectives. Finally, we apply our technique to several case studies.

\end{abstract}


\allowdisplaybreaks[0]


\maketitle

\section{Introduction}
\kbin{mind derek's mail}
\label{s:intro}
\emph{Nondeterministic probabilistic programs} are like usual imperative programs with the additional abilities to
(1) flip coins and
(2) choose between multiple execution branches in a \emph{purely nondeterministic} fashion.
These programs have various applications:
If the nondeterminism models an \emph{uncontrollable adversary}, they can be used to reason about safe abstractions or \emph{underspecifications} of (fully) probabilistic programs, processes, and systems~\cite{DBLP:journals/rc/KozineU02,DBLP:conf/vmcai/KattenbeltKNP09,DBLP:journals/fmsd/KattenbeltKNP10}.
If, on the other hand, the nondeterminism models a \emph{controllable agent}, such programs model games and security mechanisms under stochastic adversaries, and can be used for planning and control, see~\cite{DBLP:conf/iccps/FengWHT15,DBLP:journals/pe/HaesaertCA17} for various applications.
\tw{check these refs, maybe add more}

In many of these settings, a central problem is to determine \emph{strategies}  --- resolutions of the nondeterminism --- satisfying a given specification. In this paper, we consider quantitative specifications imposing bounds on probabilities to terminate in a given set of states or on expected outcomes, e.g., \enquote{{the expected final value of $x$ is at most $5$}}. 
We tackle this problem from a programmatic perspective:\smallskip
\begin{center}
	\begin{minipage}{0.92\textwidth}
		\centering
		\emph{Given a nondeterministic probabilistic program $\cc$ and a quantitative specification,}\\
		\emph{if possible, \underline{resolve the nondeterminism in $\cc$} to obtain a deterministic, but still probabilistic, program $\cc'$ that satisfies the given quantitative specification.}
	\end{minipage}
\end{center}\smallskip
The idea is that the \enquote{determinized} program $\cc'$ is a \emph{symbolic representation} of a 
strategy steering the program execution towards satisfying the given quantitative specification.
%
%
%
%
\paragraph{Modeling and Reasoning with Nondeterministic Probabilistic Programs.}
To illustrate probabilistic programs with nondeterminism as a means to reason about decision making under stochastic uncertainty, consider the (in)famous \emph{Monty Hall problem}:
In a game show, a prize is hidden behind one of three curtains, the other curtains hide nothing. Each of these three curtains hides the prize with equal \emph{probability} \nicefrac{1}{3}. The game is then played as follows:
\begin{enumerate}
    \item The contestant selects one curtain (which is not uncovered yet).
    \item The host uncovers one of the other two curtains \emph{which does not hide the prize}.
    \item Finally, the contestant may ---but does not have to--- \emph{switch} their choice to the remaining curtain, i.e., the one that is not yet uncovered and not picked initially.
\end{enumerate}
The chosen curtain is then uncovered and the contestant wins whatever is behind it.
The question is if switching the curtain in stage (3) increases the chances to win.

The Monty Hall problem is readily modeled as the program in \Cref{fig:monty-hall-intro} (adapted from \cite{McIverM05}).
Variables $\contestantCurtain, \priceCurtain, \alternativeCurtain \in \{1,2,3\}$ store the initially chosen contestant curtain, the curtain hiding the prize, and the curtain which is uncovered by the host, respectively.
The expression $\otherCurtain{c_1}{c_2}$ returns the unique curtain $c_3$ with $c_3 \neq c_1 \land c_3 \neq c_2$.
Similarly, $\randomOtherCurtain{a}$ yields one of the two curtains other than $a$ chosen uniformly at random.
The \emph{guarded commands} $\GCFST{\true}{\cc_1}{~{}\square{}~}\hspace{0.015cm} \ldots \GCSEC{\true}{\cc_n}$ model \emph{nondeterministic} choices between the respective branches.
To account for the fact that the contestant does not know which curtain hides the prize, the contestant commits to their strategy upfront (lines 1-2).

\begin{wrapfigure}[14]{r}{0.475\textwidth}
    \small
    \vspace{-6mm}
    \begin{minipage}{0.475\textwidth}
        \newcommand{\linenum}[1]{\gray{\mathtt{#1}}}
        \begin{align*}
        \linenum{1}\quad&\GCFST{\true}{\ASSIGN{\clever}{\true}}\\
        \linenum{2}\quad&\GCSEC{\true}{\ASSIGN{\clever}{\false}} \COMPSEMIC\\
        \linenum{3}\quad&\GCFST{\true}{\ASSIGN{\contestantCurtain}{1}}\\
        \linenum{4}\quad&\GCSEC{\true}{\ASSIGN{\contestantCurtain}{2}}\\
        \linenum{5}\quad&\GCSEC{\true}{\ASSIGN{\contestantCurtain}{3}} \COMPSEMIC\\
        \linenum{6}\quad&\UNIFASSIGN{\priceCurtain}{1,2,3} \COMPSEMIC \\
        \linenum{7}\quad&\GCFST%
        {\priceCurtain\neq\contestantCurtain}
        {\ASSIGN{\alternativeCurtain}{\otherCurtain{pc}{cc}}}\\
        \linenum{8}\quad&\GCSEC%
        {\priceCurtain = \contestantCurtain}
        {\ASSIGN{\alternativeCurtain}{\randomOtherCurtain{\priceCurtain}}} \COMPSEMIC\\
        \linenum{9}\quad&\GCFST{\clever}{\ASSIGN{\contestantCurtain}{\otherCurtain{\contestantCurtain}{\alternativeCurtain}}}\\
        \linenum{10}\quad&\GCSEC{\neg\clever}{\SKIP}
        \end{align*}
    \end{minipage}
    \caption{
        Monty Hall problem.
    }
    \label{fig:monty-hall-intro}
\end{wrapfigure}

\paragraph{Strategies for Loop-Free Programs}
Using the weakest preexpecation ($\WPSYMBOL$) calculus by \citet{McIverM05}, one can prove that the \emph{maximal} probability to win is $\nicefrac 2 3$.
\begin{quote}
    \emph{Our technique employs this $\WPSYMBOL$-calculus to construct, in a fully mechanizable way, a strategy in form of \underline{stren}g\underline{thenin}g\underline{s} of the predicates guarding the nondeterministic choices, which attains this winning probability.}
\end{quote}
More precisely, we obtain a determinized program in which lines 1-2 are replaced by a statement equivalent to $\ASSIGN{\clever}{\true}$, as \emph{switching in  stage (3) is indeed the best strategy}.
The choices in lines 3-5 are not resolved as they do not affect the winning probability.
In fact, our strategies are generally \emph{permissive}~\cite{DBLP:journals/corr/DragerFK0U15}: they do not remove more nondeterminism than necessary.

\paragraph{Loops, Unbounded Variables, and Invariants.}
The program from \Cref{fig:monty-hall-intro} contains no loops and only bounded variables.
However, our technique also yields strategies for \emph{possibly unbounded} nondeterministic probabilistic loops, provided they are annotated with suitable \emph{quantitative loop invariants}. To illustrate this, consider the program in \Cref{fig:game-program}.
%
\begin{figure}[tbh]
    \small
    \begin{minipage}[t]{0.45\textwidth}
        \begin{align*}
        &\WHILEDOINVOPEN{x<N}\\
        &\qquad\GCFST{\POneTurn}{\UNIFASSIGN{x}{x{+}1,\, x{+}2,\, x{+}3}}\\
        &\qquad\GCSECOPEN{\varTurn = 2}\\
        &\qquad\begin{aligned}[t]
        &\qquad\GCFST{\true}{\ASSIGN{x}{x+1}}\\
        &\qquad\GCSEC{\true}{\ASSIGN{x}{x+2}}\\
        &\qquad\GCSEC{\true}{\ASSIGN{x}{x+3}}\\
        &\qquad\GCSECCLOSE \\
        \end{aligned}\\
        &\qquad\ASSIGN{\varTurn}{3 - \varTurn} \\
        &\WHILEDOINVCLOSE{\IGame}
        \end{align*}
        \caption{
            Program modeling the Nim game.
        }
        \label{fig:game-program}
    \end{minipage}
    \hfill
    \begin{minipage}[t]{0.50\textwidth}
        \begin{align*}
        &\WHILEDOINVOPEN{x<N}\\
        &\qquad\GCFST{\POneTurn}{\UNIFASSIGN{x}{x{+}1,\, x{+}2,\, x{+}3}}\\
        &\qquad\GCSECOPEN{\varTurn = 2}\\
        &\qquad\begin{aligned}[t]
        &\qquad\GCFST{\alertannocolor{(x+1 \congmod 4 N) \lor (x+2 \congmod 4 N)}}{\ASSIGN{x}{x+1}}\\
        &\qquad\GCSEC{\alertannocolor{(x+1 \congmod 4 N) \lor (x+3 \congmod 4 N)}}{\ASSIGN{x}{x+2}}\\
        &\qquad\GCSEC{\alertannocolor{(x+1 \congmod 4 N) \lor (x \congmod 4 N)}}{\ASSIGN{x}{x+3}}\\
        &\qquad\GCSECCLOSE \\
        \end{aligned}\\
        &\qquad\ASSIGN{\varTurn}{3 - \varTurn} \\
        &\WHILEDOINVCLOSE{\IGame}
        \end{align*}
        \caption{
            Program representing strategies for Nim .
        }
        \label{fig:game-program-tamed}
    \end{minipage}
\end{figure}
%
It models a variant of the game \emph{Nim}, a 2-player zero-sum game which goes as follows:
$N$ tokens are placed on a table.
The players take turns; in each turn, the player has to remove 1, 2, or 3 tokens from the table.
The first player to remove the last token looses the game.
We have annotated the loop with a quantitative invariant $\IGame$ (see \Cref{sec:case_studies} for details).
This invariant certifies that the \emph{maximal} winning probability of the controllable player $2$ is at least $\nicefrac 2 3$. Now,
\begin{quote}
    \emph{Given a quantitative specification and suitably strong\kb{added suitably strong} quantitative loop invariants for all --- possibly nested --- loops in a program $\cc$, our technique automatically yields a program $\cc'$ where the nondeterministic choices are restricted in such a way that any strategy consistent with $\cc'$ satisfies the desired quantitative specification.}
\end{quote}
\noindent
The result of applying this technique to our example is given in \Cref{fig:game-program-tamed}.
\emph{Any} strategy playing the game according to this program wins with probability at least $\nicefrac 2 3$ \emph{for all values of $N$}.

\paragraph{Contributions}
In summary, this paper makes the following contributions:
\begin{description}
	\item[Program-level construction of strategies]
A mechanizable technique to determinize loop-free nondeterministic probabilistic programs in an optimal manner (\Cref{thm:trans_preserves_wpre}).
    Given quantitative loop invariants, our technique as well determinizes programs with --- possibly multiple nested and sequential --- loops, in a mechanizable way (\Cref{thm:correct_transformations}).
	\item[A novel proof rule for lower bounds on expected outcomes of loops]
	As a by-product of our results, we obtain a generalization of a powerful proof principle for verifying lower bounds on expected outcomes of \emph{deterministic} probabilistic loops to \emph{non}deterministic probabilistic loops ---
	a problem left open by \citet{DBLP:journals/pacmpl/HarkKGK20}. See \Cref{rem:aiming_low} for details.
	\item[A bridge to the world of Markov Decision Processes]
	We establish tight connections between our setting and the well-known problem of finding good, or even optimal, strategies for resolving nondeterminism in countably infinite MDPs.
	\item[Case studies and examples]
	\hspace{-4pt} to demonstrate the applicability of our technique (\Cref{sec:case_studies}).
\end{description}

\paragraph{Limitations and Assumptions}
We focus on quantitative specifications imposing lower- or upper bounds on expected outcomes of a program $\cc$, i.e., bounds on expected values of random variables w.r.t.\ the distribution of final states obtained from running $\cc$ (cf.\ \Cref{sec:problem_statement} for a formal problem statement).
We do \emph{not} consider long-run properties such as mean-payoff objectives~\cite{DBLP:books/wi/Puterman94} or general $\omega$-regular  properties~\cite[Ch.~10]{DBLP:books/daglib/0020348}.
Further, we restrict to \emph{memoryless and deterministic} strategies (as opposed to history-dependent randomized strategies).
For countable state spaces, this is sufficient for ($\varepsilon$-)optimality w.r.t.\ the considered objectives (cf.\ \Cref{sec:impossibilityMDP}). For this reason, we assume that program variables range over a countable domain (cf.\ \Cref{sec:syntax}).

Regarding mechanizability, our approach automatically yields optimal strategies for \emph{loop-free} programs. Hence, if it exists, we can compute a determinized program satisfying the given quantitative specification (cf.\ \Cref{sec:overview_loopfree} and \Cref{thm:trans_preserves_wpre}).
 In the presence of loops, we require that all loops are annotated with sufficiently strong invariants. We do \emph{not} automate the synthesis or the verification of such invariants. Rather, we show that if a nondeterministic program is annotated with such loop invariants, then suitable determinizations can be \mbox{computed (cf.\ \Cref{sec:overview_loops} and \Cref{thm:correct_transformations}).}

\paragraph{Paper Structure}
\Cref{sec:nondetprogs,sec:wp} introduce the deductive verification techniques and operational MDP semantics our technique is based on.
In \Cref{sec:overview} we provide an informal, illustrative bird's eye view on our approach. 
Our technique is parameterized by quantitative verification condition generators, which we introduce in \Cref{sec:vc}.
Program-level synthesis of strategies --- the main technical contribution of our paper --- is described in \Cref{sec:trans}.
We provide further case studies in \Cref{sec:case_studies}, discuss related works in \Cref{sec:rel_work}, and conclude in \Cref{sec:conclusion}.

\section{Nondeterminstic Probabilistic Programs}
\label{sec:nondetprogs}

In this section, we introduce the syntax as well as an operational Markov decision process semantics of a simple probabilistic programming language featuring nondeterminstic choices.

\subsection{Syntax}
\label{sec:syntax}

Let $\Vars = \{x,y,z,\ldots\}$ be a countably infinite set of \emph{(program) variables} with values\footnote{We have chosen the value domain $\PosRats$ for the sake of concreteness. Our results straightforwardly apply to more general (possibly many-sorted) countable domains.} from the set $\Vals = \PosRats$. The countably infinite set of \emph{(program) states} is 
\[
   \States \eeq \{ \pstate \colon \Vars \to \Vals ~\mid~  \text{$\pstate(x)=0$ for all but finitely many $x\in\Vars$}\}~.
\]
The set of \emph{predicates} over $\States$ is $
\predicates = \big\{ \guard \mid \guard\colon \States \to \{\true,\false \} \big\}$. We write $\pstate \models \guard$ instead of $\pstate \in \guard$.
A predicate $\guard$ is \emph{valid}, denoted $\entails \guard$, if $\pstate \models \guard$ for every $\pstate$, and it is called \emph{unsatisfiable} if $\neg\guard$ is valid.
Programs $\cc$ in the \emph{probabilistic guarded command language} $\pgcl$ adhere to the grammar
\begin{align*}
    \cc \qquad::=\qquad& \SKIP \tag{effectless program}\\
    &\mid \ASSIGN{x}{\aexpr} \tag{variable assignment} \\
    &\mid \COMPOSE{\cc}{\cc} \tag{sequential composition}\\
    &\mid \GC{\guard_1}{\cc}{\guard_2}{\cc} \tag{guarded choice} \\
    &\mid \PCHOICE{\cc}{\pexpr}{\cc} \tag{probabilistic choice} \\
    &\mid \WHILEDOINV{\guard}{\cc}{\inv} \tag{loop with invariant annotation $\inv$}
\end{align*}
where $\aexpr$ is a function of type $\States \to \Vals$ called \emph{expression},
$\guard_1, \guard_2$ and $\guard$ are predicates from $\predicates$,
and $\pexpr$ is a \emph{probability expression} of type $\States \to [0,1]$.
For the guarded choice, we require that for every state $\pstate$, either $\pstate \entails \guard_1$ or $\pstate \entails \guard_2$ or both hold, i.e., that $\guard_1 \vee \guard_2$ is valid.
We call a program $\cc \in \pgcl$ \emph{deterministic} if for all $\GC{\guard_1}{\cc_1'}{\guard_2}{\cc_2'}$ occurring in $\cc$, $\guard_1 \wedge \guard_2$ is unsatisfiable.
Loops are annotated with \emph{quantitative invariants} $\inv$ which are functions of type $\States \to \PosRealsInf$ (see \Cref{sec:overview_loops} for details). We often omit the invariant $\inv$ if it is irrelevant in the current context.

We briefly describe each $\pgcl$ construct.
$\SKIP$ does nothing. $\ASSIGN{x}{\aexpr}$ evaluates expression $\aexpr$ in the current state and assigns the resulting value to variable $x$.
$\COMPOSE{\cc_1}{\cc_2}$ first executes $\cc_1$, and then --- if $\cc_1$ terminates --- $\cc_2$.
The guarded choice $\GC{\guard_1}{\cc_1}{\guard_2}{\cc_2}$ first checks which of the guards $\guard_1$ and $\guard_2$ evaluates to $\true$ under the current state.
If only one of the guards, say $\guard_1$, evaluates to $\true$, then the guarded choice deterministically executes $\cc_1$.
If \emph{both} $\guard_1$ and $\guard_2$ evaluate to $\true$, then the guarded choice behaves \emph{nondeterministically} by \emph{either} executing $\cc_1$ \emph{or} $\cc_2$.
Notice that standard conditional choice $\ITE{\guard}{\cc_1}{\cc_2}$ is syntactic sugar for $\GC{\guard}{\cc_1}{\neg\guard}{\cc_2}$.
The probabilistic choice $\PCHOICE{\cc_1}{\pexpr}{\cc_2}$ introduces randomization: In state $\pstate$, $\cc_1$ is executed with probability $\pexpr(\pstate)$ and $\cc_2$ is executed with the remaining probability $1-\pexpr(\pstate)$.
Finally, the loop $\WHILEDO{\guard}{\cc}$ executes the loop body $\cc$ as long as $\guard$ evaluates to $\true$, which is the only possible source of non-termination.
We conclude this section with an example.
\begin{example}
    \label{ex:first_example_prog}
    Consider the program $\cc$ in \Cref{fig:first_example_prog}.
    Program $\cc$ is a loop which contains both randomization and nondeterminism.
    In each iteration, a fair coin is flipped.
    Depending on the outcome, the loop either terminates or increments $x$ nondeterministically by either $1$ or $2$.  \hfill $\triangle$
\end{example}

\subsection{Markov Decision Process Semantics}
\label{sec:mdp_semantics}
In this section, we define a formal semantics of $\pgcl$ programs in terms of (countably infinite) \emph{Markov decision processes} (\emph{MDP}), based on \cite{operational_vs_weakest} with adaptions from \cite{qsl}.

Formally, an MDP is a quadruple $\MDP = \left( \MS, \MI, \MA, \MP \right)$ where
$\MS$ is a countable non-empty set of states,
$\MI \subseteq \MS$ is a \emph{set} of initial states,
$\MA$ is a finite non-empty set of action labels,
and $\MP \colon \MS \times \MA \times \MS \rightarrow [0,1]$ is a transition probability function such that for all $\ms \in \MS$ and all $\ma\in\MA$, $\sum_{\ms'\in\MS} \MP(\ms,\ma,\ms')\in\{0,1\}$. 
For $\ms \in \MS$ we write $\MA(\ms) = \{\ma \in \MA \mid \sum_{\ms'\in\MS} \MP(\ms,\ma,\ms') = 1\}$ and require that $|\MA(\ms)|  \geq 1$ for every $\ms$.
An MDP $\MDP$ is called \emph{deterministic}, if $|\MA(\ms)| = 1$ for all $\ms \in \MS$.

\subsubsection{MDP Semantics of $\pgcl$}
\label{sec:actualMDPconstruction}
We first define a small-step execution relation $\ExecSymbol$ between \emph{program configurations}.
These configurations consist of (i) either a $\pgcl$ program $\cc$ that is still to be executed or a symbol $\Term$ indicating termination and (ii) a program state $\sigma$ from $\States$.
Formally, the countable set $\OpStates$ of \emph{program configurations} is given by
$\OpStates \eeq \left(\pgcl \cup \{\,\Term\,\}\right) \,\times\, \States$.
The small-step execution relation is of the form
$\ExecSymbol ~\subseteq~ \OpStates \,\times\, \{\labtau, \labalpha, \labbeta\} \,\times\, [0,1] \,\times\, \OpStates$.
Intuitively, we can think of the elements from $\ExecSymbol$ as labeled transitions between program configurations. The second component of $\ExecSymbol$ is an \emph{action label} $\actlab \in \{\labtau, \labalpha, \labbeta\}$, and the third component is the transition's probability $\pnum \in [0,1]$.
The formal definition of $\ExecSymbol$ is standard and given in \iftoggle{arxiv}{\Cref{app:mdp_rules}}{\cite[Appendix B]{arxiv}}.
In a nutshell, $\ExecSymbol$ realizes the intended semantics of $\pgcl$ programs as described in \Cref{sec:syntax}.
The action label $\labtau$ is used for all transitions except for those corresponding to a (possibly nondeterministic) guarded command $\GC{\guard_1}{\cc_1}{\guard_2}{\cc_2}$.
The transition labels $\labalpha$ and $\labbeta$ distinguish between the branches chosen by such a guarded command.
We often write $\ExecAbbr{\mc}{\actlab}{\pnum}{\mc'}$ instead of $(\mc,\actlab,\pnum,\mc') \in\,\ExecSymbol$.
We also define a binary \emph{successor relation} $\mc \rightharpoonup \mc' \iff  \exists \actlab, \pnum \colon \ExecAbbr{\mc}{\actlab}{\pnum}{\mc'}$.
We say that configuration $\mc'$ is \emph{reachable} from $\mc$ if $(\mc,\mc')$ is in the reflexive-transitive closure of the relation $\rightharpoonup$. 
Based on the execution relation $\ExecSymbol$ we can now define our MDP semantics:

\begin{definition}
    The \emph{operational Markov decision process} $\opMDP{\cc}$ of program $\cc \in \pgcl$ is
    \begin{align*}
    \opMDP{\cc} \eeq \left( \MS, \MI, \MA, \MP \right)~\text{, where}
    \end{align*}
    \begin{enumerate}
        \item $\MS = \{ (\cc', \pstate') \in \OpStates \mid \exists \pstate \in \States \colon (\cc',\pstate') \text{ is reachable from } (\cc,\pstate)\}$,
        \item $\MA = \{\labtau,\labalpha,\labbeta\}$ is the set of \emph{action labels},
        \item $\MP\colon \MS \times \MA \times \MS \to [0,1]$ is the \emph{transition probability function} given by 
        \[
        \MP(\mc, \actlab,\mc')\eeq
        \begin{cases}
        \pnum & \text{ if } \ExecAbbr{\mc}{\actlab}{\pnum}{\mc'} \\
        0 & \text{else}~,
        \end{cases}
        \]
        \item and $\MI = \{(\cc, \pstate) \mid \pstate \in \States \}$ are the initial states of $\opMDP{\cc}$.
    \end{enumerate}
\end{definition}%

\begin{example}
    A fragment of the MDP $\opMDP{\cc}$ of program $\cc$ from \Cref{ex:first_example_prog} is sketched in \Cref{fig:mdp_sem_example}, where we assume that $x \in \Nats$ and $c \in \{0,1\}$.
    We depict a fragment reachable from the initial states where $c=0$ (middle row).
    Transition probabilities equal to $1$ are omitted. \hfill $\triangle$
\end{example}

\begin{figure}[t]
    \begin{minipage}[b]{0.34\textwidth}
        \small
        \begin{align*}
        &\WHILE{c = 0} \\[1mm]
        &\qquad \{ \ASSIGN{c}{1} \} ~[0.5]~ \{ \\[1mm]
        &\qquad\qquad \GCFST{\true}{\ASSIGN{x}{x+1}} \\[1mm]
        &\qquad\qquad \GCSEC{\true}{\ASSIGN{x}{x+2}} \\[1mm]
        &\qquad \} \\[1mm]
        & \}
        \end{align*}
        \caption{Program $\cc$ from \Cref{ex:first_example_prog}.}
        \label{fig:first_example_prog}
    \end{minipage}
    \hfill
    \begin{minipage}[b]{0.60\textwidth}
        \begin{tikzpicture}[thick, on grid, node distance=12mm and 16mm, every state/.style={minimum size=6mm}, every node/.style={scale=0.8},initial where=right, initial text=]
        \node[state,label={left:$c=0$}] (0a) {};
        \node[state,above=of 0a,initial,label={left:$c=0$}] (0b) {};
        \node[state,above=of 0b,accepting,label={above:$x=0$},label={left:$c=1$}] (0c) {$\Term$};
        \node[state,right=of 0a] (1a) {};
        \node[state,right=of 0b,initial] (1b) {};
        \node[state,right=of 0c,accepting,label={above:$x=1$}] (1c) {$\Term$};
        \node[state,right=of 1a] (2a) {};
        \node[state,right=of 1b,initial] (2b) {};
        \node[state,right=of 1c,accepting,label={above:$x=2$}] (2c) {$\Term$};
        \node[state,right=of 2a] (3a) {};
        \node[state,right=of 2b,initial] (3b) {};
        \node[state,right=of 2c,accepting,label={above:$x=3$}] (3c) {$\Term$};
        \node[right=of 3a] (4a) {$\ldots$};
        \node[right=of 3b] (4b) {$\ldots$};
        \node[right=of 3c] (4c) {$\ldots$};
        
        \draw[->] (0b) edge node[right,near start] {$\labtau,0.5$} (0a);
        \draw[->] (0b) edge node[right,near start] {$\labtau,0.5$} (0c);
        \draw[->] (1b) edge node[right,near start] {$\labtau,0.5$} (1a);
        \draw[->] (1b) edge node[right,near start] {$\labtau,0.5$} (1c);
        \draw[->] (2b) edge node[right,near start] {$\labtau,0.5$} (2a);
        \draw[->] (2b) edge node[right,near start] {$\labtau,0.5$} (2c);
        \draw[->] (3b) edge node[right,near start] {$\labtau,0.5$} (3a);
        \draw[->] (3b) edge node[right,near start] {$\labtau,0.5$} (3c);
        
        \draw[->] (0a) edge node[below,near start] {$\labalpha$} (1b);
        \draw[->] (0a) edge[bend right=55] node[above, near start] {$\labbeta$} (2b);
        \draw[->] (1a) edge node[below,near start] {$\labalpha$} (2b);
        \draw[->] (1a) edge[bend right=55] node[above, near start] {$\labbeta$} (3b);
        \draw[->] (2a) edge node[below,near start] {$\labalpha$} (3b);
        \draw[->] (2a) edge[bend right=55] node[above, near start] {$\labbeta$} (4b);
        \draw[->] (3a) edge node[below,near start] {$\labalpha$} (4b);
        
        \draw[->] (0c) edge[loop right] node[right] {$\labtau$} (0c);
        \draw[->] (1c) edge[loop right] node[right] {$\labtau$} (1c);
        \draw[->] (2c) edge[loop right] node[right] {$\labtau$} (2c);
        \draw[->] (3c) edge[loop right] node[right] {$\labtau$} (3c);
        \end{tikzpicture}
        \caption{
            The semantic MDP $\opMDP{\cc}$ of program $\cc$ from \Cref{fig:first_example_prog}.
        }
        \label{fig:mdp_sem_example}
    \end{minipage}
\end{figure}

\subsubsection{Strategies}
A \emph{strategy} for $\MDP = \left( \MS, \MI, \MA, \MP \right)$ is a function $\msched \colon \MS^{+} \to \Dist\MA$ satisfying $\supp{\msched(\ms_0\ldots \ms_n)} \subseteq \MA(\ms_n)$ for all $\ms_0\ldots \ms_n \in \MS^{+}$.
Strategy $\msched$ is called \emph{memoryless} if $\msched(\ms_0\ldots \ms_n)$ depends only on $\ms_n$, and \emph{deterministic} if $\msched(\ms_0\ldots \ms_n)$ is a Dirac distribution.
Memoryless strategies can be identified with maps of type $\MS \to \Dist\MA$, and strategies which are both \emph{memoryless and deterministic} (\emph{MD}) can be identified with maps of type $\MS \to \MA$.
An MD strategy $\msched \colon \MS \to \MA$ for $\MDP$ induces a deterministic MDP $\MDP^\msched = \left( \MS, \MI, \MA, \MP' \right)$ where for $\ms,\ms' \in \MS$ and $\ma \in \MA$,
\[
\MP'(\ms,\ma,\ms') \eeq 
\begin{cases}
	 \MP(\ms,\ma,\ms'), &\text{if $a = \msched(\ms)$}\\
	 0, &\text{otherwise}
\end{cases}~.
\]
We call MDP $\MDP'$ a \emph{determinization} of $\MDP$ if there exists an MD strategy $\msched$ such that $\MDP' = \MDP^\msched$.

\subsubsection{Objectives}
\label{sec:mdp_objectives}
The fundamental objective considered in this paper is called \emph{reachability-reward}, a mixture of reachability and expected reward objectives.
Intuitively, reachability-reward is like standard reachability, with the difference that each state in the target set $T$ is a sink and has a $\PosRealsInf$-valued reward that is collected when that state is reached for the first time.
The goal is to either minimize or maximize the expected reward.
We now formalize reachability-reward.
For $\MT\subseteq \MS$, the set of finite paths eventually reaching $\MT$ is
\[
   \pathstotarget{\MT} \eeq \{\, \ms_0 \ldots \ms_m \in \MS^{+}~\mid~  \ms_m \in \MT,\, \forall k \in\{0,\ldots,m-1\}\colon \ms_k\not\in \MT  \,\} ~.
\]
Moreover, given $\ms_0\ldots\ms_m \in \MS^+$ and strategy $\msched$, we define
\[
  \probofpath{\msched}{\ms_0\ldots\ms_m }
  \eeq
  \prod_{0 \leq k <m}
  \underbrace{\sum_{\actlab \in \MA} \msched(\ms_0\ldots\ms_k)(\actlab) 
  \cdot \MP(\ms_k,\actlab,\ms_{k+1})}_{\text{probability to move in one step from $\ms_k$ to $\ms_{k+1}$ under strategy $\msched$}}
  ~,
\]
with the convention that the empty product (for $m=0$) is equal to $1$.
Finally, given $\MR \colon \MT \to \NonnegRealsInf$, we define the function $\exprew{\msched}{\MDP}{\MR} \colon \MI \to \PosRealsInf$, which maps every initial state $\mi$ to its expected (reachability-)reward $\exprew{\msched}{\MDP}{\MR}(\mi)$ under strategy $\msched$ by
\[
    \exprew{\msched}{\MDP}{\MR}(\mi) \eeq
    \sum_{\mi\ms_1\ldots\ms_m\in \pathstotarget{\MT}} \probofpath{\msched}{\mi\ms_1\ldots\ms_m} \cdot \MR(\ms_m)~.
\]
Notice that, for constant reward functions $\MR \colon \MT \to \{1\}$, $\exprew{\msched}{\MDP}{\MR}(\mi)$ is the probability to reach $T$ from initial state $\mi$ under strategy $\msched$.
Finally, we define
\begin{align}
  \minexprew{\MDP}{\MR} \eeq& \inf_{\msched \in \MSCHED} \exprew{\msched}{\MDP}{\MR} \label{eq:defMinExpRew}\\
  \text{and} \quad
   \maxexprew{\MDP}{\MR} \eeq& \sup_{\msched \in \MSCHED} \exprew{\msched}{\MDP}{\MR} \label{eq:defMaxExpRew}~,
\end{align}
where $\MSCHED$ denotes the set of all strategies for $\MDP$ (possibly using memory and/or randomization), and where the infimum and supremum are understood pointwise.
If $\MDP$ is deterministic, then there exists exactly one strategy $\msched$, and thus $\minexprew{\MDP}{\MR} = \maxexprew{\MDP}{\MR}$.
This justifies writing $\exprewmc{\MDP}{\MR}$ instead of $\minexprew{\MDP}{\MR}$ (or $\maxexprew{\MDP}{\MR}$).

\begin{example}
    Reconsider the MDP $\opMDP{\cc}$ from \Cref{fig:mdp_sem_example}.
    Suppose that reward function $\MR$ assigns the value of $x$ to each terminal state in the top row.
    Then $\maxexprew{\opMDP{\cc}}{\MR}(C,\sigma) = \sigma(x)+2$ for every initial state $(C,\sigma)$ with $\sigma(c)=0$. This maximal expected reward is attained by the MD strategy that always chooses action $\labbeta$. \hfill $\triangle$
\end{example}
\noindent
Our ultimate goal is to apply deductive program verification techniques to solve the strategy synthesis problem for MDPs arising from $\pgcl$ programs (see also \Cref{sec:problem_statement}):
\smallskip
\begin{center}
	\emph{Given MDP $\MDP$, reward function $\MR$, ${}\somenumrel{} \in \{\leq,\geq\}$, and thresholds $\thresh \colon \MI \to \PosRealsInf$,}\\
	\emph{if it exists, find a \underline{memor}y\underline{less deterministic} strategy $\msched$ with $\exprew{\msched}{\MDP}{\MR}  \somenumrel \thresh$} \\[0mm]
	or, equivalently, \\[0mm]
	\emph{if it exists, find a determinization $\MDP'$ of $\MDP$ with $\exprewmc{\MDP'}{\MR}\somenumrel\thresh$~.}
\end{center}
\smallskip
The comparison relation $\somenumrel$ in the problem statement above refers to the pointwise lifted order $\leq$ (resp.\ $\geq$) on $\PosRealsInf$ to maps of type $\MI \to \PosRealsInf$.
Note that we are interested in finding strategies that guarantee a given threshold for \emph{all} initial states.
Such strategies are called \emph{uniform} in the literature.

\subsubsection{Existence of Optimal MD Strategies}
\label{sec:impossibilityMDP}
Strategies that attain the infimum in \eqref{eq:defMinExpRew} (resp.\ the supremum in \eqref{eq:defMaxExpRew}) are called \emph{optimal}.
It is known that \emph{optimal MD strategies always exist in the minimizing setting}~\cite[Theorem 7.3.6a]{DBLP:books/wi/Puterman94}.
In fact, our theory established in the upcoming sections implies this result for the class of MDPs described by $\pgcl$ programs. 
The above problem is thus guaranteed to have a solution $\msched$ if $\somenumrel$ is $\leq$ and $\thresh \geq \minexprew{\MDP}{\MR}$.

The maximizing setting is fundamentally different and more subtle~\cite{blackwell1967positive,ornstein1969existence}.
First, optimal maximizing strategies do not exist in general, i.e., the supremum in \eqref{eq:defMaxExpRew} might not be attained.
For instance, the MDP $\MDP$ in \Cref{fig:ornsteinExample} (black states only, \blue{blue} rewards, $\MI = $ topmost row) satisfies $\maxexprew{\MDP}{\MR} = 1$, but for all strategies $\msched$ and initial states $\ms \in \MI$ we have $\exprew{\msched}{\MDP}{\MR}(\ms) < 1$.
To see this, observe that any strategy that reaches $\MT$ with positive probability must play some action $\labalpha$ (say the $n$-th one) with positive probability $p > 0$, resulting in an expected reward of at most $p \cdot \frac{n}{n+1} + (1 - p)$ which is clearly less than $1$.
This example shows that, unlike in the minimizing setting, the above problem does \emph{not} necessarily have a solution if $\somenumrel$ is $\geq$ and $\thresh \leq \maxexprew{\MDP}{\MR}$, \emph{even if general strategies are allowed}.

On the other hand, there are also situations where MD strategies are strictly less powerful than general strategies.
For example, consider the MDP in \Cref{fig:ornsteinExample} (black states only, \red{red} rewards, $\MI = $ topmost row).
Here, $\maxexprew{\MDP}{\MR} = \infty$ as witnessed by the \emph{optimal randomizing strategy} that plays each action with probability $\nicefrac 1 2$.
In contrast, each MD strategy obviously yields only finite expected reward.
MD strategies are still somewhat useful in this example because for every \emph{constant} threshold $\thresh \in \PosReals$ --- no matter how large --- we can find an MD strategy that yields expected reward at least $\thresh$ for each initial state.
However, this does not hold in general.
In fact, \citet{ornstein1969existence} gives an MDP $\MDP$ where $\maxexprew{\MDP}{\MR} = \infty$ which is attained by a randomizing strategy, but for all MD strategies $\msched$ and all $\varepsilon > 0$, there exists $\ms \in \MI$ such that $\exprew{\msched}{\MDP}{\MR} \leq \varepsilon$ (\Cref{fig:ornsteinExample}, \gray{gray} \& black states, \red{red} reward function, $\MI =$ gray states).

MD strategies are nonetheless guaranteed to be reasonably powerful under mild assumptions, which we summarize in the following theorem (where (2) is due to \cite{ornstein1969existence}):
\begin{theorem}
    \label{thm:mdp_md_good}
    Consider an MDP $\MDP$ with countable $\MS$ and reachability-reward function $\MR$.
    \begin{enumerate}
        \item For finite $\MS$ there exists an MD optimal maximizing strategy~\cite[Ch. 10]{DBLP:books/daglib/0020348}.
        \item\cite{ornstein1969existence} If $\maxexprew{\MDP}{\MR}(\ms) < \infty$ for all $\ms \in \MI$, then:
        \begin{enumerate}
            \item$\forall \varepsilon > 0 \colon \exists \text{ MD strategy $\msched$}\colon$
            $
                \maxexprew{\MDP}{\MR}
                \geq
                (1-\varepsilon) \cdot \exprew{\msched}{\MDP}{\MR}
            $,
            \item If there is an optimal maximizing strategy, then there is an MD optimal maximizing strategy. 
        \end{enumerate}
    \end{enumerate}
\end{theorem}
\noindent
Note that the premise of (2) holds if $\MR$ is a bounded function.

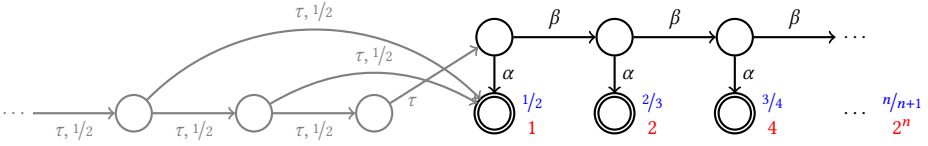
\begin{figure}[t]
    \begin{tikzpicture}[thick, on grid, node distance=10mm and 16mm, every state/.style={minimum size=6mm}, every node/.style={scale=0.8}]
        \node[gray] (-3) {$\ldots$};
        \node[state,right=of -3,gray] (-2) {};
        \node[state,right=of -2,gray] (-1) {};
        \node[state,right=of -1,gray] (0) {};
        \node[state,accepting,right=of 0,label={right,align=center:{$\blue{\nicefrac 1 2}$ \\ $\red{1}$}}] (1) {};
        \node[state,above=of 1] (1c) {};
        \node[state,accepting,right=of 1,label={right,align=center:{$\blue{\nicefrac 2 3}$ \\ $\red{2}$}}] (2) {};
        \node[state,above=of 2] (2c) {};
        \node[state,accepting,right=of 2,label={right,align=center:{$\blue{\nicefrac 3 4}$ \\ $\red{4}$}}] (3) {};
        \node[state,above=of 3] (3c) {};
        \node[right=of 3,label={right,align=center:{$\blue{\nicefrac{n}{n+1}}$ \\ $\red{2^n}$}}] (4) {$\ldots$};
        \node[above=of 4] (4c) {$\ldots$};
        
        \draw[->,gray] (-3) edge node[below] {$\labtau,\nicefrac 1 2$} (-2);
        \draw[->,gray] (-2) edge node[below] {$\labtau,\nicefrac 1 2$} (-1);
        \draw[->,gray] (-1) edge node[below] {$\labtau,\nicefrac 1 2$} (0);
        
        \draw[->,gray] (-2) edge[bend left=45] node[auto] {$\labtau,\nicefrac 1 2$} (1);
        \draw[->,gray] (-1) edge[bend left=] node[auto] {$\labtau,\nicefrac 1 2$} (1);
        
        \draw[->,gray] (0) edge[] node[below,near start] {$\labtau$} (1c);
        \draw[->] (1c) edge[] node[auto] {$\labbeta$} (2c);
        \draw[->] (2c) edge[] node[auto] {$\labbeta$} (3c);
        \draw[->] (3c) edge[] node[auto] {$\labbeta$} (4c);
        
        \draw[->] (1c) edge[] node[auto] {$\labalpha$} (1);
        \draw[->] (2c) edge[] node[auto] {$\labalpha$} (2);
        \draw[->] (3c) edge[] node[auto] {$\labalpha$} (3);
    \end{tikzpicture}
    \caption{
        Example MDP described by \citet{ornstein1969existence}.
        Transition probabilities equal to $1$ are omitted.
    }
    \label{fig:ornsteinExample}
\end{figure}

\section{Weakest Preexpectation Reasoning}
\label{sec:wp}

In this section, we introduce the program calculi we use throughout to reason about (minimal and maximal) expected outcomes of nondeterminstic probabilistic programs. Expectation-based reasoning for deterministic probabilistic programs was pioneered by Kozen \cite{Kozen1983,Kozen1985}. \citet{McIverM05} extended expectation-based reasoning to support programs with nondeterminism.

\subsection{Expectations}
Expectations\footnote{The terminology is slightly misleading: Expectations can be thought of as random variables on a program's state space rather than an expected value.} are the central objects the calculi considered in this paper operate on. They are the quantitative analogue of predicates: Instead of mapping program states to $\{\true, \false\}$,
program states are mapped to $\PosRealsInf = \PosReals \cup \{\infty\}$. Formally, the complete lattice $(\expecs, \eleq)$ of \emph{expectations} is
\[
    \expecs
    \eeq
    \{\, f \mid f \colon \States \to \NonnegRealsInf \,\} \qquad\text{where}\qquad
   f \eleq g \quad\text{iff}\quad \text{for all $\sigma \in \States$}\colon f(\sigma) \leq g(\sigma)~.
\]
Expectations are denoted by $f,g,\ldots$ and variations thereof. Infima and suprema in this lattice are thus understood pointwise. In particular, pairwise minima and maxima are given by 
\[
  f \emin g \eeq \lambda \sigma. \min \{ f(\sigma),g(\sigma)\}
  \qquad\text{and}\qquad
  f \emax g \eeq \lambda \sigma. \max \{ f(\sigma),g(\sigma)\}~.
\]
Standard arithmetic operations addition $+$ and multiplication $\cdot$ are also understood pointwise, i.e., for $\circ \in \{+,\cdot\}$, we let $f \circ g = \lambda \sigma. f(\sigma) \circ g(\sigma)$, where we set $0\cdot \infty = \infty \cdot 0 = 0$. The \emph{Iverson bracket} $\iverson{\varphi}$ casts a predicate $\varphi$ into an expectation~\cite{Iverson1962}:
\[
    [\guard]
    \eeq
    \lambda \sigma.
    \begin{cases}
    	1 & \text{ if } \sigma \models \guard  \\
    	0 & \text{ otherwise .}
    \end{cases}
\]
It is convenient to define a \emph{quantitative implication} $\gimpsymbol\colon \predicates \times \expecs \to \expecs$
 by
\[
    \gimp{\guard}{g}
    \eeq
    \lambda \sigma.
    \begin{cases}
    	g(\sigma) & \text{ if $\sigma\models \guard$ } \\
    	\infty & \text{otherwise}~.
    \end{cases}
    ~.
\]
Intuitively, $\gimpsymbol$ acts like a filter: if $\guard$ evaluates to $\true$, the implication evaluates to the value of the right-hand side's expectation. Otherwise, $\gimpsymbol$ evaluates to the top element $\infty$ of the lattice $(\expecs,\eleq)$. We agree on the following order of precedence for the connectives between expectations:
\[
   \cdot \qquad{}\textcolor{gray}{>}{}\qquad + \qquad{}\textcolor{gray}{>}{}\qquad \emin  \qquad{}\textcolor{gray}{>}{}\qquad \emax \qquad{}\textcolor{gray}{>}{}\qquad \gimpsymbol
\]
That is, $\cdot$ takes precedence over $+$, etc.
We use brackets to resolve ambiguities. Finally, given $f \in \expecs$, $x\in \Vars$, and an arithmetic expression $\aexpr$, we define the substitution of $x$ in $f$ by $\aexpr$ as
\[
    \expsubs{f}{x}{\aexpr} \eeq \lambda \sigma. f(\sigma\statesubst{x}{\aexpr(\sigma)}),\text{ where for $\val \in \Vals$, we set }\sigma\statesubst{x}{v}(y) \eeq \begin{cases}
    	v, & \text{if } y=x \\
    	\sigma(y), &\text{o.w.}
    \end{cases}
\]

\subsection{Angelic and Demonic Weakest Preexpectations}
\begin{figure}[t]
	\centering
	\begin{tabular}{l l l}
		\toprule
		$\cc$ & $\dwp{\cc}{f}$ & $\awp{\cc}{f}$ \\
		\midrule
		$\SKIP$ & $f$ & $f$\\[1ex]
		$\ASSIGN{x}{\aexpr}$ & $\expsubs{f}{x}{\aexpr}$ & \textcolor{black}{$\expsubs{f}{x}{\aexpr}$}  \\[1ex]
		$\COMPOSE{\cc_1}{\cc_2}$ & $\dwp{\cc_1}{\dwp{\cc_2}{f}}$ & $\awp{\cc_1}{\awp{\cc_2}{f}}$ \\[1ex]
		%
		$\GCFST{\guard_1}{\cc_1}$ & $ \phantom{\emin ~ }(\gimp{\guard_1}{\dwp{\cc_1}{f}})$ & $\phantom{\emax ~} [\guard_1] \cdot \awp{\cc_1}{f}$\\
		$\GCSEC{\guard_2}{\cc_2}$ & $\emin ~ (\gimp{\guard_2}{\dwp{\cc_2}{f}})$
		& $\emax ~  [\guard_2] \cdot \awp{\cc_2}{f}$ \\[1ex]
		$\PCHOICE{\cc_1}{\pexpr}{\cc_2}$ & $p \cdot \dwp{\cc_1}{f} + (1{-}\pexpr) \cdot \dwp{\cc_2}{f}$ & $\pexpr \cdot \black{\awp{\cc_1}{f}} + (1{-}\pexpr) \cdot \black{\awp{\cc_2}{f}}$\\[1ex]
		$\WHILEDO{\guard}{\cc'}$ & $\lfp g. ~ [\neg\guard]\cdot f + [\guard] \cdot \dwp{\cc'}{g}$ & $\lfp g. ~[\neg\guard]\cdot f + [\guard] \cdot \black{\awp{\cc'}{g}}$ \\
		\bottomrule
	\end{tabular}

	\caption{
		Inductive definition of $\dwp{\cc}{f}$ and $\awp{\cc}{f}$ for $f \in \expecs$. The $\lfp$ is taken w.r.t.\ $(\expecs,\eleq)$.
	}
	\label{fig:wp_table}
\end{figure}

To reason about minimal and maximal expected outcomes of programs, we introduce two program calculi --- \emph{expectation transformers} --- which associate to each $\cc \in \pgcl$ a map of type $\expecs \to \expecs$.  
\begin{definition}[Weakest preexpectation transformers]
    \label{def:wp}
    Let $\cc\in\pgcl$ and $f \in \expecs$. Each of the following is defined by induction on the structure of $\cc$ in \Cref{fig:wp_table}:
    \begin{enumerate}
    	\item $\dwp{\cc}{f} \in \expecs$ is the \emph{demonic weakest preexpectation} of $\cc$ w.r.t.\ postexpectation $f$.
    	\item $\awp{\cc}{f} \in \expecs$ is the  \emph{angelic weakest preexpectation} of $\cc$ w.r.t.\ postexpectation $f$.\hfill $\triangle$
    \end{enumerate}
\end{definition}
\noindent
The two transformers differ in the way they interpret nondeterminism: \emph{\textbf{d}emonically} vs.\ \emph{\textbf{a}ngelically}. If $\cc$ is deterministic, then $\dwp{\cc}{f}$ and $\awp{\cc}{f}$ coincide, in which case we often simply write $\wp{\cc}{f}$. Let us briefly go over the individual rules for $\SOMEWPSYMBOL\in \{\DWPSYMBOL,\AWPSYMBOL\}$.

For the effectless program $\SKIP$, $\somewp{\SKIP}{f}$ is just $f$. For assignments $\ASSIGN{x}{E}$, we substitute $x$ in $f$ by the assignment's right-hand side $E$. For sequentially composed programs $\COMPOSE{\cc_1}{\cc_2}$, we first determine the intermediate preexpectation $\somewp{\cc_2}{f}$, which is then plugged into $\SOMEWPSYMBOL\llbracket \cc_1\rrbracket$. The --- possibly nondeterministic --- guarded choice is treated in more detail below. For the probabilistic choice $\PCHOICE{\cc_1}{\pexpr}{\cc_2}$, we determine the preexpectations of the two branches and add them up, weighing each branch according to its probability of being executed. Finally, preexpectations of a loop are given by a least fixpoint which, intuitively, is the limit of all finite loop unrollings. We refer to \cite{DBLP:phd/dnb/Kaminski19} for an in-depth treatment of expectation-based reasoning.

Now, given $\cc \in \pgcl$ and $f \in \expecs$, each $\SOMEWPSYMBOL\in\{\DWPSYMBOL,\AWPSYMBOL\}$ defines an expectation $\somewp{\cc}{f}$, i.e., a map from program states $\sigma$ to the quantity $\somewp{\cc}{f}(\sigma)$. In what follows, we convey some intuition on this quantity.
Assume for the moment that $\cc$ is deterministic, i.e., $\cc$ possibly contains randomization but no nondeterminism.  In the absence of nondeterminism, $\dwp{\cc}{f}$ and $\awp{\cc}{f}$ coincide. Thinking of the postexpectation $f$ as a random variable over $\cc$'s state space, we have
\[
 \wp{\cc}{f}(\sigma) \eeq
\substack{\text{\normalsize\emph{expected value} of $f$ w.r.t.\ the (sub-)distribution of \emph{final} states} \\ \text{\normalsize reached after executing $\cc$ on \emph{initial} state $\sigma$.}}
\]
The distribution of final states is a \emph{sub}-distribution whenever $\cc$ does not terminate almost-surely on initial state $\sigma$, where the missing probability mass is the probability of divergence. 
As outlined above, and analogous to Dijkstra's weakest pre\emph{conditions} for standard programs \cite{DBLP:journals/cacm/Dijkstra75},  $\wp{\cc}{f}$ is obtained by recursively applying the rules from \Cref{fig:wp_table}, i.e., we start with the postexpectation $f$ \emph{at the end of the program} and --- as suggested by the rule for sequential composition --- move \emph{backwards} through $\cc$ until we arrive at the beginning, obtaining $\wp{\cc}{f}$. This is exemplified in \Cref{fig:example_dwp}, where we annotate\footnote{We slightly abuse notation and denote by $x$ the expectation $\lambda \sigma. \sigma(x)$.} the given program for determining $\wp{\cc}{x}$. Due to the backward-moving nature of the transformer, these annotations are best read from bottom to top. Complying with the above explanation, the result $\wp{\cc}{x} = y+3$ now tells us that the \emph{expected final value} of $x$ is given by the \emph{initial value} of $y$ plus $3$. This is intuitive: As $y$ is not initialized before it is read, the expected final value of $x$ depends on the initial value of $y$.
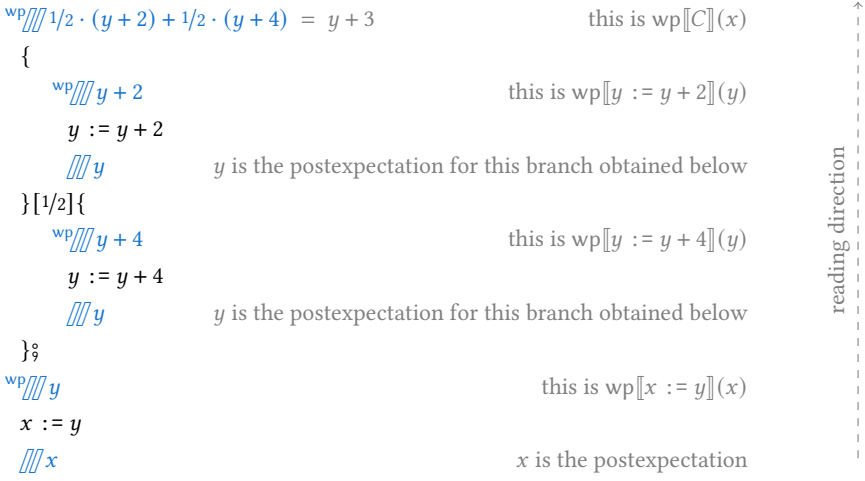
\begin{figure}[t]
	\small 
    \begin{minipage}{0.8\textwidth}
        \begin{align*}
        &\wpannotate{\nicefrac{1}{2} \cdot (y+2) + \nicefrac{1}{2} \cdot (y+4)} \textcolor{gray}{\eeq y+3} 
        \tag*{\textcolor{gray}{this is $\wp{\cc}{x}$}}\\
        &\{   \\
        &\qquad \wpannotate{y+2}  \tag*{\textcolor{gray}{this is $\wp{\ASSIGN{y}{y+2} }{y}$}}\\
        &\qquad \ASSIGN{y}{y+2}   \\
        &\qquad \annotate{y} \tag*{\textcolor{gray}{$y$ is the postexpectation for this branch obtained below}} \\
        &\}[\nicefrac{1}{2}]\{   \\
        &\qquad \wpannotate{y+4}  \tag*{\textcolor{gray}{this is $\wp{\ASSIGN{y}{y+4} }{y}$}}\\
        &\qquad \ASSIGN{y}{y+4} \\
        &\qquad \annotate{y}  \tag*{\textcolor{gray}{$y$ is the postexpectation for this branch obtained below}}\\
        &\}\fatsemi \\
        &\wpannotate{y} \tag*{\textcolor{gray}{this is $\wp{\ASSIGN{x}{y}}{x}$}} \\
        &\ASSIGN{x}{y} \\
        &\annotate{x}
        \tag*{\textcolor{gray}{$x$ is the postexpectation}}
        \end{align*}
    \end{minipage}
    \hspace{8mm}
    \begin{minipage}{0.1\textwidth}
        \begin{tikzpicture}
        \draw[->,gray,dashed] (0,0) -- node[sloped,above] {reading direction} (0,6);
        \end{tikzpicture}
    \end{minipage}
    
\caption{Example program with annotations for determining $\dwp{\cc}{x}$.}
\label{fig:example_dwp}
\end{figure}

Let us now consider a possibly nondeterministic program $\cc$. In this case, the final distribution of states obtained from executing $\cc$ on some initial state $\sigma$ might not be unique: it generally depends on the resolution of the nondeterminism. It does therefore no longer make sense to speak about \emph{the} expected value of $f$ w.r.t.\ the distribution of final states. Instead, we reason about \emph{optimal} values:
\[
\DWPSYMBOL/\awp{\cc}{f}(\sigma) \eeq
\substack{\text{\normalsize\emph{minimal/maximal} expected values of $f$ w.r.t.\ \emph{all} (sub-)distributions} \\ \text{\normalsize of \emph{final} states reached after executing $\cc$ on \emph{initial} state $\sigma$~.}}
\]
To see this, consider the preexpectation of a guarded choice in more detail: We have 
\begin{align*}
  &\dwp{\GC{\guard_1}{\cc_1}{\guard_2}{\cc_2}}{f}(\sigma) \eeq \min \{ \dwp{\cc_i}{f}(\sigma) ~\mid~ \sigma \models \guard_i \} \\
  \text{and}\quad&
  \awp{\GC{\guard_1}{\cc_1}{\guard_2}{\cc_2}}{f}(\sigma) \eeq \max \{ \awp{\cc_i}{f}(\sigma) ~\mid~ \sigma \models \guard_i \}~
\end{align*}
i.e, if both $\guard_1$ and $\guard_2$ evaluate to $\true$ under $\sigma$, then the above quantity is the minimum (resp.\ maximum) of the preexpectations of the two branches $\cc_1$ and $\cc_2$. Consider the program $\cc'$ shown in \Cref{fig:example_dwp_awp_nondet} as an example. It contains $\cc$ from \Cref{fig:example_dwp} as a subprogram. The annotations on the left-hand side determine $\dwp{\cc'}{x}$ --- mapping initial states to the \emph{minimal} expected final value of $x$ ---, and the right-hand side's annotations determine $\awp{\cc'}{x}$ --- mapping initial states to the \emph{maximal} expected final of $x$.  $\cc'$ first nondeterministically assigns either $1$ or $4$ to variable $y$, and then executes $\cc$. The minimal expected final value of $x$ is $4$, whereas the maximal one is $7$. \\

\begin{figure}[t]
	\small
	\begin{center}
	\begin{minipage}{0.45\textwidth}
		\begin{align*}
			&\dwpannotate{4 \emin 7} \textcolor{gray}{\eeq 4} \\
			& \GCFST{\true}{
				\singlelineannotatespace
				\dwpannotate{4}
				\singlelineannotatespace
				\ASSIGN{y}{1}
			   \hspace{0.25cm}
			   \annotate{y+3}
		   } \\
			&\GCSEC{\true}{
				\singlelineannotatespace
				\dwpannotate{7}
				\singlelineannotatespace
				\ASSIGN{y}{4}
				\hspace{0.25cm}
				\annotate{y+3}
			 } \fatsemi \\
			&\dwpannotate{y+3} \\
			& \cc \\
			&\annotate{x}
		\end{align*}
	\end{minipage}
	\hfill
	\begin{minipage}{0.45\textwidth}
		\begin{align*}
		&\awpannotate{4 \emax 7} \textcolor{gray}{\eeq 7} \\
		& \GCFST{\true}{
			\singlelineannotatespace
			\awpannotate{4}
			\singlelineannotatespace
			\ASSIGN{y}{1}
			\singlelineannotatespace
			\annotate{y+3}
		} \\
		&\GCSEC{\true}{
			\singlelineannotatespace
			\awpannotate{7}
			\singlelineannotatespace
			\ASSIGN{y}{4}
			\singlelineannotatespace
			\annotate{y+3}
		} \fatsemi \\
		&\awpannotate{y+3} \\
		& \cc \\
		&\annotate{x}
		\end{align*}
	\end{minipage}
\end{center}
\caption{$\DWPSYMBOL$ and $\AWPSYMBOL$ calculations w.r.t.\ postexpectation $x$. Here $\cc$ is the program from \Cref{fig:example_dwp}. Recall that $\dwp{\cc}{x} = \awp{\cc}{x} = \wp{\cc}{x}$ since $\cc$ is deterministic.}
\label{fig:example_dwp_awp_nondet}
\end{figure}

\subsection{MDP Semantics vs.\ Weakest Preexpectations}
\noindent
There is a tight connection between reachability-reward objectives in MDPs and weakest preexpectations. It is this tight connection which enables us to link our program-level strategy synthesis techniques to the well-known problem of synthesizing strategies in MDPs. More concretely, $\dwp{\cc}{f}(\sigma)$ and $\awp{\cc}{f}(\sigma)$ are the minimal and maximal \emph{expected rewards} in the MDP $\opMDP{\cc}$, respectively, where postexpectation $f$ induces the reward function: upon reaching a terminal configuration $(\Term,\sigma')$, a reward of $f(\sigma')$ is collected. Formally:

\begin{theorem}[\cite{operational_vs_weakest,qsl}]
	\label{thm:wp_vs_operational}
	Let $\cc \in \pgcl$, $f \in \expecs$ a post-expectation.
	Moreover, let $\MR_f \colon \{\Term\} \times \States \to \NonnegRealsInf$ be $\MR_f(\Term, \sigma') = f(\sigma')$. Then, for all $\sigma\in\States$,
	\begin{align*}
		\dwp{\cc}{f}(\sigma) \eeq&   \minexprew{\opMDP{\cc}}{\MR_f}(C,\sigma),\text{ and} \\
		%
		%
		\awp{\cc}{f}(\sigma) \eeq& \maxexprew{\opMDP{\cc}}{\MR_f}(C,\sigma)
		~.
	\end{align*}
   In particular, if $\cc$ is deterministic, then so is $\opMDP{\cc}$, and we have
   \[
      \wp{\cc}{f}(\sigma) \eeq   \exprewmc{\opMDP{\cc}}{\MR_f}(C,\sigma)~.
   \]
\end{theorem}
\noindent
The $\DWPSYMBOL$ and $\AWPSYMBOL$ calculi can thus be understood as a means to reason deductively about reachability-reward objectives of possibly infinite-state MDPs modeled by $\pgcl$ programs.

\section{A Bird's Eye View: Programmatic Strategy Synthesis}
\label{sec:overview}
Before we deal with the fully fledged formalization of our approach, we set the stage in this section by (i) formalizing our problem statement (\Cref{sec:problem_statement}), (ii) giving an informal description of our approach for \emph{loop-free} programs (\Cref{sec:overview_loopfree}), and (iii) demonstrating how these technique generalize to programs \emph{with} loops in \Cref{sec:overview_loops}.
All of this is done in an example-driven manner.

\subsection{Problem Statement}
\label{sec:problem_statement}
Recall the synthesis problem for MDPs $\MDP$ from \Cref{sec:mdp_objectives} of finding determinizations $\MDP'$ which guarantee bounds on the expected rewards of interest. If the state space of $\MDP$ is finite, it is known that the problem can be solved in polynomial time via linear programming~\cite{DBLP:books/wi/Puterman94}. However, this technique is in general not applicable to countably \emph{infinite-state} MDPs, which arise naturally from $\pgcl$ programs. Our goal is to obtain \emph{program-level strategy synthesis techniques} by means of weakest preexpectation reasoning. Towards lifting the synthesis problem to programs, we formalize the notions of \emph{implementations} and \emph{determinizations}.
 \begin{figure}[t]
 	\small
	\centering
	\begin{gather*}
		\frac{\cc_1' \refines \cc_1 \qquad \cc_2' \refines \cc_2}{\COMPOSE{\cc_1'}{\cc_2'} ~\refines~ \COMPOSE{\cc_1}{\cc_2}} 
		\qquad\qquad
		\frac{\guard_1' \entails \guard_1 \qquad \guard_2' \entails \guard_2 \qquad \entails \guard_1' \vee \guard_2' \qquad
			\cc_1' \refines \cc_1  \qquad \cc_2' \refines \cc_2}{\GC{\guard_1'}{\cc_1'}{\guard_2'}{\cc_2'} ~\refines~ \GC{\guard_1}{\cc_1}{\guard_2}{\cc_2}} \\
		\frac{\cc_1' \refines \cc_1 \qquad \cc_2' \refines \cc_2}{\PCHOICE{\cc_1'}{p}{\cc_2'} ~\refines~ \PCHOICE{\cc_1}{p}{\cc_2}}
		\qquad\qquad
		\frac{\cc' \refines \cc}{\WHILEDO{\guard}{\cc'} ~\refines~ \WHILEDO{\guard}{\cc}}
	\end{gather*}
	\caption{Rules defining the implements relation $\refines$. Here $\entails$ denotes \emph{entailment} between predicates, i.e., $\guard \entails \guard'$ if for all states $\sigma$ it holds that $\sigma \models \guard$ implies $\sigma \models \guard'$.}
	\label{fig:refinementrules}
\end{figure}
\begin{definition}
	\label{def:refinement}
	The \emph{implementation} relation ${}\refines{} \subseteq \pgcl \times \pgcl$ is the smallest partial order on $\pgcl$ satisfying the rules given in \Cref{fig:refinementrules}. If $\cc' \refines \cc$, then we say that $\cc'$ \emph{implements} $\cc$. If moreover $\cc'$ is deterministic, then we say that $\cc'$ is a \emph{determinization} of $\cc$.
	\defqed
\end{definition}
\noindent
If $\cc' \refines \cc$, then $\cc'$ and $\cc$ coincide syntactically up to the guards occurring in the guarded choices. The guards in $\cc'$ may be \emph{strengthened} to resolve some of the nondeterministic choices from $\cc$, which is formalized by the premises $\guard_1' \entails \guard_1$ and $\guard_2' \entails \guard_2$ in the rule for guarded choices. The premise $\entails \guard_1' \vee \guard_2'$ ensures that $\cc'$ does not eliminate \emph{all} choices from some guarded choice in $\cc$. 
\begin{example}
	\label{ex:implementations_determinizations}
	Consider the programs $\cc_1$ (left)  $\cc_2$ (middle), and $\cc$ (right).
	
	{\small 
	\begin{minipage}{0.3\textwidth}
		\begin{align*}
			& \GCFST{y \leq z}{\ASSIGN{x}{y}} \\
			&\GCSEC{y > z}{\ASSIGN{x}{z}}\fatsemi \\
			&\PCHOICE{\ASSIGN{x}{0}}{\nicefrac{1}{2}}{\ASSIGN{x}{2\cdot x}}
		\end{align*}
	\end{minipage}
	\hfill
	\begin{minipage}{0.3\textwidth}
		\begin{align*}
			& \GCFST{y \leq z}{\ASSIGN{x}{y}} \\
			&\GCSEC{y \geq z}{\ASSIGN{x}{z}}\fatsemi \\
			&\PCHOICE{\ASSIGN{x}{0}}{\nicefrac{1}{2}}{\ASSIGN{x}{2\cdot x}}
		\end{align*}
	\end{minipage}
	\hfill
	\begin{minipage}{0.3\textwidth}
		\begin{align*}
			& \GCFST{\true}{\ASSIGN{x}{y}} \\
			&\GCSEC{\true}{\ASSIGN{x}{z}}\fatsemi \\
			&\PCHOICE{\ASSIGN{x}{0}}{\nicefrac{1}{2}}{\ASSIGN{x}{2\cdot x}}
		\end{align*}
	\end{minipage}} \\

\noindent
   We have $\cc_1 \refines \cc_2 \refines \cc$. The strengthened guards $y \leq z$ and $y\geq z$ in $\cc_2$ resolve some of the nondeterminism from $\cc$. Program $\cc_2$ is, however, \emph{not} a determinization of $\cc$ since $\cc_2$ behaves nondeterministically whenever $z=y$ holds initially. Program $\cc_1$, on the other hand, \emph{is} a determinization of both $\cc_2$ and $\cc$ since the guards $y \leq z$ and $y > z$ are mutually exclusive.
   \hfill $\triangle$
\end{example}
\noindent
Since $\cc' \refines \cc$ possibly permits less nondeterministic choices than $\cc$, it is an immediate, yet important, characteristic of $\cc'$ that minimal (resp.\ maximal) expected outcomes of $\cc'$ only get larger (resp.\ smaller) when compared to expected outcomes of $\cc$. Formally:
\begin{restatable}{lemma}{refinementProp}
	\label{thm:refinementProp}
	Let $\cc, \cc' \in \pgcl$ and $f \in \expecs$. We have 
	\[
	\cc' \refines \cc
	\qqimplies
	\dwp{\cc}{f} \eleq \dwp{\cc'}{f}
	\qand        
	\awp{\cc}{f} \egeq \awp{\cc'}{f}
	~.
	\]
\end{restatable}
\begin{proof}
	See \iftoggle{arxiv}{\Cref{proof:refinementProp}}{\cite[Appendix A]{arxiv}}.
\end{proof}
\noindent
With the notions of determinizations at hand, we formalize our problem statement:\medskip

\noindent
\fbox{\parbox{\dimexpr\linewidth-2\fboxsep-2\fboxrule}{
    \begin{center}
    	\emph{Given $\cc \in \pgcl$, postexpectation $f \in \expecs$, ${}\someexprel{} \in \{\eleq,\egeq\}$, and a threshold $g \in \expecs$,} \\
    	\emph{if it exists, find a \emph{determinization} $\cc'$ of $\cc$ with $\wp{\cc'}{f} \someexprel g$.}
    \end{center}
}}%
\medskip

\noindent
Due to the tight connection between weakest preexpectations and reachability-rewards in MDPs (\Cref{thm:wp_vs_operational}), the above synthesis problem can be understood as a synthesis problem for \emph{MDPs}: Each program $\cc$ induces a countable MDP $\opMDP{\cc}$. If $\cc'$ is a determinization of $\cc$, then $\opMDP{\cc'}$ is a determinization of $\opMDP{\cc}$. Hence, from an MDP perspective, our problem statement reads:\medskip

\noindent
\fbox{\parbox{\dimexpr\linewidth-2\fboxsep-2\fboxrule}{
    \begin{center}
        \emph{Given $\cc \in \pgcl$, postexpectation $f\in\expecs$, ${}\somenumrel{} \in \{\leq,\geq\}$, and a threshold $g \in \expecs$,} \\
        \emph{if it exists, find $\cc'$ such that $\opMDP{\cc'}$ is a determinization of $\opMDP{\cc}$ \\ and for all $\sigma\in\States \colon \exprewmc{\opMDP{\cc'}}{\MR_f}(\cc',\sigma) \somenumrel g(\sigma)$.}
    \end{center}
}}%
\medskip

\subsection{First Step: Optimal Determinizations for Loop-Free Programs}
\label{sec:overview_loopfree}
\newcommand{\ccmin}{\cc_{\textnormal{min}}}
\newcommand{\ccmax}{\cc_{\textnormal{max}}}

Our first key insight is that, for \emph{loop-free} programs $\cc$, we can compute\footnote{Assuming that expectations are represented syntactically in some sufficiently expressive formal language, e.g., the one from~\cite{DBLP:journals/pacmpl/BatzKKM21} extended by $\sqcap$- and $\sqcup$ operators.} --- in a purely syntactic manner --- \emph{optimal} determinizations of $\cc$. Put more formally, given loop-free $\cc$ and postexpectation $f$, there are effectively constructible determinizations $\ccmin$ and $\ccmax$ of $\cc$ with 
\[
   \dwp{\cc}{f} \eeq \wp{\ccmin}{f} \qqand \awp{\cc}{f} \eeq \wp{\ccmax}{f}~.
\]
We exemplify the construction of $\ccmin$. The construction for $\ccmax$ is dual.
\begin{example}
\label{ex:loopFreeTrans}
Reconsider the nondeterministic program $\cc$ from \Cref{ex:implementations_determinizations} and fix the postexpectation $f \coloneqq x$. Below we give annotations for determining $\dwp{\cc}{x}$ --- the expectation which maps every initial state to the minimal expected final value of $x$:
   {\small \begin{align*}
   	&\dwpannotate{y \emin z} \\
   	& \GCFST{\true}{
   		\singlelineannotatespace\dwpalertannotate{y} \singlelineannotatespace
   		\ASSIGN{x}{y}
   	\singlelineannotatespace\annotate{x}
   } \\
   	&\GCSEC{\true}{\singlelineannotatespace\dwpalertannotate{z} \singlelineannotatespace
   		\ASSIGN{x}{z} 
   		\singlelineannotatespace\annotate{x}}\fatsemi \\
   	&\dwpannotate{\nicefrac 1 2 \cdot 0 + \nicefrac 1 2 \cdot 2 \cdot x \gray{\eeq x} } \\
   	&\PCHOICE{
        \singlelineannotatespace\dwpannotate{0} \singlelineannotatespace
        \ASSIGN{x}{0}
        \singlelineannotatespace\annotate{x}
    }{\nicefrac{1}{2}}{
        \singlelineannotatespace\dwpannotate{2 \cdot x} \singlelineannotatespace
        \ASSIGN{x}{2 \cdot x}
        \singlelineannotatespace\annotate{x}
    } \\
   	&\annotate{x}
   \end{align*}}
Since $\cc$ is loop-free, these annotations are obtained in a syntactic manner by recursively applying the rules given in \Cref{fig:wp_table}. The topmost annotation tells us that $\dwp{\cc}{x} = y \emin z$, i.e., the minimal expected \emph{final} value of $x$ is the minimum of the \emph{initial} values of $y$ and $z$. Towards constructing $\ccmin$, consider the highlighted intermediate preexpectations $\alertannocolor{y}$ and $\alertannocolor{z}$. These annotations tell us that the expected final value of $x$ will be $\alertannocolor{y}$ \emph{if $\cc$ executes the f\underline{irst} branch}, and that it will be $\alertannocolor{z}$ \emph{if $\cc$ executes the \underline{second} branch}. Hence, we can readily read off strengthenings of the guards to resolve the nondeterminism in an optimal way: If $\alertannocolor{y} < \alertannocolor{z}$ holds, the first branch should be taken. Conversely, if $\alertannocolor{y} > \alertannocolor{z}$ holds, the second branch should be taken. In case $\alertannocolor{y}=\alertannocolor{z}$ holds, both choices are optimal. We can thus construct the following implementation $\cc'$ of $\cc$:
{\small\begin{align*}
	\GC{y \leq z}{
		\ASSIGN{x}{y}
	}
	{y \geq z}{
		\ASSIGN{x}{z} 
	}\fatsemi
	\PCHOICE{\ASSIGN{x}{0}}{\nicefrac{1}{2}}{\ASSIGN{x}{2\cdot x}}
\end{align*}}%
Program $\cc'$ is still nondeterministic if $y=z$ holds initially. This nondeterminism can, however, be resolved arbitrarily in the sense that \emph{any} determinization of $\cc'$ will be optimal. We call $\cc'$ an \emph{optimal permissive determinization} of $\cc$ w.r.t. postexpectation $x$. $\cc'$ is now easily determinized by, e.g., turning one of the inequalities, say $y \leq z$, into a strict one, obtaining (one choice for) $\ccmin$:
{\small\begin{align*}
	\GC{y ~\red{<}~ z}{
		\ASSIGN{x}{y}
	} 
	{y \geq z}{
		\ASSIGN{x}{z} 
	}\fatsemi
	\PCHOICE{\ASSIGN x 0}{\nicefrac{1}{2}}{\ASSIGN{x}{2\cdot x}}
\end{align*}}%
The construction for $\AWPSYMBOL$ and $\ccmax$ is dual by flipping the inequalities. \hfill $\triangle$
\end{example}
\noindent
We formalize our construction of such optimal (permissive) determinizations for arbitrary \emph{loop-free} programs in \Cref{sec:trans}. Now reconsider our problem statement from \Cref{sec:problem_statement}:
\begin{center}
    \emph{Given $\cc \in \pgcl$, postexpectation $f$, ${}\someexprel{} \in \{\eleq,\egeq\}$, and a threshold $g \in \expecs$,} \\
    \emph{if it exists, find a \emph{determinization} $\cc'$ of $\cc$ with $\wp{\cc'}{f} \someexprel g$.}
\end{center}
In \Cref{sec:trans} further below, we show that the solution for loop-free $\cc$ is as follows:
\begin{itemize}
	\item  If $\someexprel$ is $\eleq$, then $\cc'$ exists if and only if $\dwp{\cc}{f} \eleq g$, in which case $\cc'$ is given by $\ccmin$.
\item If $\someexprel$ is $\egeq$, then $\cc'$ exists if and only if $\awp{\cc}{f} \egeq g$, in which case $\cc'$ is given by $\ccmax$.
\end{itemize}
That is, $\cc'$ exists iff the minimal (resp.\ maximal) expected final value of $f$ is upper (resp.\ lower) bounded by $g(\sigma)$ for every initial state $\sigma$. Moreover, $\cc'$ can be constructed as exemplified above.

\begin{remark}
	Our results do not imply decidability of the existence of the sought-after determinizations. Depending on the arithmetic necessary for expressing $f$, $\dwp{\cc}{f}$, $\awp{\cc}{f}$, or $g$, quantitative entailments of the form $\dwp{\cc}{f} \eleq g$ or $\awp{\cc}{f} \egeq g$ are often undecidable. \hfill $\triangle$
\end{remark}
\subsection{Second Step: From Quantitative Loop Invariants to Determinizations of Loops}
\label{sec:overview_loops}
Naturally, constructing determinizations of loops is more involved. Reasoning about minimal and maximal expected outcomes of loops requires reasoning about least fixpoints (cf.\ \Cref{fig:wp_table}), which are uncomputable in general~\cite{DBLP:journals/acta/KaminskiKM19}. How can we nonetheless determinize loops as asked for by our problem statement? Our primary observation is that
\begin{center}
 \emph{quantitative loop invariants yield determinizations of nondeterministic probabilistic loops}~.
 \end{center}
These quantitative loop invariants are generally hard to find, let alone algorithmically. Our key insight here is that once we have a quantitative loop invariant at hand, we can use it to find determinizations of loops. Our approach thus applies to any technique which verifies $\pgcl$ programs by means of quantitative loop invariant-based reasoning satisfying the assumptions formalized in \Cref{sec:trans}.
In what follows, we first introduce quantitative loop invariants, and then present an example of our construction for obtaining determinizations guided by these invariants.

\subsubsection{Quantitative Loop Invariants}
\label{sec:overview_invariants}
Analogous to classical Floyd-Hoare logic for standard programs, quantitative loop invariants establish bounds on preexpectations of a loop by reasoning inductively about \emph{one} arbitrary, but fixed, loop iteration. We introduce the following notions \cite{DBLP:phd/dnb/Kaminski19}:
\begin{definition}
	\label{def:invariants}
	Let $\cc = \WHILEDO{\guard}{\ccbody}$, $\inv,f \in \expecs$, and $\SOMEWPSYMBOL \in\{\DWPSYMBOL,\AWPSYMBOL\}$.
	\begin{enumerate}
		\item If $\iverson{\guard}\cdot \somewp{\ccbody}{\inv} + \iverson{\neg\guard}\cdot f \eleq \inv$, then we call $\inv$ a $\SOMEWPSYMBOL$-\emph{super}invariant of $\cc$ (w.r.t.\ $f$).
		\item If $\iverson{\guard}\cdot \somewp{\ccbody}{\inv} + \iverson{\neg\guard}\cdot f \egeq \inv$, then we call $\inv$ a $\SOMEWPSYMBOL$-\emph{sub}invariant of $\cc$ (w.r.t.\ $f$).\hfill $\triangle$
	\end{enumerate}
\end{definition}
\noindent
Hence, to determine whether $\inv$ is a super/sub-invariant of $\cc$, it suffices to determine $\somewp{\ccbody}{\inv}$ --- the preexpectation of the loop \emph{body} --- and to check whether the respective inequality holds. By \emph{Park induction}, every $\SOMEWPSYMBOL$-\emph{super}invariant $\inv$ of $\cc$ w.r.t.\ $f$ upper-bounds $\somewp{\WHILEDO{\guard}{\ccbody}}{f}$. 
\begin{theorem}[\cite{park1969fixpoint}]
	\label{thm:dwp_park}
	Let $\cc = \WHILEDO{\guard}{\ccbody}$, let $\inv,f \in \expecs$, and $\SOMEWPSYMBOL \in\{\DWPSYMBOL,\AWPSYMBOL\}$. If $\inv$ is a $\SOMEWPSYMBOL$-superinvariant of $\cc$ w.r.t.\ $f$, then $\somewp{\cc}{f} \eleq \inv$.
\end{theorem}
\begin{example}[Upper bounds from superinvariants]
	\label{ex:dwp_park}
	Consider the loop
	{\small \[
	\cc \quad {}={}\quad
	 \WHILEDO{c \neq 0}{
	   \underbrace{\PCHOICE{\ASSIGN{c}{0}}{\nicefrac{1}{2}}{\GC{\true}{\ASSIGN{x}{x^2}}{\true}{\ASSIGN{x}{x+1}}}}_{{}\eqqcolon{} \ccbody}
     }~.
	\]}%
	On every iteration, $\cc$ flips a fair coin and, depending on the outcome, either terminates by setting $c$ to $0$, or continues iterating after nondeterministically squaring $x$ or incrementing $x$ by one.
	
    Now suppose we wish to upper-bound $\dwp{\cc}{x}$ --- the expectation which maps every initial state to the minimal expected final value of $x$ --- by $\iverson{c \neq 0} \cdot (x+1) + \iverson{c = 0} \cdot x \eqqcolon \inv$. In words, if initially $c \neq 0$ holds, then the minimal expected final value is upper-bounded by the initial value of $x$ plus $1$. Conversely, if initially $c = 0$ holds, then $x$ is not modified at all. We employ \Cref{thm:dwp_park} to prove that $\dwp{\cc}{x} \eleq \inv$ indeed holds\footnote{This can be proven automatically using the deductive verifier \textsc{Caesar} \cite{heyvl}, see \iftoggle{arxiv}{\Cref{app:heyvl}}{\cite[Appendix D]{arxiv}}.} by verifying that $\inv$ is a $\DWPSYMBOL$-superinvariant of $\cc$ w.r.t.\ $x$:
	\begin{align*}
		 \iverson{c\neq 0} \cdot \dwp{\ccbody}{\inv} + \iverson{c=0}\cdot x 
		\eeq & \iverson{c \neq 0}\cdot \left( (x+1)  \emin \nicefrac 1 2 \cdot  (x^2+x+1) \right) + \iverson{c=0}\cdot x \\
		\eeleq & \iverson{c \neq 0}\cdot (x+1) + \iverson{c=0}\cdot x \eeq \inv~.
		\tag*{$\triangle$}
    \end{align*}
\end{example}
\noindent
Obtaining \emph{lower} bounds from \emph{sub}invariants requires additional conditions, see \Cref{sec:overview_lower_bounds}. 
\subsubsection{From Loop Invariants to Loop Determinization}
\label{sec:overview_inv_to_det}
How can we use the concept of superinvariants for finding determinizations as demanded by our problem statement? Let us, for the moment, focus on the conceptually simpler case of aiming for determinizations which guarantee \emph{upper} bounds on expected outcomes, i.e., let $\cc$ be a loop and $f,g\in\expecs$ and suppose we want to find a determinization $\cc'$ of $\cc$ with $\wp{\cc'}{f} \eleq g$. We proceed as follows: 
\begin{enumerate}
	\item Find a $\DWPSYMBOL$-superinvariant $\inv$ of $\cc$ w.r.t.\ $f$ such that $\inv$ satisfies $\inv \eleq g$, and
	\item \emph{compute} --- guided by $\inv$ --- a determinization $\cc'$ of $\cc$ such that $\inv$ \emph{as well} is a $\DWPSYMBOL$-superinvariant of $\cc'$ w.r.t.\ $f$, i.e., determinizing $\cc$ to $\cc'$ \emph{preserves superinvariance} of $\inv$ w.r.t. $f$. 
\end{enumerate}
Step (1) is undecidable in general.\tw{"Step (1) cannot be done algorithmically in general"} We assume that $\inv$ is provided by some external means. Step (2), on the other hand, is as \emph{syntactic} as our construction for loop-free programs from \Cref{sec:overview_loopfree}. Steps (1) and (2) together solve the given problem instance because
\begin{center}
	$\wp{\cc'}{f} 
	\qquad\substack{\textnormal{Thm.\ \ref{thm:dwp_park}} \\ {\text{\normalsize $\eleq$}}} \qquad
	\inv
	\qquad\substack{\textnormal{assumption} \\ {\text{\normalsize $\eleq$}}} \qquad g~.$
\end{center}
\begin{example}
	Reconsider the loop $\cc$ from \Cref{ex:dwp_park} with postexpectation $x$. Now fix expectation $x+1 \eqqcolon g$ and suppose we wish to find a determinization $\cc'$ of $\cc$ with $\wp{\cc'}{x} \eleq g$, i.e., the expected final value of $x$ shall be no greater than one plus the initial value of $x$ for all initial states. Recall from \Cref{ex:dwp_park} that expectation $\inv =  \iverson{c \neq 0} \cdot (x+1) + \iverson{c = 0} \cdot x$ is a $\DWPSYMBOL$-superinvariant of $\cc$ w.r.t.\ $x$. Moreover, we have $\inv \eleq g$. Hence, if we can determinize $\cc$ to $\cc'$ in such a way that $\inv$ remains a superinvariant of $\cc'$ w.r.t.\ $x$, we are done.
	
	For that, recall from \Cref{sec:overview_loopfree} that we can determinize \emph{loop-free} programs in an optimal permissive manner. In particular, for the loop-free body $\ccbody$ of $\cc$ we can construct $\ccbody'$ such that 
	\[\dwp{\ccbody}{\inv} \quad{}={}\quad \wp{\ccbody'}{\inv}~. \]
	Consequently, $\cc' = \WHILEDO{c \neq 0}{\ccbody'}$ will be a determinization of $\cc$ satisfying
	\begin{align*}
		&\iverson{c \neq 0} \cdot \dwp{\ccbody}{\inv} + \iverson{c=0} \cdot x
		~ \substack{\text{by construction} \\ {\text{\normalsize $=$}}} ~ 
		\iverson{c \neq 0} \cdot \wp{\ccbody'}{\inv} + \iverson{c=0} \cdot x
	    ~ \substack{\text{see Ex.\ \ref{ex:dwp_park}} \\ {\text{\normalsize $\eleq$}}}  ~ 
	    \inv
	\end{align*}
     which implies 
     \begin{align*}
     	\wp{\cc'}{x} \quad \substack{\text{Thm.\ \ref{thm:dwp_park}} \\ {\text{\normalsize $\eleq$}}} \quad \inv  \qquad \text{and therefore, by transitivity, also} \qquad \wp{\cc'}{x}  \quad
     	\eleq \quad g~.
     \end{align*}
    To construct $\ccbody'$, we proceed as in \Cref{sec:overview_loopfree}: Annotate $\ccbody$ for determining $\dwp{\ccbody}{\inv}$ and derive the strengthenings for the guards from the highlighted intermediate preexpectations:
    {\small\begin{align*}
    	  &\dwpannotate{ \ldots \text{ (not relevant)}}\\
    	  &\{ \singlelineannotatespace \dwpannotate{\ldots \text{ (not relevant)}}  \singlelineannotatespace
    	   \ASSIGN{c}{0} 
    	  \singlelineannotatespace \annotate{\inv}  \} \\
    	 &\mathrel{[\nicefrac{1}{2}]}{} \\
    	  &\{\qquad  \dwpannotate{\ldots \text{ (not relevant)}} \\
    	 &\qquad \GCFST{\true}{
    	 	\singlelineannotatespace
    	 	\dwpalertannotate{\iverson{c \neq 0} \cdot (x\cdot x+1) + \iverson{c = 0} \cdot x\cdot x}
    	 	\singlelineannotatespace
    	 	\ASSIGN{x}{x \cdot x} 
    	 	\singlelineannotatespace  \annotate{\inv}} \\
    	 &\qquad \GCSEC{\true}{
    	 	\singlelineannotatespace
    	 	\dwpalertannotate{\iverson{c \neq 0} \cdot ( x+2) + \iverson{c = 0} \cdot (x+1)}
    	 	\singlelineannotatespace
    	 	\ASSIGN{x}{x+1} 
    	 	\singlelineannotatespace \annotate{\inv}} \\
     	 &  \} \\
    	 & \annotate{\inv \coloneqq \iverson{c \neq 0} \cdot (x+1) + \iverson{c = 0} \cdot x} 
    \end{align*}}%
   Thus, whenever $\alertannocolor{\expsubs{\inv}{x}{x \cdot x}}(\sigma) < \alertannocolor{\expsubs{\inv}{x}{x +1}}(\sigma)$ holds in the current variable valuation $\pstate$, the first branch must be taken. Conversely, if $\alertannocolor{\expsubs{\inv}{x}{x \cdot x}}(\sigma) > \alertannocolor{\expsubs{\inv}{x}{x +1}}(\sigma)$ holds, the second branch must be taken. Finally, in case $\alertannocolor{\expsubs{\inv}{x}{x \cdot x}}(\sigma) = \alertannocolor{\expsubs{\inv}{x}{x +1}}(\sigma)$ holds, we may freely choose between the two branches. Now, for fixed state $\sigma$ and recalling that $\sigma(x) \geq 0$, we have 
   \begin{align*}
   	    \alertannocolor{\expsubs{\inv}{x}{x \cdot x}}(\sigma) \leq \alertannocolor{\expsubs{\inv}{x}{x +1}}(\sigma) 
   	   \qquad \text{iff}\qquad  \sigma(x)^2 \leq \sigma(x) + 1 
   	   \qquad \text{iff} \qquad  \sigma(x) \leq \nicefrac{1}{2} \cdot (1+\sqrt{5})~,
   \end{align*}
   which yields the \emph{correct-by-construction} determinization $\cc'$ of $\cc$ satisfying $\wp{\cc'}{x} \eleq g = x+1$:
   {\small \begin{align*}
   \WHILE{c \neq 0}  &  \ASSIGN{c}{0} \\
   	&\mathrel{[\nicefrac{1}{2}]}{} \\
   	&{\GC{x \leq \nicefrac{1}{2} \cdot (1+\sqrt{5})}{\ASSIGN{x}{x\cdot x}}{x > \nicefrac{1}{2} \cdot (1+\sqrt{5})}{\ASSIGN{x}{x+1}}}
   \} ~.
   \tag*{$\triangle$}
   \end{align*}}
\end{example}
\subsubsection{Subinvariants and Lower Bounds}
\label{sec:overview_lower_bounds}
We have so far disregarded the case where we want to find determinizations which guarantee \emph{lower} bounds on expected outcomes, i.e., for loop $\cc$ and $f,g,\in\expecs$, find a determinization $\cc'$ of $\cc$ with $g \eleq \wp{\cc'}{f}$. Compared to upper-bound reasoning, lower-bounding preexpectations of loops is more involved. $\SOMEWPSYMBOL$-\emph{sub}invariants $\inv$ of $\cc$ do \emph{not} necessarily lower-bound $\somewp{\cc}{f}$. This is not surprising as this is already the case for standard weakest pre\emph{conditions} for non-probabilistic programs, which are subsumed by weakest pre\emph{expectations}. In the realm of weakest preconditions, expectations are replaced by predicates and $\eleq$ corresponds to $\entails$, i.e., entailment between predicates. Proving $\varphi \entails \wp{\cc}{\psi}$ corresponds to establishing a \emph{total} correctness property of $\cc$, which generally requires additional \emph{side conditions} such as termination.
However, if we assume that we have such side conditions at hand so that
\begin{align*}
    \underbrace{\text{$\inv$ is $\SOMEWPSYMBOL$-subinvariant of $\cc$ w.r.t.\ $f$ $~\wedge~$ \emph{side conditions}}}_{\text{denote the validity of these verification conditions by $\sfsymbol{vc}\llbracket{\cc}\rrbracket(f) \in \{\true,\false\}$}}
    \qquad\text{implies}\qquad \somewp{\cc}{f} \eegeq \inv~,
\end{align*}
we may proceed in a way analogous to the above procedure for upper bounds:
\begin{enumerate}
	\item Find an $\AWPSYMBOL$-\emph{sub}invariant $\inv$ such that both $\sfsymbol{vc}\llbracket{\cc}\rrbracket(f)$ and $\inv \egeq g$ hold, and
	\item \emph{compute} --- guided by $\inv$ --- a determinization $\cc'$ of $\cc$ such that $\sfsymbol{vc}\llbracket{\cc'}\rrbracket(f)$ is satisfied \emph{as well}, i.e., determinizing $\cc$ into $\cc'$ \emph{preserves validity of the verification conditions}.
\end{enumerate}
These \emph{side conditions} come in different flavors\footnote{For instance, if both $\inv$ and $f$ are bounded by $1$, it suffices to show that $\cc$ is \emph{demonically almost-surely terminating} (cf.\ \Cref{thm:subinvcond}). If $\inv$ or $f$ are unbounded, more restrictions are needed \cite{DBLP:journals/pacmpl/HarkKGK20}.} and their restrictiveness typically depends on the expressive power of $\inv$ and the postexpectation $f$. Moreover, finding such side conditions is an active field of research~\cite{DBLP:journals/pacmpl/HarkKGK20}. This motivates our general framework presented in \Cref{sec:trans}, where we abstract from the specific side conditions under consideration. This framework may then be instantiated with verification conditions tailored to the specific problem instances under consideration. Finally, we remark that our framework also handles determinizations of \emph{arbitrary} $\pgcl$ programs --- possibly containing nested- and sequenced loops.
\subsection{On Completeness and Inherent Incompleteness of our Approach}
\label{sec:overview_incompleteness}
We show in \Cref{sec:trans_incompleteness} that our approach for upper bounds as outlined in \Cref{sec:overview_inv_to_det} is \emph{complete}:
For all $\pgcl$ programs $\cc$ and all postexpectations $f$, our framework can (in theory) produce a determinization $\cc' \refines \cc$ such that $\wp{\cc'}{f} = \dwp{\cc}{f}$.
In words, we can find a determinization that actually realizes the \emph{minimal} preexpectation.
This matches a known result about countably infinite MDP, namely that MD strategies for minimizing expected rewards independent of the initial state always exist (cf.\ \Cref{sec:impossibilityMDP}).
The situation is fundamentally different for the maximizing expectation transformer $\AWPSYMBOL$:
For some $\cc$ and $f$ there does \emph{not} exist a deterministic $\cc' \refines \cc$ such that $\wp{\cc'}{f} = \awp{\cc}{f}$.
A simple counterexample is the (non-probabilistic) loop
\[
    \cc \qeq \WHILEDO{c \neq 0}{\GC{\true}{\ASSIGN{n}{n+1}}{\true}{\ASSIGN c 0}}
\]
with postexpectation $f = \nicefrac{n}{n+1}$.
Assuming that $n \in \Nats$, this program realizes the MDP in \Cref{fig:ornsteinExample} (black states \& blue reward function).
It is clear that $\awp{\cc}{f} = [c \neq 0] \cdot 1 + [c = 0] \cdot \nicefrac{n}{n+1}$, but no determinization $\cc'$ of $\cc$ can actually terminate in a state where $n = 1$.
However, recall that under the mild assumptions of \Cref{thm:mdp_md_good}, existence of ($\epsilon$-)optimal determinizations is guaranteed.

\section{Reasoning with Verification Conditions}
\label{sec:vc}
\begin{figure}[t]
	\centering
	\begin{tabular}{l l}
		\toprule
		$\cc$ & $\vc{\cc}{f}$ \\
		\midrule
		$\SKIP$ & $\true$ \\[1ex]
		$\ASSIGN{x}{E}$ & $\true$\\[1ex]
		$\COMPOSE{\cc_1}{\cc_2}$ & $\vc{\cc_1}{\exptrans{\SOMEWPRESYMBOL}{\cc_2}{f}} \land \vc{\cc_2}{f}$ \\[1ex]
		$\GC{\guard_1}{\cc_1}{\guard_2}{\cc_2}$ & $\vc{\cc_1}{f} \land \vc{\cc_2}{f}$ \\[1ex]
		$\PCHOICE{\cc_1}{p}{\cc_2}$ & $\vc{\cc_1}{f} \land \vc{\cc_2}{f}$ \\[1ex]
		$\WHILEDOINV{\guard}{\cc'}{I}$ & $\vccond(\cc,f) \land \vc{\cc'}{I}$ \\
		\bottomrule
	\end{tabular}
	\caption{Inductive definition of the verification condition of $\cc$ w.r.t.\ $f$, $\vccond$, and $\SOMEWPSYMBOL$.}
	\label{fig:vc}
\end{figure}
Towards our framework for determinizing $\pgcl$ programs, this section introduces a simple verification condition generator adapted from \cite{DBLP:journals/scp/NavarroO22}. A distinguishing aspect of our formulation is that --- following the motivation from the previous section --- our verification condition generator is parameterized by the invariant-based proof rule for reasoning about loops. \\

\noindent
\emph{Auxiliary Expectation Transformers.} To simplify notation for reasoning about (possibly nested- or sequenced) loops, recall from \Cref{sec:syntax} that all loops occurring in a $\pgcl$ program $\cc$ are annotated with a quantitative loop invariant $\inv \in \expecs$, i.e., all loops in $C$ are of the form
%
   $\WHILEDOINV{\guard}{\cc'}{\inv}$.
%
%
We define for each $\SOMEWPSYMBOL\in\{\DWPSYMBOL,\AWPSYMBOL\}$ (and $\SOMEWPSYMBOL = \WPSYMBOL$ in case we deal with deterministic programs) an auxiliary expectation transformer $\exptransT{\SOMEWPRESYMBOL}{\cc}$. The inductive definition of $\somewpre{\cc}{f}$ is completely analogous to the definition of $\somewp{\cc}{f}$, \emph{except} for loops, for which we define
\[
   \somewpre{\WHILEDOINV{\guard}{\cc'}{\inv}}{f} ~{}\coloneqq{}~\inv~.
\]
In particular, for loop-free $\cc$, $\exptransT{\SOMEWPRESYMBOL}{\cc}$ and $\exptransT{\SOMEWPSYMBOL}{\cc}$ coincide.
As is standard in verification condition-based program verification \cite{DBLP:conf/lpar/Leino10}, $\exptransT{\SOMEWPRESYMBOL}{\cc}$ replaces the (generally uncomputable) least fixpoint by the (externally provided) annotated invariant so that determining $\somewpre{\cc}{f}$ reduces to syntactic reasoning. The idea is then to define suitable \emph{verification conditions} for $\cc$ and $f$ in such a way that the validity of these conditions implies that $\somewpre{\cc}{f}$ upper- or lower-bounds $\somewp{\cc}{f}$. \\

\noindent
\emph{The Parametric Verification Condition Generator.} In order to parameterize our verification condition generator by the invariant-based proof rule employed for approximating expected outcomes of loops, we consider objects $\vccond$ of type
%
    $ \{\cc \in \pgcl \,\mid\, \cc = \WHILEDOINV{\guard}{\cc'}{\inv}\} \times \expecs \to \Bools$
%
and call them \emph{verification condition providers} ($\VCSYMBOL$-providers, for short). The truth value $\vccond(\cc,f)$ indicates whether loop $\cc$ with invariant $\inv$ satisfies a verification condition w.r.t.\ postexpectation $f$. Now fix a $\VCSYMBOL$-provider $\vccond$ and some $\SOMEWPSYMBOL\in\{\DWPSYMBOL,\AWPSYMBOL\}$ for the remainder of this section.

\begin{definition}
	Let $\cc\in\pgcl$ and $f \in\expecs$. The \emph{verification condition} 
	%
	 $\vc{\cc}{f} \in \{\true,\false\} $
	%
	of $\cc$ w.r.t.\ $f$ (and $\vccond$ and $\SOMEWPSYMBOL$) is defined by induction on $\cc$ in \Cref{fig:vc}.\hfill $\triangle$
\end{definition}
\noindent
Intuitively, $\vc{\cc}{f}$ is $\true$ if and only if all loops occurring in $\cc$ satisfy the verification condition given by the $\VCSYMBOL$-provider $\vccond$. Notice that, for the sequential composition $\COMPOSE{\cc_1}{\cc_2}$, we employ the auxiliary transformer $\SOMEWPRESYMBOL$ on $\cc_2$ to determine the postexpectation for the $\VCSYMBOL$ of $\cc_1$ since we reason with the annotated loop invariants instead of least fixpoints.\jp{last half sentence is cryptic} The verification condition for the body $\cc'$ of a loop $\WHILEDOINV{\guard}{\cc'}{\inv}$ is taken w.r.t.\ postexpectation $\inv$ since, for nested loops, an inner loop's verification condition depends on the outer loop's invariant (see \Cref{ex:vc_nested_loop}).
We are interested in $\VCSYMBOL$-providers $\vccond$ so that validity of $\vc{\cc}{f}$ does indeed imply that $\somewpre{\cc}{f}$ bounds $\somewp{\cc}{f}$:
\begin{definition}\kb{this is a definition now, adapt this everywhere}
	\label{thm:bounding_vc}
	We say that $\vccond$ \emph{yields upper bounds} for $\SOMEWPSYMBOL \in\{\DWPSYMBOL,\AWPSYMBOL\}$, if 
	\[
	  \text{for all $\cc\in\pgcl$ and all $f \in \expecs$}\colon\quad
	   \vc{\cc}{f} ~{}\text{implies}{}~  \exptrans{\SOMEWPSYMBOL}{\cc}{f} \eleq \exptrans{\SOMEWPRESYMBOL}{\cc}{f} ~.
	\] 
	Analogously, we say that $\vccond$ \emph{yields lower bounds} for $\SOMEWPSYMBOL\in\{\DWPSYMBOL,\AWPSYMBOL\}$, if 
	\[
	\text{for all $\cc \in\pgcl$ and all $f \in \expecs$}\colon\quad
	\vc{\cc}{f} ~{}\text{implies}~{} \exptrans{\SOMEWPSYMBOL}{\cc}{f} \egeq \exptrans{\SOMEWPRESYMBOL}{\cc}{f} ~.\tag*{$\triangle$}
	\] 
\end{definition}
\noindent
It is useful to note that, since both $\exptransT{\DWPSYMBOL}{\cc}$ and $\exptransT{\AWPSYMBOL}{\cc}$ are monotonic (\iftoggle{arxiv}{\Cref{thm:healthiness}.\ref{thm:healthiness_monotonic}}{see~\cite[Appendix B]{arxiv}}), it actually suffices to show that $\vccond$ yields upper (resp.\ lower) bounds for loops in order to conclude that $\vccond$ yields upper (resp.\ lower) bounds for \emph{all} $\pgcl$ programs. See \iftoggle{arxiv}{\Cref{lem:vc_loops}}{\cite[Appendix B]{arxiv}} for details.

	\newcommand{\varinc}{inc}
\newcommand{\varp}{p}
    \begin{figure}[t]
    	\small
	\begin{align*}
		&\dwpreannotate{3y \emin z} \\
		&\COMPOSE{\ASSIGN{c}{1}}{\ASSIGN{x}{0}} \\
		&\dwpreannotate{x+ \iverson{c \neq 0} \cdot (3y \emin z)}  \tag{the $\DWPRESYMBOL$ of a loop is given by its invariant}\\
		& \WHILEDOINVOPEN{c \neq 0} \\
		& \qquad \dwpreannotate{ \expsubs{h}{p,inc}{\nicefrac{1}{4},y} \emin \expsubs{h}{p,inc}{\nicefrac{1}{2},z}}
		\tag{this is $\dwpre{\text{loop body}}{\inv}$}\\
		%
		%
		&\qquad \GCFST{\true}{
			\singlelinepreannotatespace \dwprealertannotate{\expsubs{h}{p,inc}{\nicefrac{1}{4},y} } \singlelineannotatespace
			\COMPOSE{\ASSIGN{\varinc}{y}}{\ASSIGN{\varp}{\nicefrac{1}{4}}}
			\singlelineannotatespace \annotate{h } \singlelineannotatespace
		} \\
		&\qquad \GCSEC{\true}{
			\singlelinepreannotatespace \dwprealertannotate{\expsubs{h}{p,inc}{\nicefrac{1}{2},z} } \singlelineannotatespace
			\COMPOSE{\ASSIGN{\varinc}{z}}{\ASSIGN{\varp}{\nicefrac{1}{2}}}
			\singlelineannotatespace \annotate{h} \singlelineannotatespace
		}\fatsemi \\
		&\qquad \dwpreannotate{\nicefrac{1}{2} \cdot \expsubs{\inv}{c}{0} + \nicefrac{1}{2} \cdot \expsubs{\inv'}{c'}{1} } 
		\tag{denote this expectation by $h$}\\
		&\qquad \{ \singlelinepreannotatespace \dwpreannotate{\expsubs{\inv}{c}{0} \singlelineannotatespace} \ASSIGN{c}{0} \singlelineannotatespace \annotate{\inv} \} \\
		&\qquad [\nicefrac{1}{2}] \\
		&\qquad \{
		\COMPOSE{ \singlelinepreannotatespace\dwpreannotate{\expsubs{\inv'}{c'}{1}} \singlelineannotatespace
			\ASSIGN{c'}{1}}{ \singlelinepreannotatespace\dwpreannotate{\inv'} \singlelineannotatespace
			\WHILEDOINV{c' \neq 0}{\PCHOICE{\ASSIGN{c'}{0}}{\varp}{\ASSIGN{x}{x+\varinc}}}{\inv'}
			\singlelineannotatespace \annotate{\inv}
		}
		\} \\
		&\qquad \annotate{\inv}  \tag{postexpectation for a loop's body is the loop's invariant}\\
		&\WHILEDOINVCLOSE{x+ \iverson{c \neq 0} \cdot (3y \emin z)} \tag{denote this invariant by $I$} \\
		&\annotate{x}
		\tag{$x$ is the postexpectation}
	\end{align*}
	
	\caption{Program $\cc$ with annotations for determining $\dwpre{\cc}{x}$. Invariant $\inv'$ of the inner loop is 
		%
		$I' \eeq \iverson{\varp>0}\cdot\big(\iverson{c'\neq0}\cdot\left(\frac{\varp \cdot x -\varp \cdot \varinc + \varinc}{\varp} \right) + \iverson{c'= 0}\cdot x + \iverson{c \neq 0} \cdot (3y \emin z)\big)$. Annotations for the inner loop are omitted. 
	}
	\label{fig:nested_vc}
\end{figure}
%
%
\begin{example}[An unbounded nested loop]
	\label{ex:vc_nested_loop}
	Recall from \Cref{thm:dwp_park} that $\DWPSYMBOL$-superinvariants establish upper bounds on demonic weakest preexpectations of loops. Using \iftoggle{arxiv}{\Cref{lem:vc_loops}}{\cite[Lemma B.2]{arxiv}}, it is thus straightforward to show that the $\VCSYMBOL$-provider $\superinvcond$ given by 
	\[ 
	\superinvcond(\WHILEDOINV{\guard}{\cc'}{\inv}, f) \eeq \true 
	\quad\qiff\quad
	 \iverson{\guard}\cdot\dwpre{\cc'}{\inv} + \iverson{\neg\guard} \cdot f \eeleq \inv
	\]
	yields upper bounds for $\DWPSYMBOL$. Notice that, for the loop body $\cc'$, we consider the auxiliary transformer $\DWPRESYMBOL$ to enable reasoning about nested loops. This comes at the cost of the typical dependency between inner- and outer invariants: $\dwpre{\cc'}{\inv}$ must not approximate $\dwp{\cc'}{\inv}$ too coarsely, i.e., if $\cc'$ itself contains loops, then the inner loop's superinvariants must be tight enough for the outer $\inv$ to satisfy the above superinvariant condition. 
	Now consider program $\cc$ from \Cref{fig:nested_vc}, which we annotate for determining $\dwpre{\cc}{x}$. Program $\cc$ contains a nested loop. On every iteration, the outer loop, which we denote by $\cc_1$, nondeterministically chooses a value for $\varinc$ and a probability for $\varp$. It then flips a fair coin and either terminates or executes the inner loop, which we denote by $\cc_2$. Loop $\cc_2$ depends on the choices made by $\cc_1$: It keeps flipping a coin with bias $\varp$ and either terminates or keeps incrementing $x$ by $\varinc$. 
	
	The annotations yield the conditions for $\vcsuperinvdwp{\cc}{x}$ to be $\true$. Recall that both loops must satisfy their respective verification conditions: The postexpectation for the outer loop $\cc_1$ is $x$, which yields the condition $\superinvcond(\cc_1, x)$. Moreover, the postexpectation for the inner loop $\cc_2$ is $\inv$ --- the invariant of the \emph{outer} loop ---, which yields the condition $\superinvcond(\cc_2, \inv)$. Both conditions can be shown to be $\true$. Since $\superinvcond$ yields upper bounds for $\DWPSYMBOL$, we thus get \[\dwp{\cc}{x} \eeleq \dwpre{\cc}{x} \eeq 3y \emin z~,\]
	 i.e., the minimal expected final value of $x$ is upper-bounded by the minimum of $3y$ and $z$. We will see in the next section how the validity of $\vcsuperinvdwp{\cc}{x}$ and the annotation of $\cc_1$'s body yield a correct-by-construction determinization of $\cc$ which preserves this bound. 
	 \hfill $\triangle$
\end{example}

\section{Determinizations of $\pgcl$ Programs from Weakest Preexpectations}
\label{sec:trans}
In this section, we formalize our approach for obtaining determinizations of $\pgcl$ programs from weakest preexpectations based on our parametric verification conditions from \Cref{sec:vc}.  \\

\noindent
First, recall from \Cref{sec:overview} that our approach is based on (pointwise) comparisons of intermediate preexpectations obtained from applying the rules in \Cref{fig:wp_table}. To formalize this, we define two functions $\predeleq, \predegeq \colon \expecs \times \expecs \to \predicates$ given by
\[
   \predeleq(f,g) = \lambda \sigma.~ 
   \begin{cases}
   	   \true, & \text{if $f(\sigma) \leq g(\sigma)$} \\
   	   \false, &\text{otherwise}
   	\end{cases}
   \qquad 
   \text{and}
   \qquad
   \predegeq(f,g) = \lambda \sigma.~ 
   \begin{cases}
   	\true, & \text{if $f(\sigma) \geq g(\sigma)$} \\
   	\false, &\text{otherwise}~.
   \end{cases}
\]
We often write $f \predeleq g$ and \ $f \predegeq g$ instead of $\predeleq(f,g)$ and $\predegeq(f,g)$. In words, $f \predeleq g$ (resp.\ $f \predegeq g$) yields a predicate which evaluates to $\true$ on state $\sigma$ iff $f(\sigma)$ is upper- (resp.\ lower-) bounded by $g(\sigma)$. These predicates form our sought-after strengthenings of guards to resolve nondeterminism.

\begin{figure}[t]
	\centering
	\begin{tabular}{l l}
		\toprule
		$\cc$ & $\trans{\comprel}{\SOMEWPSYMBOL}{\cc}{f}$\\
		\midrule
		$\SKIP$ & $\SKIP$ \\[1ex]
		$\ASSIGN{x}{\aexpr}$ & $\ASSIGN{x}{\aexpr}$\\[1ex]
		$\COMPOSE{\cc_1}{\cc_2}$ & $\COMPOSE{
			\trans{\comprel}{\SOMEWPSYMBOL}{\cc_1}{\exptransT{\SOMEWPRESYMBOL}{\cc_2}(f)}
		}{
			\trans{\comprel}{\SOMEWPSYMBOL}{\cc_2}{f}
		}$\\[1ex]
		$\GCFST{\guard_1}{\cc_1}$ & $\GCFST{
			\guard_1 \land (\guard_2 \limpl \exptransT{\SOMEWPRESYMBOL}{\cc_1}(f) \comprel \exptransT{\SOMEWPRESYMBOL}{\cc_2}(f))
		}{\trans{\comprel}{\SOMEWPSYMBOL}{\cc_1}{f}} $\\
		$\GCSEC{\guard_2}{\cc_2}$ & $\GCSEC{
			\guard_2 \land (\guard_1 \limpl \exptransT{\SOMEWPRESYMBOL}{\cc_2}(f) \comprel \exptransT{\SOMEWPRESYMBOL}{\cc_1}(f))
		}{\trans{\comprel}{\SOMEWPSYMBOL}{\cc_2}{f}}$\\[1ex]
		$\PCHOICE{\cc_1}{\pexpr}{\cc_2}$ & $\PCHOICE{\trans{\comprel}{\SOMEWPSYMBOL}{\cc_1}{f}}{\pexpr}{\trans{\comprel}{\SOMEWPSYMBOL}{\cc_2}{f}}$\\[1ex]
		%
		$\WHILEDOINV{\guard}{\cc'}{\inv}$ & $\WHILEDOINV{\guard}{\trans{\comprel}{\SOMEWPSYMBOL}{\cc'}{\inv}}{\inv}$ \\
		\bottomrule
	\end{tabular}
	\caption{Inductive definition of the program transformer $\TRANSSYMBOL$.}
	\label{fig:trans}
\end{figure}

Based on the above notion, we define the following expectation-based program transformation:
\begin{definition}
	\label{def:trans}
	Let $\comprel \in \{\predeleq, \predegeq\}$ and $\SOMEWPSYMBOL\in\{\DWPSYMBOL,\AWPSYMBOL\}$. The program transformer 
	\[
	    \transT{\comprel}{\SOMEWPSYMBOL} \colon \pgcl \times \expecs \to \pgcl
	\]
	is defined by induction on $\pgcl$ in \Cref{fig:trans}.\hfill$\triangle$
\end{definition}
\noindent
We often abbreviate $\trans{\predeleq}{\DWPSYMBOL}{\cc}{f}$ and $\trans{\predegeq}{\AWPSYMBOL}{\cc}{f}$ by $\dtrans{\SOMEWPSYMBOL}{\cc}{f}$ (called the \emph{demonic transformation} of $\cc$ w.r.t.\ $f$) and $\atrans{\SOMEWPSYMBOL}{\cc}{f}$ (called the \emph{angelic transformation} of $\cc$ w.r.t.\ $f$), respectively. Let us go over the rules in \Cref{fig:trans}. $\SKIP$ statements and assignments are never transformed. For the sequential composition $\COMPOSE{\cc_1}{\cc_2}$, we transform $\cc_2$ w.r.t.\ $f$ and $\cc_1$ w.r.t.\ to $\somewpre{\cc_2}{f}$. Notice that we employ the auxiliary transformer from \Cref{sec:vc} since our determinizations are guided by the annotated loop invariants. The guarded choice is the most interesting case. Recall from our informal description from \Cref{sec:overview} that we construct strengthenings for the guards $\guard_1$ and $\guard_2$ by comparing the intermediate preexpectations $\somewpre{\cc_1}{f}$ and $\somewpre{\cc_2}{f}$ of the two branches. Whenever \emph{both} $\guard_1$ and $\guard_2$ hold, there is nondeterminism that is to be resolved, which is realized by applying function $\comprel \in \{\predeleq, \predegeq\}$ to the intermediate preexpectations. Here, $\limpl$ is the standard implication between predicates, i.e, $\guard \limpl \guardb$ is $\false$ on state $\sigma$ iff $\sigma \models \guard$ and $\sigma \not\models \guardb$. We refer to \Cref{sec:overview_loopfree} for an illustrative example, where the corresponding intermediate preexpectations are highlighted. For the probabilistic choice, we simply transform the two respective branches. Finally, for loops, we transform the loop body w.r.t.\ the annotated loop invariant (cf.\ \Cref{ex:trans_nested}).\jp{add explanation why loop case does not depend on $f$} \\

\noindent
The above transformations yield implementations of the program that is to be transformed:
\begin{restatable}{theorem}{transrefines}
	\label{thm:trans_refines}
	Let $\comprel \in \{\predeleq, \predegeq\}$, $\SOMEWPSYMBOL\in\{\DWPSYMBOL,\AWPSYMBOL\}$, $\cc\in\pgcl$, $f\in\expecs$. We have $\trans{\comprel}{\SOMEWPSYMBOL}{\cc}{f} \refines \cc$.
\end{restatable}
\begin{proof}
	By induction on $\cc$. See \iftoggle{arxiv}{\Cref{proof:trans_refines}}{\cite[Appendix A]{arxiv}} for details.
\end{proof}
Hence, 
$\dwp{\cc}{f} \eleq \dwp{\dtrans{}{\cc}{f}}{f}$ and 
$\awp{\cc}{f} \egeq \awp{\atrans{}{\cc}{f}}{f}$,
%
hold by \Cref{thm:refinementProp}, i.e., naturally, resolving nondeterminism by applying the program transformations increases (decreases) minimal (maximal) expected outcomes of $\cc$. 

As demonstrated in \Cref{sec:overview_loopfree}, the program transformations generally yield programs which are still nondeterministic. This remaining nondeterminism can, however, be resolved in an arbitrary manner in the sense that $\dtrans{}{\cc}{f}$ and $\atrans{}{\cc}{f}$ yield \emph{optimal permissive determinizations} (cf.\ \Cref{sec:overview_loopfree}) w.r.t.\ the auxiliary transformers $\DWPRESYMBOL$ and $\AWPRESYMBOL$:
\begin{restatable}{theorem}{transpreserveswpre}
	\label{thm:trans_preserves_wpre}
	Let $\cc,\cc'\in\pgcl$ and  $f \in \expecs$. We have
	\[
	   \cc' \refines \dtrans{\SOMEWPSYMBOL}{\cc}{f}
	   \qquad\text{implies}\qquad
	  \dwpre{\cc'}{f}  \eeq \dwpre{\cc}{f}  ~.
	\]
	Moreover, we have
	\[
	  \cc' \refines \atrans{\SOMEWPSYMBOL}{\cc}{f}
	  \qquad\text{implies}\qquad
	  \awpre{\cc'}{f}  \eeq \awpre{\cc}{f}  ~.
	\]
\end{restatable}
\begin{proof}
	 By induction on $\cc$. See \iftoggle{arxiv}{\Cref{proof:trans_preserves_wpre}}{\cite[Appendix A]{arxiv}} for details.
\end{proof}
\noindent
In particular, if $\cc$ is loop-free then $\dwpre{\cc}{f}$ and $\dwp{\cc}{f}$ (resp.\ $\awpre{\cc}{f}$ and $\awp{\cc}{f}$) coincide, so the above theorem indeed yields that our transformation is \emph{optimal and correct for loop-free programs}. In the presence of loops, we need to ensure that the respective annotated invariants soundly bound the sought-after preexpectations. For that, recall from \Cref{sec:vc} that we introduced $\VCSYMBOL$-providers $\vccond$ and from \Cref{sec:overview_loops} that our idea is to construct determinizations which preserve validity of the so-defined verification conditions. This motivates the following definition:
\begin{definition}
	We say that $\dtransT{\SOMEWPSYMBOL}$ \emph{preserves} $\vccond$ if for all $\cc,\cc' \in\pgcl$ and all $f \in \expecs$, 
		\[
		 \vcdwp{\cc}{f}~\text{and}~ \cc' \refines \dtrans{}{\cc}{f}
		  \quad\text{implies}\quad \vcdwp{\cc'}{f}~. 
		 \]
	Moreover, we say that $\atransT{\SOMEWPSYMBOL}$ \emph{preserves} $\vccond$ if for all $\cc,\cc' \in\pgcl$ and all $f \in \expecs$, 
	\[
	\vcawp{\cc}{f}~\text{and}~ \cc' \refines \atrans{}{\cc}{f}
	\quad\text{implies}\quad \vcawp{\cc'}{f}~.
    \tag*{$\triangle$}
	\]
\end{definition}
\noindent
In words, all implementations $\cc'$ of $\dtrans{}{\cc}{f}$ (including all determinizations and $\dtrans{}{\cc}{f}$ itself) must preserve validity of the verification condition given by the $\VCSYMBOL$-provider $\vccond$ (analogously for $\atrans{}{\cc}{f}$). Hence, if $\vccond$ yields upper bounds for $\DWPSYMBOL$ (resp.\ lower bounds for $\AWPSYMBOL$), then the respective program transformations \emph{preserve} the bounds obtained from the annotated loop invariants, i.e., $\dwpre{\cc}{f}$ and $\awpre{\cc}{f}$, respectively, which is the main \mbox{result of this section:}

\begin{theorem}
	\label{thm:correct_transformations}
	Assume that $\dtransT{\SOMEWPSYMBOL}$ \emph{preserves} $\vccond$. If $\vccond$ yields upper bounds for $\DWPSYMBOL$, then for all $\cc,\cc' \in \pgcl$ and all $f\in\expecs$,
	\[
	    \vcdwp{\cc}{f}~\text{and}~ \cc' \refines \dtrans{\SOMEWPSYMBOL}{\cc}{f} \qquad\text{implies}\qquad  
	   \dwp{\cc}{f} \eeleq \exptrans{\DWPSYMBOL}{\cc'}{f}
	    \eeleq \exptrans{\DWPRESYMBOL}{\cc}{f}  ~.
	\]
	Dually, assume that $\atransT{\SOMEWPSYMBOL}$ \emph{preserves} $\vccond$. If $\vccond$ yields lower bounds for $\AWPSYMBOL$, then for all $\cc,\cc' \in \pgcl$ and all $f\in\expecs$,
	\[	 
	 \vcawp{\cc}{f}~\text{and}~ \cc' \refines \atrans{\SOMEWPSYMBOL}{\cc}{f}  \qquad\text{implies}\qquad  
	  \awp{\cc}{f} \eegeq \exptrans{\AWPSYMBOL}{\cc'}{f}
	 \eegeq \exptrans{\AWPRESYMBOL}{\cc}{f} 
	 ~.
	\]
\end{theorem}

\begin{proof}\let\qed\relax
	We prove the claim for $\DWPSYMBOL$. The reasoning for $\AWPSYMBOL$ is analogous. First, $\dwp{\cc}{f} \eleq \dwp{\cc'}{f}$
	holds by \Cref{thm:trans_refines} and \Cref{thm:refinementProp}. Moreover, since  $\dtransT{\SOMEWPSYMBOL}$ preserves $\vccond$, we have $\vcdwp{\cc'}{f}$. Hence, since $\vccond$ yields upper bounds for $\DWPSYMBOL$, we get
	\begin{align*}
		\dwp{\cc'}{f} \stackrel{\text{Def.\ \ref{thm:bounding_vc}}}{\eeleq}
		\dwpre{\cc'}{f} \stackrel{\text{Thm.\ \ref{thm:trans_preserves_wpre}}}{\eeq}
		\dwpre{\cc}{f}~.\tag*{$\square$}
	\end{align*}
\end{proof}
\begin{example}
	\label{ex:trans_nested}
	Recall the $\VCSYMBOL$-provider $\superinvcond$ from \Cref{ex:vc_nested_loop}, which yields upper bounds for $\DWPSYMBOL$. Using \Cref{thm:trans_preserves_wpre}, it is immediate that $\dtransT{\SOMEWPSYMBOL}$ preserves $\superinvcond$. Now reconsider program $\cc$ from \Cref{fig:nested_vc}. Since $\DWPRESYMBOL$ is obtained using the annotated loop invariants, the highlighted intermediate preexpectations are obtained in a purely syntactic manner. Applying minor arithmetic simplifications to the guard strengthenings obtained from these highlighted preexpectations, we thus also obtain $\dtrans{}{\cc}{x}$ in a purely syntactic manner:
	{\small
		\begin{align*}
		%
		&\COMPOSE{\ASSIGN{c}{1}}{\ASSIGN{x}{0}} \\
		%
		%
		& \WHILEDOINVOPEN{c \neq 0} \\
		%
		%
		%
		&\qquad \GCFST{6y \leq z}{
			\COMPOSE{\ASSIGN{\varinc}{y}}{\ASSIGN{\varp}{\nicefrac{1}{4}}}
		} \\
		&\qquad \GCSEC{6y \geq z}{
			\COMPOSE{\ASSIGN{\varinc}{z}}{\ASSIGN{\varp}{\nicefrac{1}{2}}}
		}\fatsemi \\
		%
		%
		&\qquad \{  \ASSIGN{c}{0}  \} \\
		&\qquad [\nicefrac{1}{2}] \\
		&\qquad \{
		\COMPOSE{ 
			\ASSIGN{c'}{1}}{ 
			\WHILEDOINV{c' \neq 0}{\PCHOICE{\ASSIGN{c'}{0}}{\varp}{\ASSIGN{x}{x+\varinc}}}{\inv'}
		}
		\} \\
		%
		%
		&\WHILEDOINVCLOSE{x+ \iverson{c \neq 0} \cdot (3y \emin z)} \\
	\end{align*}}
	The above nested loop still behaves nondeterministically whenever $6y = z$ holds at the beginning of an outer loop iteration. However, thanks to \Cref{thm:correct_transformations}, \emph{any} determinization $\cc'$ of $\dtrans{}{\cc}{x}$ (obtained by, e.g., turning the first inequality $6y \leq z$ into a strict one) satisfies
	\[
	   \dwp{\cc}{x} \eeleq \dwp{\cc'}{x} \eeq \dwpre{\cc}{x} \eeq 3y \emin z,
	\]
	i.e, the expected final value of $x$ of \emph{any} determinization $\cc'$ is upper-bounded by $3y \emin z$.
	   \hfill $\triangle$
\end{example}

\subsection{Completeness}
\label{sec:trans_incompleteness}

Recall from \Cref{sec:overview_incompleteness} that our approach is inherently incomplete when it comes to finding determinizations which guarantee lower bounds on expected outcomes of $\pgcl$ programs. Yet, \Cref{thm:correct_transformations} yields a sufficient condition for our approach to be complete for \emph{subclasses} of problem instances. Let $\progclass \subseteq \pgcl$ and $\expclass \subseteq \expecs$. We call a $\VCSYMBOL$-provider $\vccond$ \emph{demonically (resp.\ angelically) complete} w.r.t.\ $(\progclass, \expclass)$ if, (i) $\VCSYMBOL$ yields upper (resp.\ lower) bounds for $\DWPSYMBOL$ (resp.\ $\AWPSYMBOL$) and (ii) for all $\cc \in \progclass$ and $f \in \expclass$, program $\cc$ can be annotated with invariants in such a way that $\vcdwp{\cc}{f}$ (resp. $\vcawp{\cc}{f}$) holds and $\dwp{\cc}{f} = \dwpre{\cc}{f}$ (resp.\ $\awp{\cc}{f} = \awpre{\cc}{f}$). 
\begin{theorem}
	\label{thm:completeness}
	Let $\vccond$ be demonically complete w.r.t.\ $(\progclass, \expclass)$ and preserved by $\dtransT{}$. Then, for all $\cc\in\progclass$ and $f \in \expclass$, there exist invariant annotations for $\cc$ such that for all $\cc'\in\pgcl$,
	\[
	    \cc' \refines \dtrans{}{\cc}{f} \qquad \text{implies} \qquad \dwp{\cc}{f} \eeq \dwp{\cc'}{f}~.
	\]
	Analogously, let $\vccond$ be angelically complete w.r.t.\ $(\progclass, \expclass)$ and preserved by $\atransT{}$. Then, for every $\cc\in\progclass$ and $f \in \expclass$, there exist invariant annotations for $\cc$ such that for all $\cc'\in\pgcl$,
	\[
	\cc' \refines \atrans{}{\cc}{f} \qquad \text{implies} \qquad \awp{\cc}{f} \eeq \awp{\cc'}{f}~.
	\]
\end{theorem}
\noindent
As mentioned in \Cref{sec:overview_incompleteness}, \Cref{thm:completeness} immediately yields that our approach is complete w.r.t.\ \emph{upper} bounds: the $\VCSYMBOL$-provider $\superinvcond$ from Examples \ref{ex:vc_nested_loop} and \ref{ex:trans_nested} is demonically complete w.r.t.\ $(\pgcl,\expecs)$ by annotating each loop occurring in some program $\cc$ with its respective least fixpoint (cf.\ \Cref{fig:wp_table}). We provide a complete subclass for \emph{lower} bounds in the next section.

\subsection{Instances of our Framework}
\label{sec:trans_instances}
\subsubsection{Upper-Bounded Determinizations}
We restate the $\VCSYMBOL$-provider $\superinvcond$ for later reference.
\begin{theorem}
	\label{thm:superinvcond}
	The $\VCSYMBOL$-provider $\superinvcond$ yields upper bounds for $\DWPSYMBOL$ and is preserved by $\dtransT{}$:
	\[ 
	\superinvcond(\WHILEDOINV{\guard}{\ccbody}{\inv}, f) \eeq \true 
	\quad\qiff\quad
	\iverson{\guard}\cdot\dwpre{\ccbody}{\inv} + \iverson{\neg\guard} \cdot f \eeleq \inv
	\]
\end{theorem} 
\noindent
We briefly discuss additional instances employing more advanced proof rules in \Cref{sec:conclusion}.\kb{maybe omit this depending on space constraints}
\subsubsection{Lower-Bounded Determinizations}
Lower-bounding weakest preexpectations of loops is more involved. There exist specialized proof rules, which typically restrict the expressive power of the invariants, the postexpectation, and the termination behavior of the loop under consideration. \\

\noindent
We call program $\cc$ \emph{demonically almost-surely terminating} (dAST, for short) if $\dwp{\cc}{1} = 1$. Moreover, we call $\cc$ \emph{positively demonically almost-surely terminating} (dPAST, for short) if $\cc$'s expected runtime (in the sense of \cite{ert}) is finite for all initial states. Finally, we define a loop $\WHILEDOINV{\guard}{\ccbody}{\inv}$ to be \emph{suitable for optional stopping w.r.t.\ $f$} in \iftoggle{arxiv}{\Cref{app:optional_stopping}}{\cite[Appendix B]{arxiv}}. \\

\noindent
With these notions at hand, we obtain instances for lower-bounded determinizations:
\begin{theorem}
	\label{thm:subinvcond}
	The following $\VCSYMBOL$-providers yield lower bounds for $\AWPSYMBOL$ and are preserved by $\atransT{}$:
	\begin{enumerate}
		\item \label{thm:dastsubinvcond} $\dastsubinvcond$ which for $C = \WHILEDOINV{\guard}{\ccbody}{\inv}$ and $f\in\expecs$ is given by
		\[
		    \dastsubinvcond(C, f) = \true 
		    \qquad\text{iff}\qquad
		    \substack{\text{\normalsize $\iverson{\guard}\cdot \awpre{\ccbody}{\inv} + \iverson{\neg\guard}\cdot f \egeq \inv$, $\cc$ is dAST,} \\
		    \text{\normalsize  $\inv$ and $f$ are bounded by some constant $b$ .}}
		\]
		\item \label{thm:dpastsubinvcond} $\dpastsubinvcond$ which for $C = \WHILEDOINV{\guard}{\ccbody}{\inv}$ and $f\in\expecs$ is given by
		\[
		\dpastsubinvcond(C, f) = \true 
		\qquad\text{iff}\qquad
		\substack{\text{\normalsize $\iverson{\guard}\cdot \awpre{\ccbody}{\inv} + \iverson{\neg\guard}\cdot f \egeq \inv$, $\cc$ is dPAST,} \\
			\text{\normalsize $\cc$ is suitable for optional stopping w.r.t.\ $f$ .}}
		\]
	\end{enumerate}
\end{theorem}
\begin{proof}
	Preservation by $\atransT{}$ follows immediately from \Cref{thm:refinementProp} and \Cref{thm:trans_preserves_wpre} and the fact that implementations $\cc'$ of $\cc$ remain dAST and dPAST, respectively. The fact that the providers yield lower bounds for $\AWPSYMBOL$ follows from \cite[Theorem 38]{DBLP:journals/pacmpl/HarkKGK20} (see \Cref{rem:aiming_low}).
\end{proof}
\noindent
The $\VCSYMBOL$-provider $\dastsubinvcond$ is angelically complete w.r.t.\ 
\[\big( \{\cc \in \pgcl ~\mid~ \text{$\cc$ is dAST} \}, \{f \in \expecs ~\mid~ \text{$f$ bounded by some $b\in\PosReals$}\} \big) \]
by annotating every loop in a given program $\cc$ by its respective least fixpoint, which is also bounded by \iftoggle{arxiv}{\Cref{thm:healthiness}.\ref{thm:healthiness_feasible}}{\cite[Theorem B.1.2]{arxiv}}. \Cref{thm:completeness} thus yields completeness of our approach for this class of problem instances. Determining the class for which $\dpastsubinvcond$ is complete is an open problem. Notice that $\dpastsubinvcond$ does not require $\inv$ or $f$ to be bounded which comes at the cost of requiring the stronger termination criterion dPAST and being suitable for optional stopping. 

\subsubsection{Lower Bounds for $\AWPSYMBOL$}
	\label{rem:aiming_low}
	The fact that $\dpastsubinvcond$ yields lower bounds for $\AWPSYMBOL$ (cf.\ \Cref{thm:subinvcond}.\ref{thm:dpastsubinvcond}), follows from a result from the literature which applies to \emph{deterministic} $\pgcl$ programs \cite[Theorem 38]{DBLP:journals/pacmpl/HarkKGK20}. \citet{DBLP:journals/pacmpl/HarkKGK20} leave an extension of this rule for \emph{non}deterministic programs for future work. Our results from \Cref{sec:trans} yield such an extension, i.e., a novel proof rule for lower bounds on $\AWPSYMBOL$'s if possibly nondeterministic loops. First, consider non-nested loops:

	\begin{theorem}
		\label{thm:aiming_low_nondet_loopfree}
		Let $\cc = \WHILEDOINV{\guard}{\ccbody}{\inv}$ and $f \in \expecs$ with $\ccbody$ loop-free. If
		\begin{enumerate}
			\item\label{thm:aiming_low_nondet_loopfree1} $\iverson{\guard}\cdot \awp{\ccbody}{\inv} + \iverson{\neg\guard}\cdot f \egeq \inv$, and
			\item\label{thm:aiming_low_nondet_loopfree2} $\cc$ is dPAST, and
			\item\label{thm:aiming_low_nondet_loopfree3} $\cc$ is suitable for optional for stopping w.r.t.\ $f$ (cf.\ \iftoggle{arxiv}{\Cref{app:optional_stopping}}{\cite[Appendix B]{arxiv}}), 
		\end{enumerate}
	then $\inv$ lower-bounds $\awp{\cc}{f}$, i.e.,
	\[
			\awp{\cc}{f} \eegeq \inv~.
	\]
	\end{theorem}
\begin{proof}
	First notice that the conjunction of Conditions \ref{thm:aiming_low_nondet_loopfree1}, \ref{thm:aiming_low_nondet_loopfree2}, and \ref{thm:aiming_low_nondet_loopfree3} from above is equivalent to $\dpastsubinvcond(C, f) = \true$ (cf.\ \Cref{thm:subinvcond}.\ref{thm:dpastsubinvcond}).
	Since $\atransT{}$ preserves $\dpastsubinvcond$, there is a determinization $\cc'$ of $\cc$ satisfying $\dpastsubinvcond(\cc', f)$. Hence \cite[Theorem 38]{DBLP:journals/pacmpl/HarkKGK20} \emph{does} apply to this deterministic program $\cc'$ and we get 
	\begin{align*}
	\awp{\cc}{f} \quad\substack{{\normalsize \text{Lem.\ \ref{thm:refinementProp}}} \\ {\normalsize \egeq}}\quad \wp{\cc'}{f} 
	\quad\substack{{\normalsize \text{\cite[Theorem 38]{DBLP:journals/pacmpl/HarkKGK20}}} \\ {\normalsize\eegeq}}\quad
	\inv
	~.
\end{align*}
\end{proof}
\begin{example}
	Let $f = x$ and $\inv = \iverson{c=1}\cdot(x+2) + \iverson{c \neq 1}\cdot x$ and let $\cc$ be the loop
	\[
		\WHILEDOINV{c= 1}{\PCHOICE{\ASSIGN{c}{0}
			}{0.5}{
		\GC{\true}{\ASSIGN{x}{x+1}}{\true}{\ASSIGN{x}{x+2}}
		}}{\inv}~.
	\]
	We employ \Cref{thm:aiming_low_nondet_loopfree} to prove that $\awp{\cc}{x} \egeq \inv$.
	Regarding Conditions \ref{thm:aiming_low_nondet_loopfree3} and \ref{thm:aiming_low_nondet_loopfree2}, it is easy to verify that $\cc$ is dPAST (using, e.g., the $\textsf{ert}$-calculus \cite{ert}) and suitable for optional stopping w.r.t.\ to $x$. For Condition \ref{thm:aiming_low_nondet_loopfree1}, we calculate
	\begin{align*}
		& \iverson{c= 1}\cdot \awp{\ccbody}{\inv} + \iverson{c \neq 1}\cdot x \\
		\eeq &\iverson{c= 1}\cdot 0.5 \cdot \big(\expsubs{\inv}{c}{0} + (\expsubs{\inv}{x}{x+1} \sqcup \expsubs{\inv}{x}{x+2} )\big)+ \iverson{c \neq 1}\cdot x \\
		\eeq &\iverson{c= 1}\cdot 0.5 \cdot \big(x + ((x+3) \sqcup (x+4) )\big)+ \iverson{c \neq 1}\cdot x \\
		\eeq &\iverson{c= 1}\cdot 0.5 \cdot (2\cdot x +4)+ \iverson{c \neq 1}\cdot x \eeq \inv \eegeq \inv~.
		\tag*{$\triangle$}
	\end{align*}
\end{example}

\noindent
By recursively applying \Cref{thm:aiming_low_nondet_loopfree} and determinizing inner loops, this generalizes \mbox{to nested loops:}
\begin{theorem}
	Let $\cc = \WHILEDOINV{\guard}{\ccbody}{\inv}$ be a (possibly nested) loop and let $f \in \expecs$. We have 
	\[
	\exptrans{\VCSYMBOLPARA{\dpastsubinvcond}{\AWPSYMBOL}}{\cc}{f}
	\qquad\text{implies}\qquad 
	\awp{\cc}{f} \eegeq \inv~.
	\]
\end{theorem}

\section{Case Studies}
\label{sec:case_studies}

In this section, we apply our framework to synthesize bound-guaranteeing determinizations of $\pgcl$ programs, all of which can be understood as countably infinite-state MDPs.

\subsection{Game of Nim}
\label{sec:nim-case-study}
We consider a variant of the game Nim (e.g.~\cite{wiki:Nim}), a 2-player zero-sum game which goes as follows:
$N$ tokens are placed on a table.
The players take turns; in each turn, the player has to remove 1, 2, or 3 tokens from the table.
The first player to remove the last token looses the game.

Suppose we are interested in finding a strategy performing reasonably well against an opponent that plays randomly.
We model this situation as the $\pgcl$ program $\CGame$ in \Cref{fig:game-program} on \cpageref{fig:game-program}.
Variable $x$ counts the number of tokens that have been removed so far.
Variable $\varTurn \in \{1,2\}$ indicates which player takes the next token(s).
The randomized opponent is player $1$.
If the program terminates in state where $\varTurn = i$, then player $i$ wins the game.

We claim that the controllable player $2$ can guarantee the following when starting the game with $x < N$ removed tokens:
\begin{itemize}
    \item If it's player $1$'s turn and $x+1 \congmod{4} N$ (i.e., $x+1 -N$ is a multiple of $4$), or if it's player $2$'s turn and $x+1 \ncongmod{4} N$, then player $2$ \emph{can win with probability one}.
    \item If it' player $1$'s turn but $x+1 \ncongmod{4} N$, then player $2$ can still \emph{win with probability $\geq \nicefrac 2 3$}.
\end{itemize}
Formally, the \emph{maximal} probability that player $2$ wins is given by $\awp{\CGame}{\iverson{\varTurn = 2}}$. We wish to find a determinization $\CGame''$ of $\CGame$ --- a strategy for player $2$ --- which guarantees the above lower bound on player $2$'s probability to win, i.e.,
\begin{align*}
     &\iverson{x<N} \cdot 
     \Big(
     \iverson{\POneTurn} \cdot
     \big(
     \iverson{x+1 \congmod{4} N} + \tfrac{2}{3} \cdot \iverson{x+1 \ncongmod{4} N}
     \big)
     ~+~
     \iverson{\PTwoTurn} \cdot \iverson{x+1 \ncongmod{4} N}
     \Big) \\
     \eeleq& \wp{\CGame''}{\iverson{\varTurn = 2}} ~.
     \tag{$\dagger$}\label{eqn:game_spec}
\end{align*}
For that, we annotate $\CGame$ with the subinvariant $\IGame$:
\begin{align*}
    \IGame \ddefeq & \iverson{x<N} \cdot 
        \Big(
            \iverson{\POneTurn} \cdot
                \big(
                    \iverson{x+1 \congmod{4} N} + \tfrac{2}{3} \cdot \iverson{x+1 \ncongmod{4} N}
                \big)
                ~+~
                \iverson{\PTwoTurn} \cdot \iverson{x+1 \ncongmod{4} N}
        \Big)\\
    & \qquad ~+~ \iverson{x \geq  N} \cdot \iverson{\PTwoTurn}
    \tag{see \iftoggle{arxiv}{\Cref{sec:nim-verification}}{\cite{arxiv}} for a proof of subinvariance}
\end{align*}
Applying $\atrans{}{\CGame}{\iverson{\varTurn = 2}}$, we obtain $\CGame' \refines \CGame$ shown in \Cref{fig:game-program-tamed} on \cpageref{fig:game-program-tamed}. Since $\CNim$ is clearly dAST (it iterates at most $N - x$ times), we may employ the $\VCSYMBOL$-provider from \Cref{thm:subinvcond}.\ref{thm:dastsubinvcond} to conclude that $\CGame'$ represents a permissive controller for player $2$ which guarantees the above bound, i.e., \emph{each} deterministic $\CGame'' \refines \CGame'$ satisfies the requirement (\ref{eqn:game_spec}).

\subsection{Optimal Gambling}

\subsubsection{Maximizing the Winning Probability.}
\label{sec:case_study_chain}

Consider the following gambling situation.
A gambler has to collect $N$ tokens to win a prize.
The game is played in rounds:
In each round, the gambler has to choose between flipping two coins with different biases.
Suppose the coins yield heads with probability $p$ and $q$, respectively.
If the result of the coin flip is tails, then the game is immediately lost.
Otherwise, the gambler wins one token in case of the bias-$p$ coin, and two tokens in case of the bias-$q$ coin.
The goal is to maximize the probability of winning given an initial budget $c$ of tokens, and the game parameters $p$, $q$, and $N$.
We briefly describe the optimal way to play this game (note that the probability to win two tokens in consecutive rounds with the bias-$p$ coin is $p^2$):
\begin{itemize}
	\item If $p^2 \geq q$, then playing with the bias-$p$ coin is optimal.
	\item Similarly, if $p \leq q$, then the gambler should always choose the bias-$q$ coin.
	\item Otherwise $p^2 < q < p$.
	In this case, the optimal choice depends on the current budget $c$:
	If only $N-c = 1$ token is needed to win the game, then it is better to choose the bias-$p$ coin.
	More generally, if $N-c$ is an odd number, then the best strategy is to play once with the bias-$p$ coin and up to $\nicefrac{(N-c-1)}{2}$ times with the bias-$q$ coin (the order is irrelevant).
\end{itemize}
\noindent
The gamble is readily modeled as the $\pgcl$ program $\CChain$ in \Cref{fig:chain-gamble} (we assume that $c, N \in \Nats$, $a \in \{0,1\}$, and $p,q \in [0,1]$).
Note that the game is lost as soon as $a = 1$.
Finding an \emph{optimal} gambling strategy that \emph{minimizes} the probability of losing thus amounts to finding 
\[
\text{a deterministic }\CChain^{\mathit{det}} \refines \CChain \qquad\text{such that}\qquad \wp{\CChain^{\mathit{det}}}{\iverson{a=1}} \eeq \dwp{\CChain}{\iverson{a=1}}
~.
\]
For that, we annotate $\CChain$ with the superinvariant $\IChain$:

\begin{align*}
	\IChain \ddefeq
	\iverson{a=0} \cdot \biggl( 1 ~-~ \Bigl( &\iverson{p \leq q}\cdot q^{\ceil*{\nicefrac{(N-c)}{2}}} 
	~+~ \iverson{q \leq p^2}\cdot p^{N-c} \\
	&\quad ~+~ \iverson{p^2<q<p}\cdot p^{\iverson{N-c>0}\iverson{N-c \textit{ is odd}}} \cdot q^{\floor*{\nicefrac{(N-c)}{2}}}
	\Bigr)\biggr)
	~+~ \iverson{a=1}
\end{align*}

\noindent
See \iftoggle{arxiv}{\Cref{sec:chain-verification}}{\cite{arxiv}} for a proof of superinvariance. Now, applying $\dtrans{}{\CChain}{\iverson{a=1}}$, we obtain $\CChain' \refines \CChain$ shown in \Cref{fig:chain-gamble-tamed}. We may thus apply the $\VCSYMBOL$-provider from \Cref{thm:superinvcond} to conclude that $\dwp{\CChain}{\iverson{a=1}} \eleq \dwp{\CChain'}{\iverson{a=1}}$. In fact, it can be shown that $\IChain$ is the \emph{exact} least fixpoint\footnote{We have $\dwp{\IChain}{\iverson{a=1}} \geq \IChain$ because $\IChain$ is a $\DWPSYMBOL$-subinvariant of $\CChain$, $\CChain$ is dAST (the loop is executed at most $N - c$ times, regardless of the gambler's choices), and the postexpectation $\iverson{a=1}$ is bounded~\cite{DBLP:phd/dnb/Kaminski19}.} of $\CChain$, so we obtain that 
\begin{quote}
	\emph{$\CChain'$ is a correct-by-construction \underline{o}p\underline{timal im}p\underline{lementation} of $\CChain$ w.r.t.\ minimizing the probability that $a=1$ holds when the program terminates.}
\end{quote}
Observe that program $\CChain'$ indeed represents the informal description of an optimal strategy given above.
$\CChain'$ is still nondeterministic, e.g.\ if $p^2 < q < p$ and $N-c = 3$, then both choices are enabled.
However, \Cref{thm:correct_transformations} ensures that \emph{every} determinization $\CChain^{\mathit{det}} \refines \CChain'$ satisfies $\wp{\CChain^{\mathit{det}}}{\iverson{a=1}} = \dwp{\CChain}{\iverson{a=1}}$ as desired.

\begin{figure}[tb]
	\small
	\begin{minipage}[t]{0.4\textwidth}
		\begin{align*}
			&\WHILEDOINVOPEN{c<N \wedge a=0}\\
			&\qquad \IFSYMBOL~\true \phantom{p^2}\\
			&\qquad\qquad \GCARROW \PCHOICE{\ASSIGN{c}{c+1}}{p}{\ASSIGN{a}{1}} \\
			&\qquad \IFSYMBOL~\true \phantom{p^2} \\
			&\qquad\qquad \GCARROW \PCHOICE{\ASSIGN{c}{c+2}}{q}{\ASSIGN{a}{1}} \\
			&\WHILEDOINVCLOSE{\IChain}
		\end{align*}
		\caption{
			Program $\CChain$ modeling a gamble.
		}
		\label{fig:chain-gamble}
	\end{minipage}
	\hfill
	\begin{minipage}[t]{0.55\textwidth}
		\begin{align*}
			&\WHILEDOINVOPEN{c<N \wedge a=0}\\
			&\qquad \IFSYMBOL~\alertannocolor{(q\leq p^2) \llor (p^2<q<p \wedge N-c \textit{ is odd})} \\
			&\qquad\qquad \GCARROW \PCHOICE{\ASSIGN{c}{c+1}}{p}{\ASSIGN{a}{1}} \\
			&\qquad \IFSYMBOL~\alertannocolor{(p\leq q) \llor (q=p^2) \llor ( p^2<q<p \land N-c \geq 2)} \\
			&\qquad\qquad \GCARROW \PCHOICE{\ASSIGN{c}{c+2}}{q}{\ASSIGN{a}{1}} \\
			&\WHILEDOINVCLOSE{\IChain}
		\end{align*}
		\caption{
			Program $\CChain'$ encoding all optimal strategies.
		}
		\label{fig:chain-gamble-tamed}
	\end{minipage}
\end{figure}

\subsubsection{Maximizing the Expected Payoff.}
\label{sec:case_study_geo}
%
Consider a variant of the above game where the gambler receives their current number of tokens $c$ as a prize once the variable $a$ becomes $0$ (modeled by program $\CGeo$ in \Cref{fig:geo-program}). We synthesize a strategy to win at least $\tfrac{p}{1-p} \emax \tfrac{2q}{1-q}$ tokens on expectation (assuming $p,q < 1$). For that, we annotatate the loop with subinvariant
%
\begin{align*}
	\IGeo
	\ddefeq
	\iverson{a=1}\cdot c + \iverson{a = 0} \cdot \left( c + \big( 
	\tfrac{p}{1-p} ~\emax~ \tfrac{2q}{1-q} \big) \right)
\end{align*}
and show in \iftoggle{arxiv}{\Cref{sec:geo-verification}}{\cite{arxiv}} that we may employ the $\VCSYMBOL$-provider from \Cref{thm:subinvcond}.\ref{thm:dpastsubinvcond}. Hence, \Cref{thm:correct_transformations} yields that \emph{any} determinization $\CGeo''$ of $\CGeo' = \atrans{}{\CGeo}{c}$ shown in \Cref{fig:geo-program-tamed} satisfies $\wp{\CGeo''}{c} \egeq \tfrac{p}{1-p} \emax \tfrac{2q}{1-q}$, i.e., $\CGeo''$ indeed realizes the lower-bound on the expected payoff of the game we aimed for.

\begin{figure}[tb]
	\small
	\begin{minipage}[t]{0.4\textwidth}
		\begin{align*}
			&\COMPOSE{\ASSIGN{c}{0}}{\ASSIGN{a}{0}}\fatsemi \\
			&\WHILEDOINVOPEN{a = 0} \\
			&\qquad\GCFST{\true}{\PCHOICE{\ASSIGN{c}{c+1}}{p}{\ASSIGN{a}{1}}}\\
			&\qquad\GCSEC{\true}{\PCHOICE{\ASSIGN{c}{c+2}}{q}{\ASSIGN{a}{1}}}\\
			&\WHILEDOINVCLOSE{\IGeo}
		\end{align*}
		\caption{$\CGeo$}
		\label{fig:geo-program}
	\end{minipage}
	\hfill
	\begin{minipage}[t]{0.55\textwidth}
		\begin{align*}
			&\COMPOSE{\ASSIGN{c}{0}}{\ASSIGN{a}{0}}\fatsemi \\
			&\WHILEDOINVOPEN{a = 0}\\
			&\quad\GCFST{\alertannocolor{2q(1-p) \leq p(1-q)}}{\PCHOICE{\ASSIGN{c}{c+1}}{p}{\ASSIGN{a}{1}}}\\
			&\quad\GCSEC{\alertannocolor{2q(1-p) \geq p(1-q)}}{\PCHOICE{\ASSIGN{c}{c+2}}{q}{\ASSIGN{a}{1}}}\\
			&\WHILEDOINVCLOSE{\IGeo}
		\end{align*}
		\caption{$\CGeo$ after transformation, assuming $p,q < 1$.}
		\label{fig:geo-program-tamed}
	\end{minipage}
\end{figure}

\section{Related Work}
\label{sec:rel_work}

\paragraph{Strategy Synthesis in Markov Decision Processes}

MDPs have a rich mathematical theory~\cite{DBLP:books/wi/Puterman94} and widespread applications across different fields.
In machine learning, the well-known \emph{reinforcement learning} problem is typically phrased in terms of MDPs, see e.g.~\cite{DBLP:books/sp/12/OtterloW12}.
However, RL usually does not provide strict guarantees about the resulting strategy.
\emph{Exact analysis} of MDPs ---i.e., finding \emph{provably} correct strategies as we do in this paper--- is one of the primary problems studied in the Probabilistic Model Checking (PMC) community, see~\cite{DBLP:conf/lics/Katoen16} and references therein for an overview.
PMC tools such as \textsf{PRISM}~\cite{DBLP:conf/cpe/KwiatkowskaNP02} or \textsf{Storm}~\cite{DBLP:journals/sttt/HenselJKQV22} support strategy synthesis for MDPs given in $\pgcl$-like modeling languages.
Compact symbolic representations of such strategies have been studied as well~\cite{DBLP:conf/hybrid/AshokJJKWZ20}.
The main difference to our approach is that these tools, and in fact most of the PMC literature, support only \emph{finite} MDPs and work by exploring the full state space.
Several subclasses of infinite MDPs have also been studied, including solvency games~\cite{DBLP:conf/fsttcs/BergerKSV08}, 1-counter and recursive MDPs \cite{DBLP:conf/soda/BrazdilBEKW10,DBLP:journals/jacm/EtessamiY15}, or parametric MDPs~\cite{DBLP:journals/jcss/JungesK0W21}.
Some of these works yield efficient algorithms, e.g.\ deciding if a target state can be reached with probability one in an MDP with one unbounded counter is decidable in PTIME~\cite{DBLP:conf/soda/BrazdilBEKW10}.
Our $\pgcl$ programs subsume these models, but their high expressivity comes at the cost of general undecidability.

Beyond PMC, researchers in AI have studied \emph{symbolic dynamic programming}~\cite{DBLP:conf/ijcai/BoutilierRP01,DBLP:conf/uai/SannerDB11}, a class of \emph{logic}-based representations and solution methods for MDPs.
These methods can be seen as a symbolic variant of value iteration.
In contrast to that, our approach is based on programs and uses invariants rather than explicit iteration.

\paragraph{Program Refinement}
Our approach is related to program refinement, where, originally, the goal is to refine an abstract model or specification to an executable (non-probabilistic) program~\cite{DBLP:books/daglib/0096285,DBLP:books/daglib/0024570,tlaplus}.  Later, the concept of program refinement has been generalized to the probabilistic setting \cite{McIverM05,DBLP:journals/sosym/AouadhiDL19,DBLP:conf/zum/HoangJRMM05}. Our problem statement can be understood as a program refinement problem: Given a nondeterministic $\pgcl$ program $\cc$ (the abstract model) satisfying some specification (i.e., $\dwp{\cc}{f} \eleq g$ or $\awp{\cc}{f} \egeq g$), refine $\cc$ to (a determinization) $\cc'$ which preserves this specification. The aforementioned approaches typically allow for more general specifications, which, however, comes with the loss of mechanizability. Given loop invariants satisfying their respective verification conditions, our approach is highly constructive as we obtain refinements in a syntactic manner. Rather than being a formal system for deriving provably correct algorithms, our approach is thus more tailored to a planning/AI setting, where the nondeterminism models different ways for an agent/adversary to behave, and where one wants to find suitable strategies.

Finally, \cite{DBLP:journals/corr/Mamouras16} uses a variant of Hoare logic to find strategies for (non-stochastic) games. \citet{DBLP:journals/corr/Mamouras16}, however, works in an uninterpreted and purely qualtitative setting. It is unclear how to adapt this work to the quantitative probabilistic setting.

\section{Conclusion \& Future Work}
\label{sec:conclusion}
We have presented a framework for obtaining strategies for nondeterministic probabilistic programs by means of deductive verification techniques. Several instances of this framework alongside with case studies demonstrate the applicability of our approach. Future work includes mechanizing our techniques in proof assistants~\cite{DBLP:series/natosec/0001SS17} and deductive verification infrastructures~\cite{heyvl}. Moreover, we plan to extend existing invariant synthesis techniques for probabilistic programs~\cite{DBLP:conf/tacas/BatzCJKKM23,DBLP:conf/cav/BaoTPHR22,DBLP:conf/sas/KatoenMMM10,DBLP:conf/atva/FengZJZX17}  to support nondeterminism. This will yield our approach to produce program-level strategies in a \emph{fully automated manner}. In addition to that, our framework could also be instantiated with more advanced proof rules for loops (see, e.g.,~\cite{kinduction}), which often yield simpler invariants.
Finally, we intend to transfer our technique to non-functional requirements such as expected runtimes \cite{ert,aert} or more general \emph{weighted programs}~\cite{DBLP:journals/pacmpl/BatzGKKW22}.

\iftoggle{arxiv}{}{
\section*{Data Availability Statement}
A full version of this article is available online~\cite{arxiv}.
}

\begin{acks}                            
The authors thank Annabelle McIver and Carroll Morgan for their pointers to existing literature on probabilistic program refinement.
\tw{What about grantsponsor macros (see latex comments)?}
\end{acks}

\bibliographystyle{ACM-Reference-Format}
\bibliography{bibfile}


\begin{thebibliography}{53}


\ifx \showCODEN    \undefined \def \showCODEN     #1{\unskip}     \fi
\ifx \showDOI      \undefined \def \showDOI       #1{#1}\fi
\ifx \showISBNx    \undefined \def \showISBNx     #1{\unskip}     \fi
\ifx \showISBNxiii \undefined \def \showISBNxiii  #1{\unskip}     \fi
\ifx \showISSN     \undefined \def \showISSN      #1{\unskip}     \fi
\ifx \showLCCN     \undefined \def \showLCCN      #1{\unskip}     \fi
\ifx \shownote     \undefined \def \shownote      #1{#1}          \fi
\ifx \showarticletitle \undefined \def \showarticletitle #1{#1}   \fi
\ifx \showURL      \undefined \def \showURL       {\relax}        \fi
\providecommand\bibfield[2]{#2}
\providecommand\bibinfo[2]{#2}
\providecommand\natexlab[1]{#1}
\providecommand\showeprint[2][]{arXiv:#2}

\bibitem[Abrial(2010)]%
        {DBLP:books/daglib/0024570}
\bibfield{author}{\bibinfo{person}{Jean{-}Raymond Abrial}.}
  \bibinfo{year}{2010}\natexlab{}.
\newblock \bibinfo{booktitle}{\emph{Modeling in Event-B - System and Software
  Engineering}}.
\newblock \bibinfo{publisher}{Cambridge University Press}.
\newblock


\bibitem[Aouadhi et~al\mbox{.}(2019)]%
        {DBLP:journals/sosym/AouadhiDL19}
\bibfield{author}{\bibinfo{person}{Mohamed~Amine Aouadhi},
  \bibinfo{person}{Beno{\^{\i}}t Delahaye}, {and} \bibinfo{person}{Arnaud
  Lanoix}.} \bibinfo{year}{2019}\natexlab{}.
\newblock \showarticletitle{Introducing probabilistic reasoning within
  Event-B}.
\newblock \bibinfo{journal}{\emph{Softw. Syst. Model.}} \bibinfo{volume}{18},
  \bibinfo{number}{3} (\bibinfo{year}{2019}), \bibinfo{pages}{1953--1984}.
\newblock


\bibitem[Ashok et~al\mbox{.}(2020)]%
        {DBLP:conf/hybrid/AshokJJKWZ20}
\bibfield{author}{\bibinfo{person}{Pranav Ashok}, \bibinfo{person}{Mathias
  Jackermeier}, \bibinfo{person}{Pushpak Jagtap}, \bibinfo{person}{Jan
  Kret{\'{\i}}nsk{\'{y}}}, \bibinfo{person}{Maximilian Weininger}, {and}
  \bibinfo{person}{Majid Zamani}.} \bibinfo{year}{2020}\natexlab{}.
\newblock \showarticletitle{dtControl: decision tree learning algorithms for
  controller representation}. In \bibinfo{booktitle}{\emph{{HSCC}}}.
  \bibinfo{publisher}{{ACM}}, \bibinfo{pages}{17:1--17:7}.
\newblock


\bibitem[Back and von Wright(1998)]%
        {DBLP:books/daglib/0096285}
\bibfield{author}{\bibinfo{person}{Ralph{-}Johan Back} {and}
  \bibinfo{person}{Joakim von Wright}.} \bibinfo{year}{1998}\natexlab{}.
\newblock \bibinfo{booktitle}{\emph{Refinement Calculus - {A} Systematic
  Introduction}}.
\newblock \bibinfo{publisher}{Springer}.
\newblock


\bibitem[Baier and Katoen(2008)]%
        {DBLP:books/daglib/0020348}
\bibfield{author}{\bibinfo{person}{Christel Baier} {and}
  \bibinfo{person}{Joost{-}Pieter Katoen}.} \bibinfo{year}{2008}\natexlab{}.
\newblock \bibinfo{booktitle}{\emph{Principles of model checking}}.
\newblock \bibinfo{publisher}{{MIT} Press}.
\newblock


\bibitem[Bao et~al\mbox{.}(2022)]%
        {DBLP:conf/cav/BaoTPHR22}
\bibfield{author}{\bibinfo{person}{Jialu Bao}, \bibinfo{person}{Nitesh
  Trivedi}, \bibinfo{person}{Drashti Pathak}, \bibinfo{person}{Justin Hsu},
  {and} \bibinfo{person}{Subhajit Roy}.} \bibinfo{year}{2022}\natexlab{}.
\newblock \showarticletitle{Data-Driven Invariant Learning for Probabilistic
  Programs}. In \bibinfo{booktitle}{\emph{{CAV} {(1)}}}
  \emph{(\bibinfo{series}{Lecture Notes in Computer Science},
  Vol.~\bibinfo{volume}{13371})}. \bibinfo{publisher}{Springer},
  \bibinfo{pages}{33--54}.
\newblock


\bibitem[Batz et~al\mbox{.}(2023a)]%
        {arxiv}
\bibfield{author}{\bibinfo{person}{Kevin Batz}, \bibinfo{person}{Tom~Jannik
  Biskup}, \bibinfo{person}{Joost-Pieter Katoen}, {and} \bibinfo{person}{Tobias
  Winkler}.} \bibinfo{year}{2023}\natexlab{a}.
\newblock \bibinfo{title}{Programmatic Strategy Synthesis: Resolving
  Nondeterminism in Probabilistic Programs}.
\newblock
\newblock
\showeprint[arxiv]{2311.06889}~[cs.LO]


\bibitem[Batz et~al\mbox{.}(2023b)]%
        {DBLP:conf/tacas/BatzCJKKM23}
\bibfield{author}{\bibinfo{person}{Kevin Batz}, \bibinfo{person}{Mingshuai
  Chen}, \bibinfo{person}{Sebastian Junges}, \bibinfo{person}{Benjamin~Lucien
  Kaminski}, \bibinfo{person}{Joost{-}Pieter Katoen}, {and}
  \bibinfo{person}{Christoph Matheja}.} \bibinfo{year}{2023}\natexlab{b}.
\newblock \showarticletitle{Probabilistic Program Verification via Inductive
  Synthesis of Inductive Invariants}. In \bibinfo{booktitle}{\emph{{TACAS}
  {(2)}}} \emph{(\bibinfo{series}{Lecture Notes in Computer Science},
  Vol.~\bibinfo{volume}{13994})}. \bibinfo{publisher}{Springer},
  \bibinfo{pages}{410--429}.
\newblock


\bibitem[Batz et~al\mbox{.}(2021a)]%
        {kinduction}
\bibfield{author}{\bibinfo{person}{Kevin Batz}, \bibinfo{person}{Mingshuai
  Chen}, \bibinfo{person}{Benjamin~Lucien Kaminski},
  \bibinfo{person}{Joost{-}Pieter Katoen}, \bibinfo{person}{Christoph Matheja},
  {and} \bibinfo{person}{Philipp Schr{\"{o}}er}.}
  \bibinfo{year}{2021}\natexlab{a}.
\newblock \showarticletitle{Latticed k-Induction with an Application to
  Probabilistic Programs}. In \bibinfo{booktitle}{\emph{{CAV} {(2)}}}
  \emph{(\bibinfo{series}{Lecture Notes in Computer Science},
  Vol.~\bibinfo{volume}{12760})}. \bibinfo{publisher}{Springer},
  \bibinfo{pages}{524--549}.
\newblock


\bibitem[Batz et~al\mbox{.}(2022)]%
        {DBLP:journals/pacmpl/BatzGKKW22}
\bibfield{author}{\bibinfo{person}{Kevin Batz}, \bibinfo{person}{Adrian
  Gallus}, \bibinfo{person}{Benjamin~Lucien Kaminski},
  \bibinfo{person}{Joost{-}Pieter Katoen}, {and} \bibinfo{person}{Tobias
  Winkler}.} \bibinfo{year}{2022}\natexlab{}.
\newblock \showarticletitle{Weighted programming: a programming paradigm for
  specifying mathematical models}.
\newblock \bibinfo{journal}{\emph{Proc. {ACM} Program. Lang.}}
  \bibinfo{volume}{6}, \bibinfo{number}{{OOPSLA1}} (\bibinfo{year}{2022}),
  \bibinfo{pages}{1--30}.
\newblock


\bibitem[Batz et~al\mbox{.}(2021b)]%
        {DBLP:journals/pacmpl/BatzKKM21}
\bibfield{author}{\bibinfo{person}{Kevin Batz},
  \bibinfo{person}{Benjamin~Lucien Kaminski}, \bibinfo{person}{Joost{-}Pieter
  Katoen}, {and} \bibinfo{person}{Christoph Matheja}.}
  \bibinfo{year}{2021}\natexlab{b}.
\newblock \showarticletitle{Relatively complete verification of probabilistic
  programs: an expressive language for expectation-based reasoning}.
\newblock \bibinfo{journal}{\emph{Proc. {ACM} Program. Lang.}}
  \bibinfo{volume}{5}, \bibinfo{number}{{POPL}} (\bibinfo{year}{2021}),
  \bibinfo{pages}{1--30}.
\newblock


\bibitem[Batz et~al\mbox{.}(2019)]%
        {qsl}
\bibfield{author}{\bibinfo{person}{Kevin Batz},
  \bibinfo{person}{Benjamin~Lucien Kaminski}, \bibinfo{person}{Joost{-}Pieter
  Katoen}, \bibinfo{person}{Christoph Matheja}, {and} \bibinfo{person}{Thomas
  Noll}.} \bibinfo{year}{2019}\natexlab{}.
\newblock \showarticletitle{Quantitative separation logic: a logic for
  reasoning about probabilistic pointer programs}.
\newblock \bibinfo{journal}{\emph{Proc. {ACM} Program. Lang.}}
  \bibinfo{volume}{3}, \bibinfo{number}{{POPL}} (\bibinfo{year}{2019}),
  \bibinfo{pages}{34:1--34:29}.
\newblock


\bibitem[Batz et~al\mbox{.}(2023c)]%
        {aert}
\bibfield{author}{\bibinfo{person}{Kevin Batz},
  \bibinfo{person}{Benjamin~Lucien Kaminski}, \bibinfo{person}{Joost{-}Pieter
  Katoen}, \bibinfo{person}{Christoph Matheja}, {and} \bibinfo{person}{Lena
  Verscht}.} \bibinfo{year}{2023}\natexlab{c}.
\newblock \showarticletitle{A Calculus for Amortized Expected Runtimes}.
\newblock \bibinfo{journal}{\emph{Proc. {ACM} Program. Lang.}}
  \bibinfo{volume}{7}, \bibinfo{number}{{POPL}} (\bibinfo{year}{2023}),
  \bibinfo{pages}{1957--1986}.
\newblock


\bibitem[Berger et~al\mbox{.}(2008)]%
        {DBLP:conf/fsttcs/BergerKSV08}
\bibfield{author}{\bibinfo{person}{Noam Berger}, \bibinfo{person}{Nevin Kapur},
  \bibinfo{person}{Leonard~J. Schulman}, {and} \bibinfo{person}{Vijay~V.
  Vazirani}.} \bibinfo{year}{2008}\natexlab{}.
\newblock \showarticletitle{Solvency Games}. In
  \bibinfo{booktitle}{\emph{{FSTTCS}}} \emph{(\bibinfo{series}{LIPIcs},
  Vol.~\bibinfo{volume}{2})}. \bibinfo{publisher}{Schloss Dagstuhl -
  Leibniz-Zentrum f{\"{u}}r Informatik}, \bibinfo{pages}{61--72}.
\newblock


\bibitem[Blackwell(1967)]%
        {blackwell1967positive}
\bibfield{author}{\bibinfo{person}{David Blackwell}.}
  \bibinfo{year}{1967}\natexlab{}.
\newblock \showarticletitle{Positive dynamic programming}. In
  \bibinfo{booktitle}{\emph{Proceedings of the 5th Berkeley symposium on
  Mathematical Statistics and Probability}}, Vol.~\bibinfo{volume}{1}.
  University of California Press Berkeley, \bibinfo{pages}{415--418}.
\newblock


\bibitem[Boutilier et~al\mbox{.}(2001)]%
        {DBLP:conf/ijcai/BoutilierRP01}
\bibfield{author}{\bibinfo{person}{Craig Boutilier}, \bibinfo{person}{Raymond
  Reiter}, {and} \bibinfo{person}{Bob Price}.} \bibinfo{year}{2001}\natexlab{}.
\newblock \showarticletitle{Symbolic Dynamic Programming for First-Order MDPs}.
  In \bibinfo{booktitle}{\emph{{IJCAI}}}. \bibinfo{publisher}{Morgan Kaufmann},
  \bibinfo{pages}{690--700}.
\newblock


\bibitem[Br{\'{a}}zdil et~al\mbox{.}(2010)]%
        {DBLP:conf/soda/BrazdilBEKW10}
\bibfield{author}{\bibinfo{person}{Tom{\'{a}}s Br{\'{a}}zdil},
  \bibinfo{person}{V{\'{a}}clav Brozek}, \bibinfo{person}{Kousha Etessami},
  \bibinfo{person}{Anton{\'{\i}}n Kucera}, {and} \bibinfo{person}{Dominik
  Wojtczak}.} \bibinfo{year}{2010}\natexlab{}.
\newblock \showarticletitle{One-Counter Markov Decision Processes}. In
  \bibinfo{booktitle}{\emph{{SODA}}}. \bibinfo{publisher}{{SIAM}},
  \bibinfo{pages}{863--874}.
\newblock


\bibitem[Dijkstra(1975)]%
        {DBLP:journals/cacm/Dijkstra75}
\bibfield{author}{\bibinfo{person}{Edsger~W. Dijkstra}.}
  \bibinfo{year}{1975}\natexlab{}.
\newblock \showarticletitle{Guarded Commands, Nondeterminacy and Formal
  Derivation of Programs}.
\newblock \bibinfo{journal}{\emph{Commun. {ACM}}} \bibinfo{volume}{18},
  \bibinfo{number}{8} (\bibinfo{year}{1975}), \bibinfo{pages}{453--457}.
\newblock


\bibitem[Dr{\"{a}}ger et~al\mbox{.}(2015)]%
        {DBLP:journals/corr/DragerFK0U15}
\bibfield{author}{\bibinfo{person}{Klaus Dr{\"{a}}ger},
  \bibinfo{person}{Vojtech Forejt}, \bibinfo{person}{Marta~Z. Kwiatkowska},
  \bibinfo{person}{David Parker}, {and} \bibinfo{person}{Mateusz Ujma}.}
  \bibinfo{year}{2015}\natexlab{}.
\newblock \showarticletitle{Permissive Controller Synthesis for Probabilistic
  Systems}.
\newblock \bibinfo{journal}{\emph{Log. Methods Comput. Sci.}}
  \bibinfo{volume}{11}, \bibinfo{number}{2} (\bibinfo{year}{2015}).
\newblock


\bibitem[Etessami and Yannakakis(2015)]%
        {DBLP:journals/jacm/EtessamiY15}
\bibfield{author}{\bibinfo{person}{Kousha Etessami} {and}
  \bibinfo{person}{Mihalis Yannakakis}.} \bibinfo{year}{2015}\natexlab{}.
\newblock \showarticletitle{Recursive Markov Decision Processes and Recursive
  Stochastic Games}.
\newblock \bibinfo{journal}{\emph{J. {ACM}}} \bibinfo{volume}{62},
  \bibinfo{number}{2} (\bibinfo{year}{2015}), \bibinfo{pages}{11:1--11:69}.
\newblock


\bibitem[Feng et~al\mbox{.}(2015)]%
        {DBLP:conf/iccps/FengWHT15}
\bibfield{author}{\bibinfo{person}{Lu Feng}, \bibinfo{person}{Clemens
  Wiltsche}, \bibinfo{person}{Laura~R. Humphrey}, {and} \bibinfo{person}{Ufuk
  Topcu}.} \bibinfo{year}{2015}\natexlab{}.
\newblock \showarticletitle{Controller synthesis for autonomous systems
  interacting with human operators}. In \bibinfo{booktitle}{\emph{{ICCPS}}}.
  \bibinfo{publisher}{{ACM}}, \bibinfo{pages}{70--79}.
\newblock


\bibitem[Feng et~al\mbox{.}(2017)]%
        {DBLP:conf/atva/FengZJZX17}
\bibfield{author}{\bibinfo{person}{Yijun Feng}, \bibinfo{person}{Lijun Zhang},
  \bibinfo{person}{David~N. Jansen}, \bibinfo{person}{Naijun Zhan}, {and}
  \bibinfo{person}{Bican Xia}.} \bibinfo{year}{2017}\natexlab{}.
\newblock \showarticletitle{Finding Polynomial Loop Invariants for
  Probabilistic Programs}. In \bibinfo{booktitle}{\emph{{ATVA}}}
  \emph{(\bibinfo{series}{Lecture Notes in Computer Science},
  Vol.~\bibinfo{volume}{10482})}. \bibinfo{publisher}{Springer},
  \bibinfo{pages}{400--416}.
\newblock


\bibitem[Gretz et~al\mbox{.}(2012)]%
        {operational_vs_weakest}
\bibfield{author}{\bibinfo{person}{Friedrich Gretz},
  \bibinfo{person}{Joost{-}Pieter Katoen}, {and} \bibinfo{person}{Annabelle
  McIver}.} \bibinfo{year}{2012}\natexlab{}.
\newblock \showarticletitle{Operational Versus Weakest Precondition Semantics
  for the Probabilistic Guarded Command Language}. In
  \bibinfo{booktitle}{\emph{{QEST}}}. \bibinfo{publisher}{{IEEE} Computer
  Society}, \bibinfo{pages}{168--177}.
\newblock


\bibitem[Haesaert et~al\mbox{.}(2017)]%
        {DBLP:journals/pe/HaesaertCA17}
\bibfield{author}{\bibinfo{person}{Sofie Haesaert}, \bibinfo{person}{Nathalie
  Cauchi}, {and} \bibinfo{person}{Alessandro Abate}.}
  \bibinfo{year}{2017}\natexlab{}.
\newblock \showarticletitle{Certified policy synthesis for general Markov
  decision processes: An application in building automation systems}.
\newblock \bibinfo{journal}{\emph{Perform. Evaluation}}  \bibinfo{volume}{117}
  (\bibinfo{year}{2017}), \bibinfo{pages}{75--103}.
\newblock


\bibitem[Hark et~al\mbox{.}(2020)]%
        {DBLP:journals/pacmpl/HarkKGK20}
\bibfield{author}{\bibinfo{person}{Marcel Hark},
  \bibinfo{person}{Benjamin~Lucien Kaminski}, \bibinfo{person}{J{\"{u}}rgen
  Giesl}, {and} \bibinfo{person}{Joost{-}Pieter Katoen}.}
  \bibinfo{year}{2020}\natexlab{}.
\newblock \showarticletitle{Aiming low is harder: induction for lower bounds in
  probabilistic program verification}.
\newblock \bibinfo{journal}{\emph{Proc. {ACM} Program. Lang.}}
  \bibinfo{volume}{4}, \bibinfo{number}{{POPL}} (\bibinfo{year}{2020}),
  \bibinfo{pages}{37:1--37:28}.
\newblock


\bibitem[Hensel et~al\mbox{.}(2022)]%
        {DBLP:journals/sttt/HenselJKQV22}
\bibfield{author}{\bibinfo{person}{Christian Hensel},
  \bibinfo{person}{Sebastian Junges}, \bibinfo{person}{Joost{-}Pieter Katoen},
  \bibinfo{person}{Tim Quatmann}, {and} \bibinfo{person}{Matthias Volk}.}
  \bibinfo{year}{2022}\natexlab{}.
\newblock \showarticletitle{The probabilistic model checker Storm}.
\newblock \bibinfo{journal}{\emph{Int. J. Softw. Tools Technol. Transf.}}
  \bibinfo{volume}{24}, \bibinfo{number}{4} (\bibinfo{year}{2022}),
  \bibinfo{pages}{589--610}.
\newblock


\bibitem[Hoang et~al\mbox{.}(2005)]%
        {DBLP:conf/zum/HoangJRMM05}
\bibfield{author}{\bibinfo{person}{Thai~Son Hoang}, \bibinfo{person}{Zhendong
  Jin}, \bibinfo{person}{Ken Robinson}, \bibinfo{person}{Annabelle McIver},
  {and} \bibinfo{person}{Carroll Morgan}.} \bibinfo{year}{2005}\natexlab{}.
\newblock \showarticletitle{Development via Refinement in Probabilistic {B} -
  Foundation and Case Study}. In \bibinfo{booktitle}{\emph{{ZB}}}
  \emph{(\bibinfo{series}{Lecture Notes in Computer Science},
  Vol.~\bibinfo{volume}{3455})}. \bibinfo{publisher}{Springer},
  \bibinfo{pages}{355--373}.
\newblock


\bibitem[Iverson(1962)]%
        {Iverson1962}
\bibfield{author}{\bibinfo{person}{Kenneth~E. Iverson}.}
  \bibinfo{year}{1962}\natexlab{}.
\newblock \bibinfo{booktitle}{\emph{A Programming Language}}.
\newblock \bibinfo{publisher}{John Wiley \& Sons, Inc.},
  \bibinfo{address}{USA}.
\newblock
\showISBNx{0471430145}


\bibitem[Junges et~al\mbox{.}(2021)]%
        {DBLP:journals/jcss/JungesK0W21}
\bibfield{author}{\bibinfo{person}{Sebastian Junges},
  \bibinfo{person}{Joost{-}Pieter Katoen}, \bibinfo{person}{Guillermo~A.
  P{\'{e}}rez}, {and} \bibinfo{person}{Tobias Winkler}.}
  \bibinfo{year}{2021}\natexlab{}.
\newblock \showarticletitle{The complexity of reachability in parametric Markov
  decision processes}.
\newblock \bibinfo{journal}{\emph{J. Comput. Syst. Sci.}}
  \bibinfo{volume}{119} (\bibinfo{year}{2021}), \bibinfo{pages}{183--210}.
\newblock


\bibitem[Kaminski(2019)]%
        {DBLP:phd/dnb/Kaminski19}
\bibfield{author}{\bibinfo{person}{Benjamin~Lucien Kaminski}.}
  \bibinfo{year}{2019}\natexlab{}.
\newblock \emph{\bibinfo{title}{Advanced weakest precondition calculi for
  probabilistic programs}}.
\newblock \bibinfo{thesistype}{Ph.\,D. Dissertation}. \bibinfo{school}{{RWTH}
  Aachen University, Germany}.
\newblock


\bibitem[Kaminski et~al\mbox{.}(2019)]%
        {DBLP:journals/acta/KaminskiKM19}
\bibfield{author}{\bibinfo{person}{Benjamin~Lucien Kaminski},
  \bibinfo{person}{Joost{-}Pieter Katoen}, {and} \bibinfo{person}{Christoph
  Matheja}.} \bibinfo{year}{2019}\natexlab{}.
\newblock \showarticletitle{On the hardness of analyzing probabilistic
  programs}.
\newblock \bibinfo{journal}{\emph{Acta Informatica}} \bibinfo{volume}{56},
  \bibinfo{number}{3} (\bibinfo{year}{2019}), \bibinfo{pages}{255--285}.
\newblock


\bibitem[Kaminski et~al\mbox{.}(2016)]%
        {ert}
\bibfield{author}{\bibinfo{person}{Benjamin~Lucien Kaminski},
  \bibinfo{person}{Joost{-}Pieter Katoen}, \bibinfo{person}{Christoph Matheja},
  {and} \bibinfo{person}{Federico Olmedo}.} \bibinfo{year}{2016}\natexlab{}.
\newblock \showarticletitle{Weakest Precondition Reasoning for Expected
  Run-Times of Probabilistic Programs}. In \bibinfo{booktitle}{\emph{{ESOP}}}
  \emph{(\bibinfo{series}{Lecture Notes in Computer Science},
  Vol.~\bibinfo{volume}{9632})}. \bibinfo{publisher}{Springer},
  \bibinfo{pages}{364--389}.
\newblock


\bibitem[Katoen(2016)]%
        {DBLP:conf/lics/Katoen16}
\bibfield{author}{\bibinfo{person}{Joost{-}Pieter Katoen}.}
  \bibinfo{year}{2016}\natexlab{}.
\newblock \showarticletitle{The Probabilistic Model Checking Landscape}. In
  \bibinfo{booktitle}{\emph{{LICS}}}. \bibinfo{publisher}{{ACM}},
  \bibinfo{pages}{31--45}.
\newblock


\bibitem[Katoen et~al\mbox{.}(2010)]%
        {DBLP:conf/sas/KatoenMMM10}
\bibfield{author}{\bibinfo{person}{Joost{-}Pieter Katoen},
  \bibinfo{person}{Annabelle McIver}, \bibinfo{person}{Larissa Meinicke}, {and}
  \bibinfo{person}{Carroll~C. Morgan}.} \bibinfo{year}{2010}\natexlab{}.
\newblock \showarticletitle{Linear-Invariant Generation for Probabilistic
  Programs: - Automated Support for Proof-Based Methods}. In
  \bibinfo{booktitle}{\emph{{SAS}}} \emph{(\bibinfo{series}{Lecture Notes in
  Computer Science}, Vol.~\bibinfo{volume}{6337})}.
  \bibinfo{publisher}{Springer}, \bibinfo{pages}{390--406}.
\newblock


\bibitem[Kattenbelt et~al\mbox{.}(2009)]%
        {DBLP:conf/vmcai/KattenbeltKNP09}
\bibfield{author}{\bibinfo{person}{Mark Kattenbelt}, \bibinfo{person}{Marta~Z.
  Kwiatkowska}, \bibinfo{person}{Gethin Norman}, {and} \bibinfo{person}{David
  Parker}.} \bibinfo{year}{2009}\natexlab{}.
\newblock \showarticletitle{Abstraction Refinement for Probabilistic Software}.
  In \bibinfo{booktitle}{\emph{{VMCAI}}} \emph{(\bibinfo{series}{Lecture Notes
  in Computer Science}, Vol.~\bibinfo{volume}{5403})}.
  \bibinfo{publisher}{Springer}, \bibinfo{pages}{182--197}.
\newblock


\bibitem[Kattenbelt et~al\mbox{.}(2010)]%
        {DBLP:journals/fmsd/KattenbeltKNP10}
\bibfield{author}{\bibinfo{person}{Mark Kattenbelt}, \bibinfo{person}{Marta~Z.
  Kwiatkowska}, \bibinfo{person}{Gethin Norman}, {and} \bibinfo{person}{David
  Parker}.} \bibinfo{year}{2010}\natexlab{}.
\newblock \showarticletitle{A game-based abstraction-refinement framework for
  Markov decision processes}.
\newblock \bibinfo{journal}{\emph{Formal Methods Syst. Des.}}
  \bibinfo{volume}{36}, \bibinfo{number}{3} (\bibinfo{year}{2010}),
  \bibinfo{pages}{246--280}.
\newblock


\bibitem[Kozen(1983)]%
        {Kozen1983}
\bibfield{author}{\bibinfo{person}{Dexter Kozen}.}
  \bibinfo{year}{1983}\natexlab{}.
\newblock \showarticletitle{A Probabilistic {PDL}}. In
  \bibinfo{booktitle}{\emph{Proceedings of the 15th Annual {ACM} Symposium on
  Theory of Computing, 25-27 April, 1983, Boston, Massachusetts, {USA}}}.
  \bibinfo{publisher}{{ACM}}, \bibinfo{pages}{291--297}.
\newblock
\urldef\tempurl%
\url{https://doi.org/10.1145/800061.808758}
\showDOI{\tempurl}


\bibitem[Kozen(1985)]%
        {Kozen1985}
\bibfield{author}{\bibinfo{person}{Dexter Kozen}.}
  \bibinfo{year}{1985}\natexlab{}.
\newblock \showarticletitle{A Probabilistic {PDL}}.
\newblock \bibinfo{journal}{\emph{J. Comput. Syst. Sci.}} \bibinfo{volume}{30},
  \bibinfo{number}{2} (\bibinfo{year}{1985}), \bibinfo{pages}{162--178}.
\newblock
\urldef\tempurl%
\url{https://doi.org/10.1016/0022-0000(85)90012-1}
\showDOI{\tempurl}


\bibitem[Kozine and Utkin(2002)]%
        {DBLP:journals/rc/KozineU02}
\bibfield{author}{\bibinfo{person}{Igor Kozine} {and} \bibinfo{person}{Lev~V.
  Utkin}.} \bibinfo{year}{2002}\natexlab{}.
\newblock \showarticletitle{Interval-Valued Finite Markov Chains}.
\newblock \bibinfo{journal}{\emph{Reliab. Comput.}} \bibinfo{volume}{8},
  \bibinfo{number}{2} (\bibinfo{year}{2002}), \bibinfo{pages}{97--113}.
\newblock


\bibitem[Kwiatkowska et~al\mbox{.}(2002)]%
        {DBLP:conf/cpe/KwiatkowskaNP02}
\bibfield{author}{\bibinfo{person}{Marta~Z. Kwiatkowska},
  \bibinfo{person}{Gethin Norman}, {and} \bibinfo{person}{David Parker}.}
  \bibinfo{year}{2002}\natexlab{}.
\newblock \showarticletitle{{PRISM:} Probabilistic Symbolic Model Checker}. In
  \bibinfo{booktitle}{\emph{Computer Performance Evaluation / {TOOLS}}}
  \emph{(\bibinfo{series}{Lecture Notes in Computer Science},
  Vol.~\bibinfo{volume}{2324})}. \bibinfo{publisher}{Springer},
  \bibinfo{pages}{200--204}.
\newblock


\bibitem[Lamport(2002)]%
        {tlaplus}
\bibfield{author}{\bibinfo{person}{Leslie Lamport}.}
  \bibinfo{year}{2002}\natexlab{}.
\newblock \bibinfo{booktitle}{\emph{Specifying Systems, The {TLA+} Language and
  Tools for Hardware and Software Engineers}}.
\newblock \bibinfo{publisher}{Addison-Wesley}.
\newblock


\bibitem[Leino(2010)]%
        {DBLP:conf/lpar/Leino10}
\bibfield{author}{\bibinfo{person}{K.~Rustan~M. Leino}.}
  \bibinfo{year}{2010}\natexlab{}.
\newblock \showarticletitle{Dafny: An Automatic Program Verifier for Functional
  Correctness}. In \bibinfo{booktitle}{\emph{{LPAR} (Dakar)}}
  \emph{(\bibinfo{series}{Lecture Notes in Computer Science},
  Vol.~\bibinfo{volume}{6355})}. \bibinfo{publisher}{Springer},
  \bibinfo{pages}{348--370}.
\newblock


\bibitem[Mamouras(2016)]%
        {DBLP:journals/corr/Mamouras16}
\bibfield{author}{\bibinfo{person}{Konstantinos Mamouras}.}
  \bibinfo{year}{2016}\natexlab{}.
\newblock \showarticletitle{Synthesis of Strategies Using the Hoare Logic of
  Angelic and Demonic Nondeterminism}.
\newblock \bibinfo{journal}{\emph{Log. Methods Comput. Sci.}}
  \bibinfo{volume}{12}, \bibinfo{number}{3} (\bibinfo{year}{2016}).
\newblock


\bibitem[McIver and Morgan(2005)]%
        {McIverM05}
\bibfield{author}{\bibinfo{person}{Annabelle McIver} {and}
  \bibinfo{person}{Carroll Morgan}.} \bibinfo{year}{2005}\natexlab{}.
\newblock \bibinfo{booktitle}{\emph{Abstraction, Refinement and Proof for
  Probabilistic Systems}}.
\newblock \bibinfo{publisher}{Springer}.
\newblock
\showISBNx{978-0-387-40115-7}
\urldef\tempurl%
\url{https://doi.org/10.1007/b138392}
\showDOI{\tempurl}


\bibitem[M{\"{u}}ller et~al\mbox{.}(2017)]%
        {DBLP:series/natosec/0001SS17}
\bibfield{author}{\bibinfo{person}{Peter M{\"{u}}ller}, \bibinfo{person}{Malte
  Schwerhoff}, {and} \bibinfo{person}{Alexander~J. Summers}.}
  \bibinfo{year}{2017}\natexlab{}.
\newblock \showarticletitle{Viper: {A} Verification Infrastructure for
  Permission-Based Reasoning}.
\newblock In \bibinfo{booktitle}{\emph{Dependable Software Systems
  Engineering}}. \bibinfo{series}{{NATO} Science for Peace and Security Series
  - {D:} Information and Communication Security}, Vol.~\bibinfo{volume}{50}.
  \bibinfo{publisher}{{IOS} Press}, \bibinfo{pages}{104--125}.
\newblock


\bibitem[Navarro and Olmedo(2022)]%
        {DBLP:journals/scp/NavarroO22}
\bibfield{author}{\bibinfo{person}{Marcelo Navarro} {and}
  \bibinfo{person}{Federico Olmedo}.} \bibinfo{year}{2022}\natexlab{}.
\newblock \showarticletitle{Slicing of probabilistic programs based on
  specifications}.
\newblock \bibinfo{journal}{\emph{Sci. Comput. Program.}}
  \bibinfo{volume}{220} (\bibinfo{year}{2022}), \bibinfo{pages}{102822}.
\newblock


\bibitem[Ornstein(1969)]%
        {ornstein1969existence}
\bibfield{author}{\bibinfo{person}{Donald Ornstein}.}
  \bibinfo{year}{1969}\natexlab{}.
\newblock \showarticletitle{On the existence of stationary optimal strategies}.
\newblock \bibinfo{journal}{\emph{Proc. Amer. Math. Soc.}}
  \bibinfo{volume}{20}, \bibinfo{number}{2} (\bibinfo{year}{1969}),
  \bibinfo{pages}{563--569}.
\newblock


\bibitem[Park(1969)]%
        {park1969fixpoint}
\bibfield{author}{\bibinfo{person}{David Park}.}
  \bibinfo{year}{1969}\natexlab{}.
\newblock \showarticletitle{Fixpoint Induction and Proofs of Program
  Properties}.
\newblock \bibinfo{journal}{\emph{Machine intelligence}}  \bibinfo{volume}{5}
  (\bibinfo{year}{1969}).
\newblock


\bibitem[Puterman(1994)]%
        {DBLP:books/wi/Puterman94}
\bibfield{author}{\bibinfo{person}{Martin~L. Puterman}.}
  \bibinfo{year}{1994}\natexlab{}.
\newblock \bibinfo{booktitle}{\emph{Markov Decision Processes: Discrete
  Stochastic Dynamic Programming}}.
\newblock \bibinfo{publisher}{Wiley}.
\newblock


\bibitem[Sanner et~al\mbox{.}(2011)]%
        {DBLP:conf/uai/SannerDB11}
\bibfield{author}{\bibinfo{person}{Scott Sanner},
  \bibinfo{person}{Karina~Valdivia Delgado}, {and}
  \bibinfo{person}{Leliane~Nunes de Barros}.} \bibinfo{year}{2011}\natexlab{}.
\newblock \showarticletitle{Symbolic Dynamic Programming for Discrete and
  Continuous State MDPs}. In \bibinfo{booktitle}{\emph{{UAI}}}.
  \bibinfo{publisher}{{AUAI} Press}, \bibinfo{pages}{643--652}.
\newblock


\bibitem[Schr\"{o}er et~al\mbox{.}(2023)]%
        {heyvl}
\bibfield{author}{\bibinfo{person}{Philipp Schr\"{o}er}, \bibinfo{person}{Kevin
  Batz}, \bibinfo{person}{Benjamin~Lucien Kaminski},
  \bibinfo{person}{Joost-Pieter Katoen}, {and} \bibinfo{person}{Christoph
  Matheja}.} \bibinfo{year}{2023}\natexlab{}.
\newblock \showarticletitle{A Deductive Verification Infrastructure for
  Probabilistic Programs}.
\newblock  \bibinfo{volume}{7}, \bibinfo{number}{OOPSLA2}, Article
  \bibinfo{articleno}{294} (\bibinfo{date}{oct} \bibinfo{year}{2023}),
  \bibinfo{numpages}{31}~pages.
\newblock
\urldef\tempurl%
\url{https://doi.org/10.1145/3622870}
\showDOI{\tempurl}


\bibitem[van Otterlo and Wiering(2012)]%
        {DBLP:books/sp/12/OtterloW12}
\bibfield{author}{\bibinfo{person}{Martijn van Otterlo} {and}
  \bibinfo{person}{Marco~A. Wiering}.} \bibinfo{year}{2012}\natexlab{}.
\newblock \showarticletitle{Reinforcement Learning and Markov Decision
  Processes}.
\newblock In \bibinfo{booktitle}{\emph{Reinforcement Learning}}.
  \bibinfo{series}{Adaptation, Learning, and Optimization},
  Vol.~\bibinfo{volume}{12}. \bibinfo{publisher}{Springer},
  \bibinfo{pages}{3--42}.
\newblock


\bibitem[Wikipedia(2023)]%
        {wiki:Nim}
\bibfield{author}{\bibinfo{person}{Wikipedia}.}
  \bibinfo{year}{2023}\natexlab{}.
\newblock \bibinfo{title}{{Nim} --- {W}ikipedia{,} The Free Encyclopedia}.
\newblock
  \bibinfo{howpublished}{\url{http://en.wikipedia.org/w/index.php?title=Nim&oldid=1163491825}}.
\newblock
\newblock
\shownote{[Online; accessed 10-July-2023]}.


\end{thebibliography}

\iftoggle{arxiv}{
    \allowdisplaybreaks
    \appendix
    \newpage
    %
    \section{Omitted Proofs}
\subsection{Proof of~\Cref{thm:refinementProp}}
\label{proof:refinementProp}
\refinementProp*
\begin{proof}
    By induction on the structure of $\cc$. For all cases except for guarded choice, the claim follows straightforwardly by the I.H.\ and monotonicity of the respective transformer (cf.\ \Cref{thm:healthiness}.\ref{thm:healthiness_monotonic}). For $\cc = \GC{\guard_1}{\cc_1}{\guard_2}{\cc_2}$, consider the following:
    By definition of $\refines$, $\cc'$ is  of the form $\cc' = \GC{\guard_1'}{\cc_1'}{\guard_2'}{\cc_2'}$ and with $\guard_1'\entails \guard_1$, $\guard_2'\entails\guard_2$, and $\entails (\guard_1'\vee\guard_2')$. Now fix some $\sigma \in \States$. We proceed by a case distinction:
    
    
    If $\sigma \models\guard_1$ and $\sigma\not\models\guard_2$, then also $\sigma\models\guard_1'$ and $\sigma\not\models\guard_2'$, which yields equality of the respective preexpectations.
    
    The case $\sigma\not\models\guard_1$ and $\sigma\models\guard_2$ is analogous to the above.
    
    Finally, assume $\sigma \models\guard_1$ and $\sigma\models\guard_2$ and consider the $\DWPSYMBOL$ transformer (the argument for the $\AWPSYMBOL$ is analogous).
    We have 
    \begin{align*}
    	& \dwp{\cc}{f}(\sigma) \\
    	\eeq & \min \{\dwp{\cc_1}{f}(\sigma), \dwp{\cc_2}{f}(\sigma)  \} \\
    	\lleq & \min \{\dwp{\cc_1'}{f}(\sigma), \dwp{\cc_2'}{f}(\sigma)  \} 
    	\tag{I.H.\ on $\cc_1$ and $\cc_2$} \\
    	\lleq & \dwp{\cc'}{f}(\sigma)~.
    	\tag{if $A \subseteq B$, then $\min B \leq \min A$}
    \end{align*}

\end{proof}

\subsection{Proof of \Cref{thm:trans_refines}}
\label{proof:trans_refines}
\transrefines*

\begin{proof}
		By induction on $\cc$. For the base cases, there is nothing to show. All composite cases except for the guarded choice follow immediately from the I.H.\ (notice that the claim holds for \emph{arbitrary} postexpectations). Now let $\cc = \GC{\guard_1}{\cc_1}{\guard_2}{\cc_2}$. First, observe that 
	\[ 
	\trans{\comprel}{\SOMEWPSYMBOL}{\cc_1}{f} \refines \cc_1  \qquad\text{and}\qquad \trans{\comprel}{\SOMEWPSYMBOL}{\cc_2}{f} \refines \cc_2
	\]
	hold by the I.H. Now, abbreviating the guards occurring in $\trans{\comprel}{\SOMEWPSYMBOL}{\cc}{f}$ by
	\[
	\guardb_1 = \guard_1 \land (\guard_2 \limpl \exptransT{\SOMEWPRESYMBOL}{\cc_1}(f) \comprel \exptransT{\SOMEWPRESYMBOL}{\cc_2}(f))
	\quad\text{and}\quad
	\guardb_2 = \guard_2 \land (\guard_1 \limpl \exptransT{\SOMEWPRESYMBOL}{\cc_2}(f) \comprel \exptransT{\SOMEWPRESYMBOL}{\cc_1}(f))~,
	\]
	respectively, the entailments $\guardb_1\entails \guard_1$ and $\guardb_2 \entails \guard_2$ hold trivially. For $\entails (\guardb_1 \vee \guardb_2)$, consider the following: for every state $\sigma$, we have
	\[
	\sigma ~{}\models{}~ \exptransT{\SOMEWPRESYMBOL}{\cc_1}(f) \comprel \exptransT{\SOMEWPRESYMBOL}{\cc_2}(f) ~{}\vee{}~ \exptransT{\SOMEWPRESYMBOL}{\cc_2}(f) \comprel \exptransT{\SOMEWPRESYMBOL}{\cc_1}(f)
	\] 
	Hence, if $\entails \guard_1 \vee \guard_2$, then also $\entails \guardb_1 \vee \guardb_2$. \\
\end{proof}

\subsection{Proof of~\Cref{thm:trans_preserves_wpre}}
\label{proof:trans_preserves_wpre}
\transpreserveswpre*
\begin{proof}
		We prove the claim for $\SOMEWPSYMBOL= \DWPSYMBOL$. The reasoning for $\AWPSYMBOL$ is dual. We proceed by induction on $\cc$. For the base cases, there is nothing to show. \\
	\noindent
	\emph{The case $\cc = \COMPOSE{\cc_1}{\cc_2}$.}
	First notice that $\cc' \refines \dtrans{\SOMEWPSYMBOL}{\cc}{f}$ implies  $\cc' \eeq \COMPOSE{\cc_1'}{\cc_2'}$ for some $\cc_1'\refines \dtrans{\SOMEWPSYMBOL}{\cc_1}{\exptransT{\SOMEWPRESYMBOL}{\cc_2}(f)}$ and $\cc_2' \refines \dtrans{\SOMEWPSYMBOL}{\cc_2}{f}$. Hence,
	\begin{align*}
		&\somewpre{\cc'}{f} \\
		\eeq& \somewpre{\COMPOSE{\cc_1'}{\cc_2'}}{f} \\
		\eeq& \somewpre{\cc_1'}{\somewpre{\cc_2'}{f}} \\
		\eeq&\somewpre{\cc_1'}{\somewpre{\cc_2}{f}}
		\tag{I.H.\ on $\cc_2$} \\
		\eeq&\somewpre{\cc_1}{\somewpre{\cc_2}{f}}
		\tag{I.H.\ on $\cc_1$} \\
		\eeq & \somewpre{\COMPOSE{\cc_1}{\cc_2}}{f}  \\
		\eeq & \somewpre{\cc}{f} ~.
	\end{align*}
	\emph{The case $\cc = \GC{\guard_1}{\cc_1}{\guard_2}{\cc_2}$.} Abbreviate the guards occurring in $\dtrans{\SOMEWPSYMBOL}{\cc}{f}$ by 
	\[
	\guardb_1 = \guard_1 \land (\guard_2 \limpl \exptransT{\SOMEWPREONESYMBOL}{\cc_1}(f) \predeleq \exptransT{\SOMEWPREONESYMBOL}{\cc_2}(f))
	\quad\text{and}\quad
	\guardb_2 = \guard_2 \land (\guard_1 \limpl \exptransT{\SOMEWPREONESYMBOL}{\cc_2}(f) \predeleq \exptransT{\SOMEWPREONESYMBOL}{\cc_1}(f))~,
	\]
	respectively. Since $\cc' \refines \dtrans{\SOMEWPSYMBOL}{\cc}{f}$, we have
	\[\cc' \eeq \GC{\guard_1'}{\cc_1'}{\guard_2'}{\cc_2'}\] 
	for some $\cc_1'\refines \dtrans{\SOMEWPSYMBOL}{\cc_1}{f}$ and $\cc_2' \refines \dtrans{\SOMEWPSYMBOL}{\cc_2}{f}$ 
	with $\guard_1' \entails \guardb_1$, $\guard_2'\entails \guardb_2$, and $\entails (\guard_1' \vee \guard_2')$. Moreover, since $\dtrans{\SOMEWPSYMBOL}{\cc}{f} \refines \cc$ by \Cref{thm:trans_refines}, the I.H.\ implies that  $\somewpre{\cc'}{f}(\sigma)$ and $\somewpre{\cc}{f}(\sigma)$ differ only if both $\guard_1$ and $\guard_2$ evaluate to $\true$ under $\sigma$, i.e.,
	\[ 
	\somewpre{\cc'}{f}(\sigma)~{} \neq{}~ \somewpre{\cc}{f}(\sigma)
	\qquad\text{implies}\qquad 
	\sigma \models \guard_1 ~\text{and}~ \sigma \models \guard_2 ~.
	\]
	Hence, it suffices to consider such $\sigma$. We now distinguish the following cases:
	
	If $\somewpre{\cc_1}{f}(\sigma) < \somewpre{\cc_2}{f}(\sigma)$, then $\sigma \models \guardb_1$ and $\sigma \not\models \guardb_2$ and thus $\sigma \models \guard_1'$ and $\sigma \not\models \guard_2'$. Hence,
	\begin{align*}
		& \somewpre{\cc'}{f}(\sigma) \\
		\eeq & \somewpre{\cc_1'}{f}(\sigma)
		\tag{$\sigma \models\guard_1'$ and $\sigma \not\models\guard_2'$} \\
		\eeq & \somewpre{\cc}{f}(\sigma)~.
		\tag{$\somewpre{\cc_1}{f}(\sigma) < \somewpre{\cc_2}{f}(\sigma)$}
	\end{align*}
	The reasoning for the case $\somewpre{\cc_1}{f}(\sigma) > \somewpre{\cc_2}{f}(\sigma)$ is analogous.
	
	Finally, if $\somewpre{\cc_1}{f}(\sigma) = \somewpre{\cc_2}{f}(\sigma)$, then $\sigma \models\guardb_1$ and $\sigma \models\guardb_2$. Moreover, we have $\somewpre{\cc_1'}{f}(\sigma) = \somewpre{\cc_2'}{f}(\sigma)$ by the I.H.\  Since at least one of $\guard_1'$ and $\guard_2'$ is $\true$ under $\sigma$, we get
	\begin{align*}
		&\somewpre{\cc'}{f}(\sigma) \\
		\eeq & \min\{ \somewpre{\cc_1'}{f}(\sigma), \somewpre{\cc_2'}{f}(\sigma) \}
		\tag{both values coincide} \\
		\eeq& \min\{\somewpre{\cc_1}{f}(\sigma), \somewpre{\cc_2}{f}(\sigma)\} 
		\tag{by I.H.}\\
		\eeq & \somewpre{\cc}{f}(\sigma)~.
	\end{align*}
	\noindent
	\emph{The case $\cc = \WHILEDOINV{\guard}{\cc}{\inv}$} is trivial.
\end{proof}
    \section{Omitted Definitions}

\subsection{Formal Construction of the MDP Semantics}
\label{app:mdp_rules}
\begin{figure}[h]
\begin{align*}
  &
  \infer{
    \Exec{\SKIP}{\pstate}{\labtau}{1}{\Term}{\pstate}
  }{
  }
  \quad
  \infer{
    \Exec{\ASSIGN{x}{\aexpr}}{\pstate}{\labtau}{1}{\Term}{\pstate\statesubst{x}{\pstate(\aexpr)}}
  }{
  }
\quad
\infer{
	\Exec{\Term}{\pstate}{\labtau}{1}{\Term}{\pstate}
}{
}
  \\[1ex]
  &
  \infer{
    \Exec{\COMPOSE{\cc_1}{\cc_2}}{\pstate}{\actlab}{\pnum}{\cc_2}{\pstate'}
  }{
    \Exec{\cc_1}{\pstate}{\actlab}{\pnum}{\Term}{\pstate'}
  }
  \quad
  \infer{
    \Exec{\COMPOSE{\cc_1}{\cc_2}}{\pstate}{\actlab}{\pnum}{\COMPOSE{\cc_1'}{\cc_2}}{\pstate'}
  }{
    \Exec{\cc_1}{\pstate}{\actlab}{\pnum}{\cc_1'}{\pstate'}
  }
  \\[1ex]
  &
  \infer{
    \Exec{\GC{\guard_1}{\cc_1}{\guard_2}{\cc_2}}{\pstate}{\labalpha}{1}{\cc_1}{\pstate}
  }{
    \pstate \models \guard_1
  }
  \quad
  \infer{
	\Exec{\GC{\guard_1}{\cc_1}{\guard_2}{\cc_2}}{\pstate}{\labbeta}{1}{\cc_2}{\pstate}
}{
	\pstate \models \guard_2
}
  \\[1ex]
  &
    \infer{
  	    \Exec{\PCHOICE{\cc_1}{\pexpr}{\cc_2}}{\pstate}{\labtau}{\pexpr(\pstate)}{\cc_1}{\pstate}
  	    }{
          \cc_1 \neq \cc_2
  	    }
    \quad
    \infer{
  	    \Exec{\PCHOICE{\cc_1}{\pexpr}{\cc_2}}{\pstate}{\labtau}{1-\pexpr(\pstate)}{\cc_2}{\pstate}
  	    }{
          \cc_1 \neq \cc_2
  	    }
    \quad
    \infer{
        \Exec{\PCHOICE{\cc}{\pexpr}{\cc}}{\pstate}{\labtau}{1}{\cc}{\pstate}
        }{
        }
    \\[1ex]
 & \infer{
    \Exec{\WHILEDO{\guard}{\cc}}{\pstate}{\labtau}{1}{\Term}{\pstate}
  }{
    \pstate \models \neg\guard
  }
  \quad
  \infer{
  \Exec{\WHILEDO{\guard}{\cc}}{\pstate}{\labtau}{1}{\COMPOSE{\cc}{\WHILEDO{\guard}{\cc}}}{\pstate}
  }{
    \pstate \models \guard
  }
\end{align*}
\caption{Inference rules determining the execution relation $\ExecSymbol$.}
\label{table:op}
\end{figure}

In this section we provide the details ommited in \Cref{sec:mdp_semantics}.

The inference rules for determining the execution relation $\ExecSymbol$ are given in \Cref{table:op}.
Let us briefly go over those rules.
The rules for $\SKIP$, assignments, conditionals, and loops are standard.
In each case, the execution proceeds deterministically, hence all transition probabilities are $1$.
For guarded choice statements $\GC{\guard_1}{\cc_1}{\guard_2}{\cc_2}$, the execution relation has one transition for each enabled branch in the current configuration.
Recall that exactly one or both branches are enabled.
The action labels $\labalpha$ and $\labbeta$ are used to mark which of the two branches was taken.
For a probabilistic choice $\PCHOICE{\cc_1}{\pexpr}{\cc_2}$ where $\cc_1 \neq \cc_2$ are different programs, there are two possible executions:
With probability $\pexpr$ we execute $\cc_1$ and with probability $1-\pexpr$, we execute $\cc_2$. 
If the programs $\cc_1$ and $\cc_2$ happen to be equal, i.e., $\cc_1 = \cc_2 = \cc$, then the execution proceeds deterministically with program $\cc$.
Note that action label $\labtau$ is used in all transitions except in those that correspond to executing a guarded choice statement.

\subsection{Suitable for Optional Stopping}
\label{app:optional_stopping}
\cite{DBLP:journals/pacmpl/HarkKGK20} employ the Optional Stopping Theorem to obtain a proof rule for lower bounds on expected outcomes of loops. This proof rule relies on the following side condition: A loop $\WHILEDOINV{\guard}{\ccbody}{\inv}$ is \emph{suitable for optional stopping w.r.t.\ $f$} if (i) $\inv = \iverson{\guard}\cdot \inv' + \iverson{\neg\guard}\cdot f$ for some $\inv' \in \expecs$, (ii) $f$,$\inv$, and $\iverson{\guard}\cdot \awp{\ccbody}{\inv} + \iverson{\neg\guard}\cdot f$ are all pointwise smaller than $\infty$, and (iii) $\inv$ is \emph{demonically conditionally difference bounded}, i.e., there is $b \in \PosReals$ such that for all states $\sigma$ we have $(\iverson{\guard} \cdot \awp{\ccbody}{|\inv - \inv(\sigma)|})(\sigma) \leq b$.\jp{explain why this is called demonically cdb but defined in terms of awp}

\subsection{Facts about $\DWPSYMBOL$, $\AWPSYMBOL$, and $\VCSYMBOL$}

\begin{theorem}[\cite{McIverM05}]
	\label{thm:healthiness}
	Let $\cc \in\pgcl$ and $\SOMEWPSYMBOL \in \{\DWPSYMBOL,\AWPSYMBOL\}$. We have:
	\begin{enumerate}
		\item\label{thm:healthiness_monotonic} 
		$\exptransT{\SOMEWPSYMBOL}{\cc}$ is \emph{monotonic}, i.e., 
		\[
		\text{for all $f,g \in\expecs$}\colon \qquad
		f \eeleq g \quad\text{implies}\quad
		\somewp{\cc}{f} \eeleq \somewp{\cc}{g}~.
		\]
		\item\label{thm:healthiness_feasible} $\exptransT{\SOMEWPSYMBOL}{\cc}$ is \emph{feasible}, i.e., 
		\[
		\text{for all constants $b \in \PosReals$ and all $f \in \expecs$ bounded}\,\footnote{We say that $f$ is bounded by $b$ if $f(\sigma) \leq b$ for all $\sigma \in\States$.}\text{ by $b$}\colon \somewp{\cc}{f} \eleq b~.
		\]
	\end{enumerate}
\end{theorem}
%
\begin{lemma}
	\label{lem:vc_loops}
	We have that $\vccond$ \emph{yields upper bounds for $\SOMEWPSYMBOL\in\{\DWPSYMBOL,\AWPSYMBOL\}$}, if
	\[
	\text{for all $\cc = \WHILEDOINV{\guard}{\cc'}{\inv}$ and all $f \in \expecs$}\colon\quad
	\vc{\cc}{f} ~\text{implies}~ \exptrans{\SOMEWPSYMBOL}{\cc}{f} \eleq \inv~.
	\]
	Analogously, we have that $\vccond$ \emph{yields lower bounds for $\SOMEWPSYMBOL\in\{\DWPSYMBOL,\AWPSYMBOL\}$}, if
	\[
	\text{for all $\cc = \WHILEDOINV{\guard}{\cc'}{\inv}$ and all $f \in \expecs$}\colon\quad
	\vc{\cc}{f} ~\text{implies}~ \exptrans{\SOMEWPSYMBOL}{\cc}{f} \egeq \inv~.
	\]
\end{lemma}
\begin{proof}
	By induction on $\cc$ using \Cref{thm:healthiness}.\ref{thm:healthiness_monotonic} (monotonicity of $\exptransT{\SOMEWPSYMBOL}{\cc}$).
\end{proof}
    {\allowdisplaybreaks

\clearpage
\section{Omitted Details In Examples and Case Studies}

Throughout this section, given a loop $\WHILEDOINV{\guard}{\cc}{\inv}$ and a postexpectation $f$ (which will be clear from the context), we define  $\Phi(\inv) = \iverson{\neg\guard} \cdot f +  \iverson{\guard} \cdot \somewp{\cc}{\inv}$, where $\SOMEWPSYMBOL$ is either $\AWPSYMBOL$ or $\DWPSYMBOL$, again depending on the context.

\subsection{Monty Hall Problem}
\label{app:monty}
In this section we show our technique can be applied formally to the Monty Hall example from \Cref{fig:monty-hall-intro} on \cpageref{fig:monty-hall-intro} in \Cref{s:intro}.

Note that the expression $\otherCurtain{c_1}{c_2}$ is equivalent to $6 - (c_1 + c_2)$.
Similarly, the statement $\ASSIGN{x}{\randomOtherCurtain{c}}$ can be implemented as $\PCHOICE%
{\ASSIGN{x}{2-\floor*{\tfrac{c}{2}}}}
{\tfrac{1}{2}}
{\ASSIGN{x}{4-\ceil*{\tfrac{c}{2}}}}$.

\begin{figure}[h]
    \small
    \begin{align*}
    &\GCFST{\true}{\ASSIGN{\clever}{\true}}\\
    &\GCSEC{\true}{\ASSIGN{\clever}{\false}} \COMPSEMIC\\
    &\GCFST{\true}{\ASSIGN{\contestantCurtain}{1}}\\
    &\GCSEC{\true}{\ASSIGN{\contestantCurtain}{2}}\\
    &\GCSEC{\true}{\ASSIGN{\contestantCurtain}{3}} \COMPSEMIC\\
    &\PCHOICE%
    {\ASSIGN{\priceCurtain}{1}}
    {\tfrac{1}{3}}
    {
        \PCHOICE{\ASSIGN{\priceCurtain}{2}}%
        {\tfrac{1}{2}}
        {\ASSIGN{\priceCurtain}{3}}
    } \COMPSEMIC\\
    &\GCFST%
    {\priceCurtain\neq\contestantCurtain}
    {\ASSIGN{\alternativeCurtain}{6-(\priceCurtain+\contestantCurtain)}}\\
    &\GCSEC%
    {\priceCurtain = \contestantCurtain}
    {
        \PCHOICE%
        {\ASSIGN{\alternativeCurtain}{2-\floor*{\tfrac{\priceCurtain}{2}}}}
        {\tfrac{1}{2}}
        {\ASSIGN{\alternativeCurtain}{4-\ceil*{\tfrac{\priceCurtain}{2}}}}
    } \COMPSEMIC\\
    &\GCFST{\clever}{\ASSIGN{\contestantCurtain}{6-(\contestantCurtain + \alternativeCurtain)}}\\
    &\GCSEC{\neg\clever}{\SKIP}
    \end{align*}
    \caption{Monty Hall Problem. We assume the variable domains $\contestantCurtain,\priceCurtain,\alternativeCurtain \in \{1,2,3\}$ and $\clever \in \{\true, \false\}$.}
    \label{fig:monty-hall}
\end{figure}

\begin{figure}[h]
    \small
    \begin{align*}
    &\awpannotate{\tfrac{2}{3}}\\
    &\GCFST{\true}{
        \singlelineannotatespace
        \awpannotate{\tfrac{2}{3}}
        \singlelineannotatespace
        \ASSIGN{\clever}{\true}
    }\\
    &\GCSEC{\true}{
        \singlelineannotatespace
        \awpannotate{\tfrac{1}{3}}
        \singlelineannotatespace
        \ASSIGN{\clever}{\false}
    }\\
    &\awpannotate{\tfrac{1}{3}\iverson{\neg\clever} + \tfrac{2}{3}\iverson{\clever}}\\
    &\GCFST{\true}{
        \singlelineannotatespace
        \awpannotate{\tfrac{1}{3}\iverson{\neg\clever} + \tfrac{2}{3}\iverson{\clever}}
        \singlelineannotatespace
        \ASSIGN{\contestantCurtain}{1}
    }\\
    &\GCSEC{\true}{
        \singlelineannotatespace
        \awpannotate{\tfrac{1}{3}\iverson{\neg\clever} + \tfrac{2}{3}\iverson{\clever}}
        \singlelineannotatespace
        \ASSIGN{\contestantCurtain}{2}
    }\\
    &\GCSEC{\true}{
        \singlelineannotatespace
        \awpannotate{\tfrac{1}{3}\iverson{\neg\clever} + \tfrac{2}{3}\iverson{\clever}}
        \singlelineannotatespace
        \ASSIGN{\contestantCurtain}{3}
    }\\
    &\awpannotate{
        \tfrac{1}{3}\iverson{\neg\clever}\iverson{\contestantCurtain\in\{ 1,2,3 \}}
        + \tfrac{1}{3}\iverson{\clever}(
        \iverson{\contestantCurtain\neq 1} +
        \iverson{\contestantCurtain\neq 2} +
        \iverson{\contestantCurtain\neq 3}
        )
    }\\
    &\PCHOICE%
    {\ASSIGN{\priceCurtain}{1}}
    {\tfrac{1}{3}}
    {
        \PCHOICE{\ASSIGN{\priceCurtain}{2}}%
        {\tfrac{1}{2}}
        {\ASSIGN{\priceCurtain}{3}}
    }\\
    &\awpannotate{
        \iverson{\neg\clever}\iverson{\priceCurtain=\contestantCurtain}
        + \iverson{\clever}\iverson{\priceCurtain\neq\contestantCurtain}
    }\\
    &\GCFSTOPEN%
    {\priceCurtain\neq\contestantCurtain}\\
    &\qquad\awpannotate{
        \iverson{\neg\clever}\iverson{\priceCurtain=\contestantCurtain}
        + \iverson{\clever}
    }\\
    &\qquad\ASSIGN{\alternativeCurtain}{6-(\priceCurtain+\contestantCurtain)}\\
    &\qquad\annotate{
        \iverson{\neg\clever}\iverson{\priceCurtain=\contestantCurtain}
        + \iverson{\clever}\iverson{\priceCurtain=6-(\contestantCurtain + \alternativeCurtain)}
    }\\
    &\GCFSTCLOSE\\
    &\GCSECOPEN%
    {\priceCurtain = \contestantCurtain}\\
    &\qquad\awpannotate{
        \iverson{\neg\clever}\iverson{\priceCurtain=\contestantCurtain}
    }\\
    &\qquad\PCHOICE%
    {\ASSIGN{\alternativeCurtain}{2-\floor*{\tfrac{\priceCurtain}{2}}}}
    {\tfrac{1}{2}}
    {\ASSIGN{\alternativeCurtain}{4-\ceil*{\tfrac{\priceCurtain}{2}}}}\\
    &\qquad\annotate{
        \iverson{\neg\clever}\iverson{\priceCurtain=\contestantCurtain}
        + \iverson{\clever}\iverson{\priceCurtain=6-(\contestantCurtain + \alternativeCurtain)}
    }\\
    &\GCSECCLOSE\\
    &\awpannotate{
        \iverson{\neg\clever}\iverson{\priceCurtain=\contestantCurtain}
        + \iverson{\clever}\iverson{\priceCurtain=6-(\contestantCurtain + \alternativeCurtain)}
    }\\
    &\GCFST{\clever}{
        \singlelineannotatespace
        \awpannotate{\iverson{\priceCurtain=6-(\contestantCurtain + \alternativeCurtain)}}
        \singlelineannotatespace
        \ASSIGN{\contestantCurtain}{6-(\contestantCurtain + \alternativeCurtain)}
        \singlelineannotatespace
        \annotate{\iverson{\priceCurtain=\contestantCurtain}}
    }\\
    &\GCSEC{\neg\clever}{
        \singlelineannotatespace
        \awpannotate{\iverson{\priceCurtain=\contestantCurtain}}
        \singlelineannotatespace
        \SKIP
        \singlelineannotatespace
        \annotate{\iverson{\priceCurtain=\contestantCurtain}}
    }\\
    &\annotate{\iverson{\priceCurtain=\contestantCurtain}}
    \end{align*}
    
    \caption{The Monty Hall problem from \Cref{fig:monty-hall} with $\AWPSYMBOL$ annotations.}
    \label{fig:monty-hall-anno}
\end{figure}

\begin{figure}[b]
    \small
    \begin{align*}
    &\GCFST{\true}{\ASSIGN{\clever}{\true}}\\
    &\GCSEC{\red{\false}}{\ASSIGN{\clever}{\false}}\\
    &\GCFST{\true}{\ASSIGN{\contestantCurtain}{1}}\\
    &\GCSEC{\true}{\ASSIGN{\contestantCurtain}{2}}\\
    &\GCSEC{\true}{\ASSIGN{\contestantCurtain}{3}}\\
    &\ldots \text{ (the rest of the program is deterministic and remains unchanged)}
    \end{align*}
    
    \caption{The transformed Monty Hall program resulting from the $\AWPSYMBOL$'s computed in \Cref{fig:monty-hall-anno}. Intuitively, the \red{red} branch is disabled because $\tfrac 1 3 < \tfrac 2 3$.}
    \label{fig:monty-hall-tamed}
\end{figure}

\clearpage
\subsection{Optimal Gambling}
\label{sec:chain-verification}

In this section we provide the details omitted in \Cref{sec:case_study_chain}.
In particular, we verify the invariant $\IChain$.

\begin{figure}[h]
    \small
    \begin{align*}
    &\eqannotate{\IChain}\\
    &\phiannotate{\Phi(\IChain)}\\
    &\WHILEDOINVOPEN{c<N \wedge a=0}\\
    &\qquad\dwpannotate{(\dagger)}\\
    &\qquad\GCFSTOPEN{\true}\\
    &\qquad\qquad\dwpannotate{
        \begin{aligned}[t]
        \iverson{a=0}\cdot\Biggl( 1 - \biggl( &\iverson{p\leq q}\cdot p\cdot q^{\ceil*{\sfrac{N-(c+1)}{2}}}\\
        &\quad+{} \iverson{q\leq p^2}\cdot p^{N-(c+1) + 1}\\
        &\quad+{} \iverson{p^2<q<p}\cdot p\cdot p^{\iverson{N-(c+1)>0}\iverson{N-c \textit{ even}}} \cdot q^{\floor*{\sfrac{N-(c+1)}{2}}}
        \biggr)\Biggr) + \iverson{a=1}
        \end{aligned}
    }\\
    &\qquad\qquad\PCHOICE{\ASSIGN{c}{c+1}}{p}{\ASSIGN{a}{1}}\\
    &\qquad\qquad\dwpannotate{\IChain}\\
    &\qquad\GCFSTCLOSE\\
    &\qquad\GCSECOPEN{\true}\\
    &\qquad\qquad\dwpannotate{
        \begin{aligned}[t]
        \iverson{a=0}\cdot\Biggl( 1 - \biggl( &\iverson{p\leq q}\cdot q^{\ceil*{\sfrac{N-(c+2)}{2}}+1}\\
        &\quad+{} \iverson{q\leq p^2}\cdot q\cdot p^{N-(c+2)}\\
        &\quad+{} \iverson{p^2<q<p}\cdot p^{\iverson{N-(c+2)>0}\iverson{N-c \textit{ odd}}} \cdot q^{\floor*{\sfrac{N-(c+2)}{2}}+1}
        \biggr)\Biggr) + \iverson{a=1}
        \end{aligned}
    }\\
    &\qquad\qquad\PCHOICE{\ASSIGN{c}{c+2}}{q}{\ASSIGN{a}{1}}\\
    &\qquad\qquad\dwpannotate{\IChain}\\
    &\qquad\GCSECCLOSE\\
    &\qquad\starannotate{\IChain}\\
    &\WHILEDOINVCLOSE{\IChain}\\
    &\annotate{\iverson{a=1}}
    \end{align*}
    \caption{$\CChain$ annotated.}
    \label{fig:chain-program-anno}
\end{figure}

The expectation $(\dagger)$ in \Cref{fig:chain-program-anno} is given by:

\begin{align*}
\iverson{a=0}\cdot\min\Biggl\{
\Biggl( 1 - \biggl( &\iverson{p\leq q}\cdot p\cdot q^{\ceil*{\sfrac{N-(c+1)}{2}}}\\
&\quad+{} \iverson{q\leq p^2}\cdot p^{N-(c+1) + 1}\\
&\quad+{} \iverson{p^2<q<p}\cdot p\cdot p^{\iverson{N-(c+1)>0}\iverson{N-c \textit{ even}}} \cdot q^{\floor*{\sfrac{N-(c+1)}{2}}}
\biggr)\Biggr),\\
\Biggl( 1 - \biggl( &\iverson{p\leq q}\cdot q^{\ceil*{\sfrac{N-(c+2)}{2}}+1}\\
&\quad+{} \iverson{q\leq p^2}\cdot q\cdot p^{N-(c+2)}\\
&\quad+{} \iverson{p^2<q<p}\cdot p^{\iverson{N-(c+2)>0}\iverson{N-c \textit{ odd}}} \cdot q^{\floor*{\sfrac{N-(c+2)}{2}}+1}
\biggr)\Biggr)
\Biggr\} + \iverson{a=1}\\
= \iverson{a=0}\cdot\Biggl( 1 - \max\Biggl\{
&\iverson{p\leq q}\cdot p\cdot q^{\ceil*{\sfrac{N-(c+1)}{2}}}\\
&\quad+{} \iverson{q\leq p^2}\cdot p^{N-(c+1) + 1}\\
&\quad+{} \iverson{p^2<q<p}\cdot p\cdot p^{\iverson{N-(c+1)>0}\iverson{N-c \textit{ even}}} \cdot q^{\floor*{\sfrac{N-(c+1)}{2}}},\\
&\iverson{p\leq q}\cdot q^{\ceil*{\sfrac{N-(c+2)}{2}}+1}\\
&\quad+{} \iverson{q\leq p^2}\cdot q\cdot p^{N-(c+2)}\\
&\quad+{} \iverson{p^2<q<p}\cdot p^{\iverson{N-(c+2)>0}\iverson{N-c \textit{ odd}}} \cdot q^{\floor*{\sfrac{N-(c+2)}{2}}+1}
\Biggr\} + \iverson{a=1}
\end{align*}

Thus, $\Phi(\IChain)$ is given by:
\begin{align*}
\iverson{c<N \wedge a=0}\cdot\min\Biggl\{
\Biggl( 1 - \biggl( &\iverson{p\leq q}\cdot p\cdot q^{\ceil*{\sfrac{N-(c+1)}{2}}}\\
&\quad+{} \iverson{q\leq p^2}\cdot p^{N-(c+1) + 1}\\
&\quad+{} \iverson{p^2<q<p}\cdot p\cdot p^{\iverson{N-(c+1)>0}\iverson{N-c \textit{ even}}} \cdot q^{\floor*{\sfrac{N-(c+1)}{2}}}
\biggr)\Biggr),\\
\Biggl( 1 - \biggl( &\iverson{p\leq q}\cdot q^{\ceil*{\sfrac{N-(c+2)}{2}}+1}\\
&\quad+{} \iverson{q\leq p^2}\cdot q\cdot p^{N-(c+2)}\\
&\quad+{} \iverson{p^2<q<p}\cdot p^{\iverson{N-(c+2)>0}\iverson{N-c \textit{ odd}}} \cdot q^{\floor*{\sfrac{N-(c+2)}{2}}+1}
\biggr)\Biggr)
\Biggr\}\\
& +{} \iverson{\neg \left(c<N \wedge a=0\right)}\iverson{a=1}
\end{align*}

To simplify this expression, we make a case distinction.
We fix a state $\sigma$ and assume that
$\sigma(a)=1$, as well as $\sigma(c)\leq\sigma(N)$.
For the first case we further assume
that $\sigma(p)\leq\sigma(q)$. Then we have:
\begin{align*}
1-\max\left\{p\cdot q^{\ceil*{\sfrac{N-(c+1)}{2}}},\, q^{\ceil*{\sfrac{N-(c+2)}{2}}+1}\right\} = 1 - q^{\ceil*{\sfrac{N-c}{2}}}
\end{align*}
Next consider the case $\sigma(q)\leq\sigma(p^2)$. In this case:
\begin{align*}
&1-\max\left\{p^{N-(c+1) + 1},\, q\cdot p^{N-(c+2)}\right\}\\
&\qquad= 1-\max\left\{p^2\cdot p^{N-(c+2)},\, q\cdot p^{N-(c+2)}\right\}\\
&\qquad= 1-p^{N-c}
\end{align*}
Next consider the case $\sigma(p^2)<\sigma(q)<\sigma(p)$. We get:
\begin{align*}
&1-\max\left\{p\cdot p^{\iverson{N-(c+1)>0}\iverson{N-c \textit{ even}}} \cdot q^{\floor*{\sfrac{N-(c+1)}{2}}},\
p^{\iverson{N-(c+2)>0}\iverson{N-c \textit{ odd}}} \cdot q^{\floor*{\sfrac{N-(c+2)}{2}}+1} \right\}\\
&\qquad\overset{\star}{=} 1 - p^{\iverson{N-c>0}\iverson{N-c \textit{ odd}}} \cdot q^{\floor*{\sfrac{N-c}{2}}}
\end{align*}
To show that $\star$ holds, we again consider multiple cases. Consider the case
$\sigma(N)=\sigma(c)+1$:
\tb{Now we should probably use $\sigma(p)$ and so on, or possibly make a comment about it.}
\begin{align*}
&1-\max\left\{p,\ q\right\} = 1-p
\end{align*}
Consider now the case $\sigma(N)=\sigma(c)+2$:
\begin{align*}
&1-\max\left\{p^2,\ q\right\} = 1-q
\end{align*}
Otherwise we have:
\begin{align*}
&1-\max\left\{p\cdot p^{\iverson{N-c \textit{ even}}} \cdot q^{\floor*{\sfrac{N-(c+1)}{2}}},\
p^{\iverson{N-c \textit{ odd}}} \cdot q^{\floor*{\sfrac{N-c}{2}}} \right\}\\
&\qquad= 1 - p^{\iverson{N-c \textit{ odd}}} \cdot q^{\floor*{\sfrac{N-c}{2}}}
\end{align*}
In summary:\tb{Fix alignment}
\begin{align*}
\Phi(\IChain) =
\iverson{c<N \wedge a=0}\cdot\Biggl( 1 - \biggl( &\iverson{p\leq q}\cdot q^{\ceil*{\sfrac{N-c}{2}}}\\
&\quad+{} \iverson{q\leq p^2}\cdot p^{N-c}\\
&\quad+{} \iverson{p^2<q<p}\cdot p^{\iverson{N-c>0}\iverson{N-c \textit{ odd}}} \cdot q^{\floor*{\sfrac{N-c}{2}}}
\biggr)\Biggr)\\
& +{} \iverson{\neg \left(c<N \wedge a=0\right)}\iverson{a=1}\\
= \iverson{a=0}\cdot\Biggl( 1 - \biggl( &\iverson{p\leq q}\cdot q^{\ceil*{\sfrac{N-c}{2}}}\\
&\quad+{} \iverson{q\leq p^2}\cdot p^{N-c}\\
&\quad+{} \iverson{p^2<q<p}\cdot p^{\iverson{N-c>0}\iverson{N-c \textit{ odd}}} \cdot q^{\floor*{\sfrac{N-c}{2}}}
\biggr)\Biggr) + \iverson{a=1}\\
=\IChain
\end{align*}

\clearpage
\subsection{Optimal Gambling 2}
\label{sec:geo-verification}

In this section we provide the details omitted in \Cref{sec:case_study_geo}.
In particular, we show that $\IGeo$ is a superinvariant of the loop in $\CGeo$ (\Cref{fig:geo-program} on \cpageref{fig:geo-program}).
This is done in \Cref{fig:geo-program-anno} using our annotation style.

\begin{figure}[h]
    \small
    \begin{align*}
    &\succeqannotate{
        \iverson{a=1}\cdot c + \iverson{a= 0}\cdot\max\,
        \left\{c + \tfrac{2q}{1-q},\, c+\tfrac{p}{1-p}\right\}
    }\\
    &\eqannotate{
        \iverson{a=1}c + \iverson{a= 0} \cdot
        \left( c +
            \max\, \left\{\tfrac{p(q + 1)}{1-q},\, \tfrac{p}{1-p}, \tfrac{2q}{1-q},\, \tfrac{q(2-p)}{1-p} \right\},
        \right)
    }\\
    &\phiannotate{
        \iverson{a=1}c + \iverson{a= 0}\cdot \max\, \left\{
        \begin{aligned}
            &\max\,
            \left\{\tfrac{p(q + 1)}{1-q},\, \tfrac{p}{1-p}\right\} + c\\
            &\max\,
            \left\{\tfrac{2q}{1-q},\, \tfrac{q(2-p)}{1-p}\right\} + c
        \end{aligned}
        \right\}
    }\\
    &\WHILEDOINVOPEN{a= 0}\\
    &\qquad\awpannotate{
        \iverson{a= 0}\cdot \max\, \left\{
        \begin{aligned}
            &\max\,
            \left\{\tfrac{p(q + 1)}{1-q},\, \tfrac{p}{1-p}\right\} + c\\
            &\max\,
            \left\{\tfrac{2q}{1-q},\, \tfrac{q(2-p)}{1-p}\right\} + c
        \end{aligned}
        \right\}
    }\\
    &\qquad\GCFSTOPEN{\true}\\
    &\qquad\qquad\eqannotate{
    \iverson{a= 0}\cdot\max\,
    \left\{\tfrac{p(q + 1)}{1-q},\, \tfrac{p}{1-p}\right\} + c
    }\\
    &\qquad\qquad\eqannotate{
        \iverson{a= 0}\cdot\max\,
        \left\{\tfrac{p2q}{1-q},\, \tfrac{p^2}{1-p}\right\} + p + c
    }\\
    &\qquad\qquad\awpannotate{
        \iverson{a=1}\cdot p(c+1) + \iverson{a= 0}\cdot p \cdot\max\,
        \left\{c+1 + \tfrac{2q}{1-q},\, c+1+\tfrac{p}{1-p}\right\} + (1-p)\cdot c
    }\\
    &\qquad\qquad\PCHOICE{\ASSIGN{c}{c+1}}{p}{\ASSIGN{a}{1}}\\
    &\qquad\qquad\awpannotate{
        \iverson{a=1}\cdot c + \iverson{a= 0}\cdot\max\,
        \left\{c + \tfrac{2q}{1-q},\, c+\tfrac{p}{1-p}\right\}
    }\\
    &\qquad\GCFSTCLOSE\\
    &\qquad\GCSECOPEN{\true}\\
    &\qquad\qquad\eqannotate{
        \iverson{a= 0}\cdot\max\,
        \left\{\tfrac{2q}{1-q},\, \tfrac{q(2-p)}{1-p}\right\} + c
    }\\
    &\qquad\qquad\eqannotate{
        \iverson{a= 0}\cdot\max\,
        \left\{\tfrac{2q^2}{1-q},\, \tfrac{pq}{1-p}\right\} + 2q + c
    }\\
    &\qquad\qquad\awpannotate{
        \iverson{a=1}\cdot q(c+2) + \iverson{a= 0}\cdot q \cdot\max\,
        \left\{c+2 + \tfrac{2q}{1-q},\, c+2+\tfrac{p}{1-p}\right\} + (1-q)\cdot c
    }\\
    &\qquad\qquad\PCHOICE{\ASSIGN{c}{c+2}}{q}{\ASSIGN{a}{1}}\\
    &\qquad\qquad\awpannotate{
        \iverson{a=1}\cdot c + \iverson{a= 0}\cdot\max\,
        \left\{c + \tfrac{2q}{1-q},\, c+\tfrac{p}{1-p}\right\}
    }\\
    &\qquad\GCSECCLOSE\\
    &\qquad\starannotate{
        \iverson{a=1}\cdot c + \iverson{a= 0}\cdot\max\,
        \left\{c + \tfrac{2q}{1-q},\, c+\tfrac{p}{1-p}\right\}
    }\\
    &\WHILEDOINVCLOSE{\IGeo}\\
    &\annotate{c}
    \end{align*}
    \caption{
        Program $\CGeo$ annotated. We assume variable domains $p,q \in (0,1) \cap \Rats$, $a \in \{0,1\}$, and $c \in \Nats$.
    }
    \label{fig:geo-program-anno}
\end{figure}

Regarding the program transformation given in \Cref{fig:geo-program-tamed} on \cpageref{fig:geo-program-tamed}, it suffices to compare the topmost $\AWPSYMBOL$'s in the two if-branches, i.e.,
\[
    \iverson{a= 0}\cdot\max\,
    \left\{\tfrac{p(q + 1)}{1-q},\, \tfrac{p}{1-p}\right\} + c
    \qquad\text{vs}\qquad
    \iverson{a= 0}\cdot\max\,
    \left\{\tfrac{2q}{1-q},\, \tfrac{q(2-p)}{1-p}\right\} + c
\]

It can be shown\footnote{\url{https://www.wolframalpha.com/input?i=max\%7B+p\%281\%2Bq\%29+\%2F+\%281-q\%29\%2C+p+\%2F+\%281-p\%29+\%7D+\%3E\%3D+max+\%7B2q+\%2F+\%281-q\%29\%2C+q\%282-p\%29+\%2F+\%281-p\%29\%7D\%2C+p+\%3E+0\%2C+q\%3E+0+\%2C+p+\%3C+1\%2C+q+\%3C+1}}
that the left expectation is $\egeq$ the right expectation if and only if the guard in the first branch of the guarded command statement in \Cref{fig:geo-program-tamed} holds, and similarly for direction $\eleq$ and the second branch.

\clearpage
\subsection{Game of Nim}
\label{sec:nim-verification}

In this section, we provide the details omitted in \Cref{sec:nim-case-study} (also see \Cref{s:intro}).

\subsubsection{Verifying the Invariant}

We show that $\IGame\eleq\Phi(\IGame)$. $\Phi(\IGame)$ is given by:
\begin{align*}
&\iverson{x\geq N}\iverson{\PTwoTurn}\\
&\qquad+{} \iverson{x<N}\cdot\left(
\begin{aligned}
&\iverson{\POneTurn}\left(
\begin{aligned}
&\sfrac{1}{3}\left(\iverson{x+1\geq N} + \iverson{x+1<N}\iverson{x+2\ncongmod{4} N}\right)\\
&\quad +{}\sfrac{1}{3}\left(\iverson{x+2\geq N} + \iverson{x+2<N}\iverson{x+3\ncongmod{4} N}\right)\\
&\quad +{}\sfrac{1}{3}\left(\iverson{x+3\geq N} + \iverson{x+3<N}\iverson{x+4\ncongmod{4} N}\right)
\end{aligned}\right)\\
&\quad+{}\iverson{\PTwoTurn}\max\left\{
\begin{aligned}
&\iverson{x+1<N}\left(
\begin{aligned}
&\iverson{x+2\congmod{4} N}\\
&\quad +{}\sfrac{2}{3}\iverson{x+2\ncongmod{4} N}
\end{aligned}\right),\\
&\iverson{x+2<N}\left(
\begin{aligned}
&\iverson{x+3\congmod{4} N}\\
&\quad +{}\sfrac{2}{3}\iverson{x+3\ncongmod{4} N}
\end{aligned}\right),\\
&\iverson{x+3<N}\left(
\begin{aligned}
&\iverson{x\congmod{4} N}\\
&\quad +{}\sfrac{2}{3}\iverson{x\ncongmod{4} N}
\end{aligned}\right)\\
\end{aligned}\right\}\\
\end{aligned}
\right)
\end{align*}

We make a case distinction by fixing some $\pstate\in\States$.
If $\pstate(x)\geq\pstate(N)$, then
$\IGame(\pstate) = \Phi(\IGame)(\pstate)$. If $\pstate(x)<\pstate(N)$, we distinguish further. In the
following we assume that $\pstate(N)\notin\{\pstate(x)+1,\pstate(x)+2,\pstate(x)+3\}$, since otherwise
it is easy to see that $\IGame(\pstate) \leq \Phi(\IGame)(\pstate)$.

For the first case we assume that $\pstate(x)+1 \congmod{4} \pstate(N)$. Then:
\begin{align*}
\Phi(\IGame)(\pstate)
&= \iverson{\POneTurn}(\pstate) + \sfrac{2}{3}\cdot \iverson{\PTwoTurn}(\pstate)\\
&\geq \iverson{\POneTurn}(\pstate)\\
&= \IGame(\pstate)
\end{align*}
If otherwise $\pstate(x)+1 \ncongmod{4} \pstate(N)$, then:
\begin{align*}
\Phi(\IGame)(\pstate)
&= \sfrac{2}{3}\cdot \iverson{\POneTurn}(\pstate) + \iverson{\PTwoTurn}(\pstate)\\
&= \IGame(\pstate)
\end{align*}
We conclude that $\IGame(\pstate) \leq \Phi(\IGame)(\pstate)$ for all $\pstate\in\States$ and thus
$\IGame\eleq\Phi(\IGame)$.

\subsubsection{Obtaining the Transformed Program}

We first annotate the program in \Cref{fig:game-program-annotated} (not all $\AWPSYMBOL$'s are necessary to read off predicates needed for the strategy).

\begin{figure}[h]
    \small
    \begin{align*}
    &\WHILEDOINVOPEN{x<N}\\
    &\qquad\GCFST{\POneTurn}{\UNIFASSIGN{x}{x{+}1,\, x{+}2,\, x{+}3}}\\
    &\qquad \awpannotate{\ldots \text{(not relevant for constructing the strategy)}} \\
    &\qquad\GCSECOPEN{\varTurn = 2}\\
    &\qquad\begin{aligned}[t]
    &\qquad\GCFST{\true}{\singlelinepreannotatespace \awpannotate{\expsubs{\IGame}{\varTurn,x}{3 - \varTurn,x+1}} \singlelineannotatespace \ASSIGN{x}{x+1} \singlelineannotatespace \annotate{\expsubs{\IGame}{\varTurn}{3 - \varTurn}} \singlelineannotatespace }\\
    &\qquad\GCSEC{\true}{\singlelinepreannotatespace \awpannotate{\expsubs{\IGame}{\varTurn,x}{3 - \varTurn,x+2}} \singlelineannotatespace \ASSIGN{x}{x+2} \singlelineannotatespace \annotate{\expsubs{\IGame}{\varTurn}{3 - \varTurn}} \singlelineannotatespace }\\
    &\qquad\GCSEC{\true}{\singlelinepreannotatespace \awpannotate{\expsubs{\IGame}{\varTurn,x}{3 - \varTurn,x+3}} \singlelineannotatespace \ASSIGN{x}{x+3} \singlelineannotatespace \annotate{\expsubs{\IGame}{\varTurn}{3 - \varTurn}} \singlelineannotatespace }\\
    &\qquad\GCSECCLOSE \\
    \end{aligned}\\
    &\qquad\awpannotate{\expsubs{\IGame}{\varTurn}{3 - \varTurn}}\\
    &\qquad\ASSIGN{\varTurn}{3 - \varTurn} \\
    &\qquad\starannotate{\IGame}\\
    &\WHILEDOINVCLOSE{\IGame}
    \end{align*}
    \caption{Game example $\CNim$ annotated.}
    \label{fig:game-program-annotated}
\end{figure}

We have to compare the topmost $\AWPSYMBOL$'s inside the nondeterministic branches.
They are as follows:

\begin{align*}
I_1 \ddefeq & \annocolor{\expsubs{\IGame}{\varTurn,x}{3 - \varTurn,x+1}} \\
\eeq & \iverson{x+1<N} \cdot 
\Big(
\iverson{\PTwoTurn} \cdot
\big(
\iverson{x+2 \congmod{4} N} + \tfrac{2}{3} \cdot \iverson{x+2 \ncongmod{4} N}
\big)
~+~
\iverson{\POneTurn} \cdot \iverson{x+2 \ncongmod{4} N}
\Big)\\
& \qquad ~+~ \iverson{x+1 \geq  N} \cdot \iverson{\POneTurn}
\end{align*}

\begin{align*}
I_2 \ddefeq & \annocolor{\expsubs{\IGame}{\varTurn,x}{3 - \varTurn,x+2}} \\
\eeq & \iverson{x+2<N} \cdot 
\Big(
\iverson{\PTwoTurn} \cdot
\big(
\iverson{x+3 \congmod{4} N} + \tfrac{2}{3} \cdot \iverson{x+3 \ncongmod{4} N}
\big)
~+~
\iverson{\POneTurn} \cdot \iverson{x+3 \ncongmod{4} N}
\Big)\\
& \qquad ~+~ \iverson{x+2 \geq  N} \cdot \iverson{\POneTurn}
\end{align*}

\begin{align*}
I_3 \ddefeq & \annocolor{\expsubs{\IGame}{\varTurn,x}{3 - \varTurn,x+3}} \\
\eeq & \iverson{x+3<N} \cdot 
\Big(
\iverson{\PTwoTurn} \cdot
\big(
\iverson{x \congmod{4} N} + \tfrac{2}{3} \cdot \iverson{x \ncongmod{4} N}
\big)
~+~
\iverson{\POneTurn} \cdot \iverson{x \ncongmod{4} N}
\Big)\\
& \qquad ~+~ \iverson{x+3 \geq  N} \cdot \iverson{\POneTurn}
\end{align*}

We make a case distinction:
\begin{itemize}
    \item Suppose that $\pstate \models (x+1 \congmod{4} N)$.
    Then
    \begin{align*}
    I_1 &\eeq \tfrac{2}{3} [x+1 < N]\cdot[\varTurn = 2] + [\varTurn = 1] \\
    I_2 &\eeq \tfrac{2}{3} [x+2 < N]\cdot[\varTurn = 2] + [\varTurn = 1] \\
    I_2 &\eeq \tfrac{2}{3} [x+3 < N]\cdot[\varTurn = 2] + [\varTurn = 1] \\
    \end{align*}
    If moreover $\pstate \models (x+1 \geq N)$, then $I_1(\pstate) = I_2(\pstate) = I_3(\pstate)$.
    Otherwise $\pstate \models (x + 1 < N)$, and thus, since we assumed that $\pstate \models (x+1 \congmod{4} N)$, we have $\pstate \models (x + 1 + 4 \leq N)$, i.e., $\pstate \models (x + 4 < N)$.
    But then again $I_1(\pstate) = I_2(\pstate) = I_3(\pstate)$.
    In summary, if the current program state satisfies $\pstate \models (x+1 \congmod{4} N)$, then it does not matter by how much we increment $x$, which is reflected in the strategy program in \Cref{fig:game-program-tamed}.
    \item Next suppose that $\pstate \models (x+2 \congmod{4} N)$.
    Then
    \begin{align*}
    I_1 &\eeq [x+1 < N]\cdot[\varTurn = 2] + [x+1 \geq N]\cdot [\varTurn = 1] \\
    I_2 &\eeq \tfrac{2}{3} [x+2 < N]\cdot[\varTurn = 2] + [\varTurn = 1] \\
    I_2 &\eeq \tfrac{2}{3} [x+3 < N]\cdot[\varTurn = 2] + [\varTurn = 1] \\
    \end{align*}
    If $\pstate \models (turn = 2)$, then clearly $I_1(\pstate) \geq I_2(\pstate) = I_3(\pstate)$.
    If $\pstate \models (turn = 1)$, then $I_1(\pstate) \leq I_2(\pstate) = I_3(\pstate) = 1$.
    However, the latter situation does not actually happen in our program because the branch we are considering is guarded by the predicate $turn =2$.
    We therefore increment $x$ by $1$ if $\pstate \models (x+2 \congmod{4} N)$.
    \item The argument for the cases $\pstate \models (x+3 \congmod{4} N)$ and $\pstate \models (x \congmod{4} N)$ is analogous.
\end{itemize}

} 
    \clearpage
\section{Automated Loop Invariant Verification using \textsc{Caesar} }
\label{app:heyvl}

In the following, we provide the input to the deductive verifier \textsc{Caesar} \cite{heyvl} for automatically verifying the loop invariant from \Cref{ex:dwp_park}.
\lstset{escapeinside={<@}{@>}}
\begin{lstlisting}
	coproc simple_loop() -> (c: UInt, x: UInt)
	pre !?(true)
	post x
	{
		var choice: Bool;
		
		covalidate;
		coassume <@\textcolor{blue}{[c!=0]*(x+1) + [c==0]*x}@>;
		if(c != 0){
			choice = flip((1/2));
			if(choice){
				c = 0;
			}else{
				if \cap {
					x = x*x
				}else{
					x = x+1
				}
			}
			
			coassert <@\textcolor{blue}{[c!=0]*(x+1) + [c==0]*x}@>;
			coassume !?(false)
		}else{
			
		}
		
		
	}
\end{lstlisting}
}{}

\end{document}